\documentclass[12pt]{article}
\usepackage{authblk}
\usepackage{geometry}
\usepackage{amsmath}
\usepackage{amssymb}
\usepackage{amsthm}
\usepackage{mathrsfs}
\usepackage{customcommands}
\usepackage{bm}
\usepackage{Nettastyle}
\usepackage{dsfont}
\usepackage{comment}
\usepackage[utf8]{inputenc}
\usepackage{calc}
\usepackage{accents}
\usepackage{amsfonts}
\newtheorem{assump}{Assumption}

\usepackage{braket}
\usepackage[normalem]{ulem}
\usepackage{color}

\usepackage{ulem}
\usepackage{mathtools}
\usepackage{tensor} 
\usepackage{tikz}
\usetikzlibrary{arrows, positioning, quotes, intersections}
\usepackage{subcaption}
\usepackage{graphicx}  

\usepackage{thmtools, thm-restate}

\newtheorem{conjecture}{Conjecture}

\newcommand{\cyan}[1]{\iffalse\textcolor{cyan}{#1}\fi}

\newcommand{\mK}{\mathcal{K}}

\newcommand{\mH}{\mathcal{H}}

\DeclareMathOperator{\tr}{tr}

\title{Cryptographic Censorship}

\author[1]{Netta Engelhardt,}
\author[1]{{\AA}smund Folkestad,}
\author[1]{Adam Levine,}
\author[1,2]{Evita Verheijden,}
\author[1,3]{\\and Lisa Yang \vspace{-0.1cm}}

\affiliation[1]{Center for Theoretical Physics, Massachusetts Institute of Technology, \\Cambridge, MA 02139, USA}
\affiliation[2]{Black Hole Initiative at Harvard University, \\ 20 Garden Street, Cambridge, MA 02138, USA}
\affiliation[3]{Computer Science and Artificial Intelligence Laboratory,\\
Massachusetts Institute of Technology, Cambridge, MA 02139, USA}
\emailAdd{engeln@mit.edu}
\emailAdd{afolkest@mit.edu}
\emailAdd{arlevine@mit.edu}
\emailAdd{evitamhv@mit.edu}
\emailAdd{lisayang@mit.edu}

\abstract{We formulate  and take two large strides towards proving a quantum version of the weak cosmic censorship conjecture. We first prove ``Cryptographic Censorship'': a theorem showing that when the time evolution operator of a holographic CFT is approximately pseudorandom (or Haar random) on some code subspace, then there must be an event horizon in the corresponding bulk dual. This result provides a general condition that guarantees (in finite time) event horizon formation, with minimal assumptions about the global spacetime structure. Our theorem relies on an extension of a recent quantum learning no-go theorem and is proved using new techniques of pseudorandom measure concentration. To apply this result to cosmic censorship, we separate singularities into classical, semi-Planckian, and Planckian types. We illustrate that classical and semi-Planckian singularities are compatible with approximately pseudorandom CFT time evolution; thus, if such singularities are indeed approximately pseudorandom, by Cryptographic Censorship, they cannot exist in the absence of event horizons. This result provides a sufficient condition guaranteeing that seminal holographic results on quantum chaos and thermalization, whose general applicability relies on typicality of horizons, will not be invalidated by the formation of naked singularities in AdS/CFT.}

\begin{document}


\newcommand{\beq}{\begin{equation}} 
\newcommand{\eeq}{\end{equation}}
\renewcommand{\[}{\left[}
\renewcommand{\]}{\right]}
\renewcommand{\(}{\left(}
\renewcommand{\)}{\right)}
\renewcommand{\gets}{\leftarrow} 

\newcommand{\Nat}{\mathbb{N}} 
\newcommand{\poly}{\mathrm{poly}}
\newcommand{\Hn}{n} 
\newcommand{\Hd}{d} 
\newcommand{\secp}{\kappa} 
\newcommand{\negl}{\mathrm{negl}} 
\newcommand{\avg}{\mathrm{avg}}
\newcommand{\avgover}[1]{\underset{#1}{\mathrm{avg}}} 
\newcommand{\const}{\alpha} 
\newcommand{\POVM}{E} 
\newcommand{\operator}{Q} 
\newcommand{\swap}{\mathsf{S}} 
\newcommand{\SWAP}{\mathrm{SWAP}} 
\newcommand{\Dist}{\mathsf{Dist}} 
\newcommand{\HCP}{\mathrm{HCP}} 
\newcommand{\deltaMC}{\exp{\left(- \epsilon^2 \unidim/48\right)}}
\newcommand{\deltaMChigh}{e^{- \epsilon^2 \unidim/48}}
\newcommand{\deltaMCerror}{e^{-(\te - \hat{\epsilon})^2 \unidim/48}}
\newcommand{\deltaMCp}{e^{- (\te + \hat{\epsilon})^2 \unidim/48}}
\newcommand{\subdist}{k}
\newcommand{\deltaMCU}{e^{-\epsilon^2 d/48(1+K)^2}}
\newcommand{\deltaMCS}{e^{-\epsilon^2 d/192}}

\newcommand{\Hil}{\mathcal{H}} 
\newcommand{\code}{\mathrm{code}} 
\newcommand{\unidim}{d}
\newcommand{\key}{k} 
\newcommand{\PRU}{U} 
\newcommand{\PRUens}{\mathcal{U}} 
\newcommand{\PRS}{\psi} 
\newcommand{\PRSens}{\Psi} 
\newcommand{\ketPRS}{\ket{\psi}} 
\newcommand{\Alg}{\mathcal{A}} 
\newcommand{\avgf}[1]{{\langle F \rangle^{#1}}} 
\newcommand{\te}{\tilde{\epsilon}}

\newcommand{\HN}{N} 
\newcommand{\CWop}{\mathcal{O}} 
\newcommand{\HKLLop}{\mathcal{O}} 
\maketitle

\section{Introduction}\label{sec:intro} 
Curvature singularities are ubiquitous in General Relativity: the evolution of smooth initial data generically results in such singularities~\cite{Pen64,Haw66,HawPen70}. Observationally, however, evidence for singularity formation from generic gravitational processes (such as matter collapse) remains at best indirect. This apparent absence of directly-observable singularities is often explained by the general expectation that the formation of horizons is likewise generic: so singularities are usually cloaked behind event horizons. This hypothesis --- that singularities resulting from initial data evolution should be generically hidden behind horizons --- is known as the Weak Cosmic Censorship Conjecture (WCCC)~\cite{Pen65}. It is a cornerstone assumption of almost all results in modern gravitational physics, from classical theorems such as the Hawking area law~\cite{Haw71}, to key results from the general ``it from qubit'' program. In particular, the expectation that CFT thermalization corresponds to black hole equilibration, and more generally that black hole statistics are reflected in CFT thermal states, crucially depends on the assumption that states with horizons are typical in some measure.\footnote{This is not identical to the firewall sense of typicality of~\cite{MarPol13}. See footnote~\ref{fn:typical} on page \pageref{fn:typical}.} If the WCCC were false, naked singularities could be as typical as black holes; our understanding of quantum chaos, eigenstate thermalization, and black hole information in holography would be called into question.

Surprisingly, it turns out that naked singularities do exist in solutions to the Einstein equation satisfying various energy conditions and regular initial data, and these states do \textit{not} constitute a measure zero set.  There is by now a large repertoire of counterexamples to the WCCC: in more than four spacetime dimensions, the studies of Gregory--Laflamme type instabilities of~\cite{GreLaf94} have illustrated the existence of ``pinch-off'' singularities as the end-stage of certain black holes (e.g., the black string~\cite{LehPre10}) ---  see~\cite{Gre11} for a review; more recently,~\cite{HorSan16, CriSan16, CriHor17, HorSan19} constructed extended naked singularities in asymptotically AdS$_{4}$ spacetimes with various matter profiles; see also~\cite{Fol22}. A similar mechanism was used by~\cite{EpeBog20} to construct a likely counterexample in the four-dimensional asymptotically flat setting. 

It thus appears that horizons and singularities are not related in any trivial or obvious way: curvature singularities may form outside (or in the absence) of any horizons, and horizons may form without singularities behind them~\cite{Bar68}. There is no extant formulation of WCCC that is (1) not violated and (2) sufficient to guarantee the consistency of seminal results in AdS/CFT such as the universality of the saturation of the chaos bound and the bulk manifestation of the CFT thermal partition function. 

Some progress towards understanding cosmic censorship in quantum gravity was made in ~\cite{HorSan16, CriSan16,CriHor17}, where it was noted that many (though not quite all) of the counterexamples also violate the Weak Gravity Conjecture~\cite{ArkMot06}, suggesting the possibility that the WCCC in four dimensions is valid in the landscape but violated in the swampland. Further evidence in favor of this conclusion was presented in~\cite{Fol22}, which constructed large sets of AdS$_{4}$ initial data satisfying the gravitational constraints and the null energy condition whose future development must include naked singularities, but these theories violated properties required of theories admitting a UV completion via holography~\cite{EngHor19}. This hypothesis would suggest that while classical General Relativity may violate the WCCC, the semiclassical limit of quantum gravity does not. 

However, this na\"ive ``swampland censorship'' scenario is clearly incomplete. In higher dimensions, the Gregory--Laflamme instability is not excluded by any known swampland constraints; moreover the inclusion of quantum corrections is expected to lead to black hole evaporation, which also violates the WCCC. However, the absence of naked singularities in the observable universe together with the top down arguments from quantum gravity suggest that \textit{some} version of WCCC remains valid: intuitively, a commonality between the Gregory--Laflamme instability and the evaporating black hole is that there is a heuristic sense, which in this article we will make precise, that the naked part of the singularities is ``small''. See Fig.~\ref{fig:endpoint-evap}.
\begin{figure}
    \centering
    \includegraphics[scale=0.75]{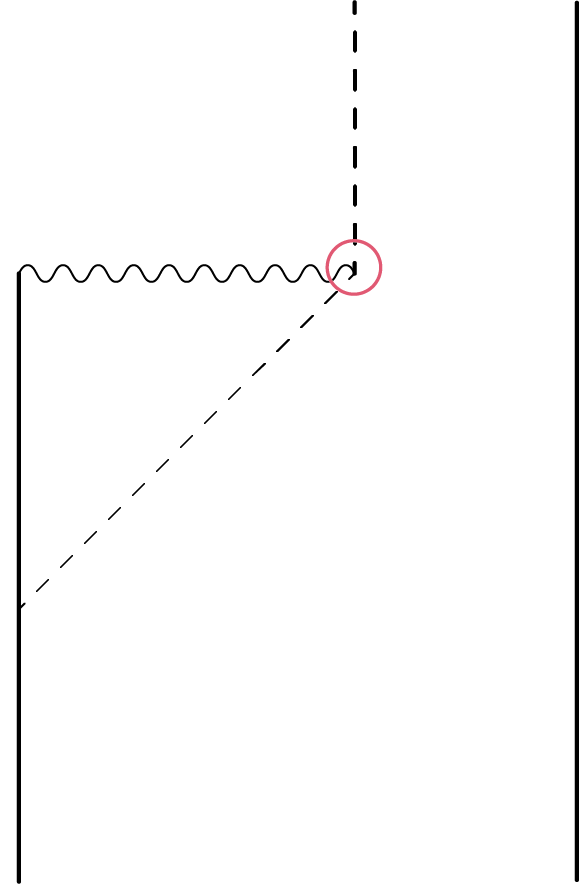}
    \caption{The evaporation point of a black hole is, in some sense that we make precise in Sec.~\ref{sec:size-of-sings}, ``small''.}
    \label{fig:endpoint-evap}
\end{figure}
That is, naked singularities can exist, but they must only be visible to asymptotic observers for a ``short" period of time. Furthermore, they can only exist in spacetimes that have horizons --- so that the bulk of the singularity can remain hidden.

A formulation of such a ``quantum gravitational'' version of cosmic censorship -- which we propose later in this article -- and its subsequent proof would require as \emph{sine qua non} both (1) a condition that can guarantee the formation of horizons in broad generality, and (2) a definition of a ``small'' naked singularity. 

Point (1) is a notoriously difficult open problem. Event horizons are by definition teleological: an event horizon is defined with reference to future infinity. Consequently, very few extant results can guarantee the existence of such a global construct even under very restrictive assumptions (and none guarantee that a given event horizon will successfully hide the full extent of a curvature singularity). See~\cite{ Tho72, ShaTeu91, ChoPre09} for a series of failed attempts to construct criteria to this effect. Addressing point (1) would also have the potential to shed light on a number of other research programs, such as thermalization in AdS/CFT (see e.g. \cite{CheWil15} for a review) and the firewall problem~\cite{AMPS, AMPSS, MarPol13}. Point (2) has seen various approaches (see~\cite{Emp20}), but none of those have successfully made contact with the generic existence of horizons, since (1) has thus far remained unresolved.

Here we accomplish both (1) and (2) in the context of AdS/CFT: we give a general guarantee for the existence of event horizons at finite time, and criteria for the size of a singularity that make contact with the guarantee for event horizons. In particular, we classify singularities into three categories: classical, semi-Planckian and Planck-sized, and we motivate an assumption on the behavior of singularities that would guarantee the existence of horizons at finite time --- via point (1) --- in spacetimes containing such classical or semi-Planckian singularities. In the rest of this section, we give a rough outline of the central ingredients of (1) and (2).

\subsubsection*{A General Guarantee for Horizon Formation}

Under the standard assumption of global hyperbolicity,\footnote{We really mean AdS hyperbolicity here, which is essentially global hyperbolicity when the asymptotic boundary is included. See~\cite{Wal12} for the precise definition.}  which guarantees various consistency results such as unitary invariance of the von Neumann entropy and entanglement wedge nesting, there are various sufficient conditions that may guarantee existence of an event horizon, such as exponential reconstruction complexity or chaos. However, global hyperbolicity \textit{implies} the WCCC as a consequence: almost all results ensuring consistency of the AdS/CFT dictionary assume cosmic censorship --- even though it is known that cosmic censorship can be violated in perfectly reasonable AdS/CFT setups (e.g., the evaporating AdS black hole). This strongly suggests that while the WCCC is sufficient for consistency, it is not in fact necessary. We will argue for a more minimalistic form of censorship, which \textit{is} necessary;\footnote{We hope that this sentence is not taken out of context.} it remains to be seen whether this version is also sufficient.

Our goal is to find a cosmic censorship conjecture that follows directly from consistency of the AdS/CFT dictionary. Clearly, an essential part of a holographic formulation of cosmic censorship is the holographic dual of an event horizon: we are looking for the boundary dual to the process of bulk horizon formation. This specifically means that we are interested in the existence of a region outside of the causal wedge of the entire asymptotic boundary. The separation between the causal (or simple~\cite{EngWal17a, EngWal18, EngPen21b}) wedge and the entanglement wedge has been previously linked in~\cite{EngPen21a, EngPen21b, AkeEng22} to the separation between polynomial and exponential quantum complexity classes in the dual theory (see also \cite{BouFef19} and footnote~\ref{fn:BFV}). Our approach to finding a guarantee for the formation of event horizons will thus focus on reconstruction restrictions by complexity class.

More concretely, we ask when there is a guarantee for the existence of some operator $Q$ in the large-$N$ limit that cannot be reconstructed from knowledge of the causal wedge alone, given that the latter can be reconstructed efficiently, e.g., by jumping into the bulk, making a measurement, and sailing back towards the boundary.\footnote{By an efficient reconstruction, we mean one that can be implemented by a quantum circuit whose size is polynomially bounded in the logarithm of the dimension of the code subspace under consideration.} 

In~\cite{YanEng23}, two of us provided such a guarantee at large but finite $N$ under certain assumptions. The full statement of the relevant theorem will be reviewed at length in Sec.~\ref{sec:PR-and-learning-result}; here we provide a heuristic statement only, starting with the most significant assumption: that the action of the fundamental time evolution operator $U$ of the system is well modeled on our code subspace by a \textit{pseudorandom} unitary~\cite{JiLiu18} (we use ``fundamental'' to refer to the full quantum description of the gravitational system, as in~\cite{AkeEng22}). In AdS/CFT, this assumption is valid whenever the time evolution of the CFT on a particular choice of code subspace is well approximated by a pseudorandom unitary at any finite value of $N$;\footnote{We will eventually work at strictly infinite $N$, but then we will work only with time-evolved states rather than the unitary operator that generates that time evolution, since the aforementioned unitary operator is not well-defined in the strict infinite $N$ limit.} this is not dissimilar from the assumption of e.g.~\cite{HayPre07} that the fundamental time evolution of a black hole is well modeled by a two-design or as in~\cite{AkeEng22, KimPre22} by a (pseudo)random unitary. Pseudorandom unitaries are operators drawn from an ensemble $\{U_{k}\}$ that are indistinguishable from Haar random unitaries by any efficient quantum algorithm. Our results below also hold for Haar random unitaries, but the advantage of pseudorandom unitaries over Haar random unitaries is that by definition they can be constructed by an efficient process. While existence of pseudorandom unitaries has not been explicitly demonstrated yet, candidate constructions have appeared in, e.g., \cite{JiLiu18}. We will not make use of any such explicit constructions.  

Under this first assumption of pseudorandom unitary (PRU) dynamics, the theorem of~\cite{YanEng23} considers any efficient quantum algorithm ${\cal A}$ that:
\begin{enumerate}
    \item[(a)] Prepares polynomially complex states $\ket{\varphi}$ and has access to the action of the time evolution operator $U$ on the expectation values of polynomially complex observables in such states; 
    \item[(b)] Uses the information in (a) to attempt to learn $U$ (or a circuit representation of $U$), the pseudorandom time evolution operator, thereby producing some operator $\widehat{U}$; 
    \item[(c)] After completing (b), is handed the actual state of the system $\ket{\psi}$ and is asked to predict $U$ acting on $\ket{\psi}$ using the outcome of (b).
\end{enumerate}

The theorem states that on average  --- and as we will show in Sec.~\ref{sec:complearnholog}, for typical pseudorandom states and unitaries\footnote{\label{fn:typical}Caution: this is not the typical use of `typical'. Here we will use the qualifier typical specifically for states in the code subspace to which measure concentration applies; this depends on the choice of distribution of states under consideration. For example, we may consider ``typical Haar random states'' or ``typical pseudorandom states''. We will always note the distribution under discussion. This nomenclature is not to be confused with the previously used `generic', which refers to an open subset of the space of solutions to the Einstein equation.} ---  any efficient algorithm will fail to produce an accurate guess for $U$ acting on $\ket{\psi}$.\footnote{Note that \emph{prima facie}, the fact that we extend Theorem 2 of \cite{YanEng23} to any fixed, typical $U$ and $\psi$ appears to contradict some of the statements in the introduction of~\cite{YanEng23}, i.e., that randomness is a \textit{sine qua non} for a quantum learning no-go theorem. This appears to be the case because if $\PRU=\PRU_0$ is fixed, there exists a trivial $\Alg^U$ that just exactly outputs $\PRU_0$. The point is that $U_0$ is here `atypical', and the algorithm will fail for all other $U\neq U_0$.} That is, an efficient algorithm's best guess $\widehat{U}$ for $U$ will satisfy:
\be
| \bra{\psi}\widehat{U}^{\dagger}U\ket{\psi}|^2 \leq1-\alpha~, 
\ee
where $\alpha$ is an $O(1)$ constant, independent of the size of the Hilbert space in which $\ket{\psi}$ lives. This result in particular guarantees the existence of a distinguishing operator $Q$ whose expectation value differs  between $\widehat{U}\ket{\psi}$ and $U\ket{\psi}$ by at least $2\alpha$.  

We now apply this result to our gravitational setup using the fact that reconstruction is a particular type of learning. We will consider code subspaces ${\cal H}_{\rm code}$ of dimension $e^{O(G_{N}^{a})}$, where $-1\leq a<0$, which do not necessarily have horizons in typical states. For example, a microcanonical window of width $O(N)$; another example is a set of states where the leading order contribution to the geometry is the same for all of the states at least for a time of $O(G_{N}^{0})$. We will show that it is easy to prepare candidates for such code subspaces without obviously incurring horizons. We next assume that the time evolution is \textit{gravitationally} pseudorandom: we will define this term more precisely in Sec.~\ref{sec:complearnholog}, but roughly speaking it denotes an $O(G_{N}^{a})$ amount of pseudorandomness; we find that this amount of pseudorandomness is generally compatible with an ${\cal H}_{\rm code}$ of size $e^{O(G_{N}^{a})}$, which are the code subspaces that we consider. Next, we argue for the existence of an efficient algorithm that, as part of the learning phase, can reconstruct the causal wedge of any efficiently-preparable state in this code subspace. Under certain technical assumptions, there is an algorithm that, given access to some CFT data, is able to output a guess $\widehat{U}$ for the time evolution of the initial state that accurately reproduces the expectation values of all operators in the causal wedge. However, the theorem of~\cite{YanEng23} guarantees that the fidelity between the predicted time evolution and the actual time evolved state is low, which implies that there exists a distinguishing operator $Q$ whose expectation values differ at $O(1)$ between the actual time evolved state and the predicted time evolution of the state. 

Since there exists a learning algorithm that can accurately reconstruct the time evolution of all operators in the causal wedge but cannot accurately reconstruct the expectation value of some operator $Q$, which we show exists in the large-$N$ limit, we immediately find that the dual to $Q$ lies outside of the causal wedge. Thus, the boundary of the causal wedge is nontrivial and an event horizon must exist.

We call this result \textit{Cryptographic Censorship}: 

\begin{center}
    \textit{(Pseudo)random dynamics guarantee event horizon formation in typical (pseudo)random states.} 
\end{center}

This is the first general condition that guarantees the formation of horizons without requiring asymptotic time evolution, thus addressing the first task above: the prescription of a broadly applicable quantum complexity theoretic diagnostic of event horizon formation.\footnote{Here we are defining an event horizon as $\partial J^{-}[\mathscr{I}]$, so this includes holographic cosmologies in AdS.}  This result can be thought of as a quantum instability of generic asymptotically AdS spacetimes, assuming that typical pseudorandom dynamics are a good model for (sufficiently long) time evolution of generic spacetimes. We may apply this result to our two examples of code subspaces: microcanonical windows of width $O(N)$ and states with an identical leading order geometry at some fixed time. In the case of the former, we find that pseudorandom (respectively, random) states in a given microcanonical window are  exponentially (respectively, doubly exponentially) likely to have horizons. In the case of the latter, we find that whenever it is possible to squeeze a gravitational amount of pseudorandomness into an $O(1)$ time band on the boundary, an event horizon must form.  

Before we address the second task mentioned above --- classifying the size of a singularity --- let us motivate why pseudorandom time evolution is natural in strongly gravitating systems. There is a general expectation that some amount of (approximate) randomness is crucial for understanding black hole dynamics: due to scrambling, many aspects of black holes appear to be well modeled by Haar random unitary dynamics or sometimes by the weaker notion of $k$-designs~\cite{HayPre07}. The chaotic nature of black hole physics~\cite{SheSta14,RobSta14, MalShe15} and the connection between randomness and chaos~\cite{CotGur16} have strengthened this expectation. However, we also expect that time evolution in quantum gravity can be efficiently implemented. Haar random unitaries and (exact) $k$-designs do not have this property. We thus find pseudorandom unitaries to be a well-motivated model of black hole dynamics. On the boundary side of the AdS/CFT duality, many holographic CFTs are known to exhibit chaotic dynamics in typical (high energy) states. Here, we will show that sufficiently pseudorandom unitary (boundary) time evolution in fact \textit{guarantees} the existence of an event horizon in typical states (and thus, by definition, of a black hole).\footnote{\label{fn:BFV}At this point it is worth noting that \cite{BouFef19} also discussed pseudorandomness in the context of reconstruction complexity: they showed that an easily prepared set of black hole microstates --- specifically black holes formed by collapse together with the application of some randomly chosen shocks --- should form a pseudorandom ensemble, and that as a consequence, boundary reconstruction of the volume of the black hole interior (for a code subspace containing all these states) must be exponentially complex. 
Our aim here is notably different: crucially, we do not start with an ensemble of black hole microstates, but instead show that gravitational pseudorandom boundary dynamics \textit{implies} the existence of an event horizon. }

Thus far we have argued that pseudorandomness is a valid model for time evolution of strongly gravitating systems using established facts about black holes. We shall postpone motivating pseudorandomness as a model for time evolution in the context of naked singularities until we have defined what we mean by ``large'' and  ``small'' naked singularities. We therefore proceed to give an overview of this second task. 

\subsubsection*{Size of a Naked Singularity}

Intuitively, the naked singularities that we would like quantum gravity to exclude are classically extended with divergent curvatures, and emitting large amounts of observable, causally propagating radiation as prescribed by some fundamental quantum gravity evolution. 

\begin{figure}
    \centering
    \includegraphics[scale=0.75]{{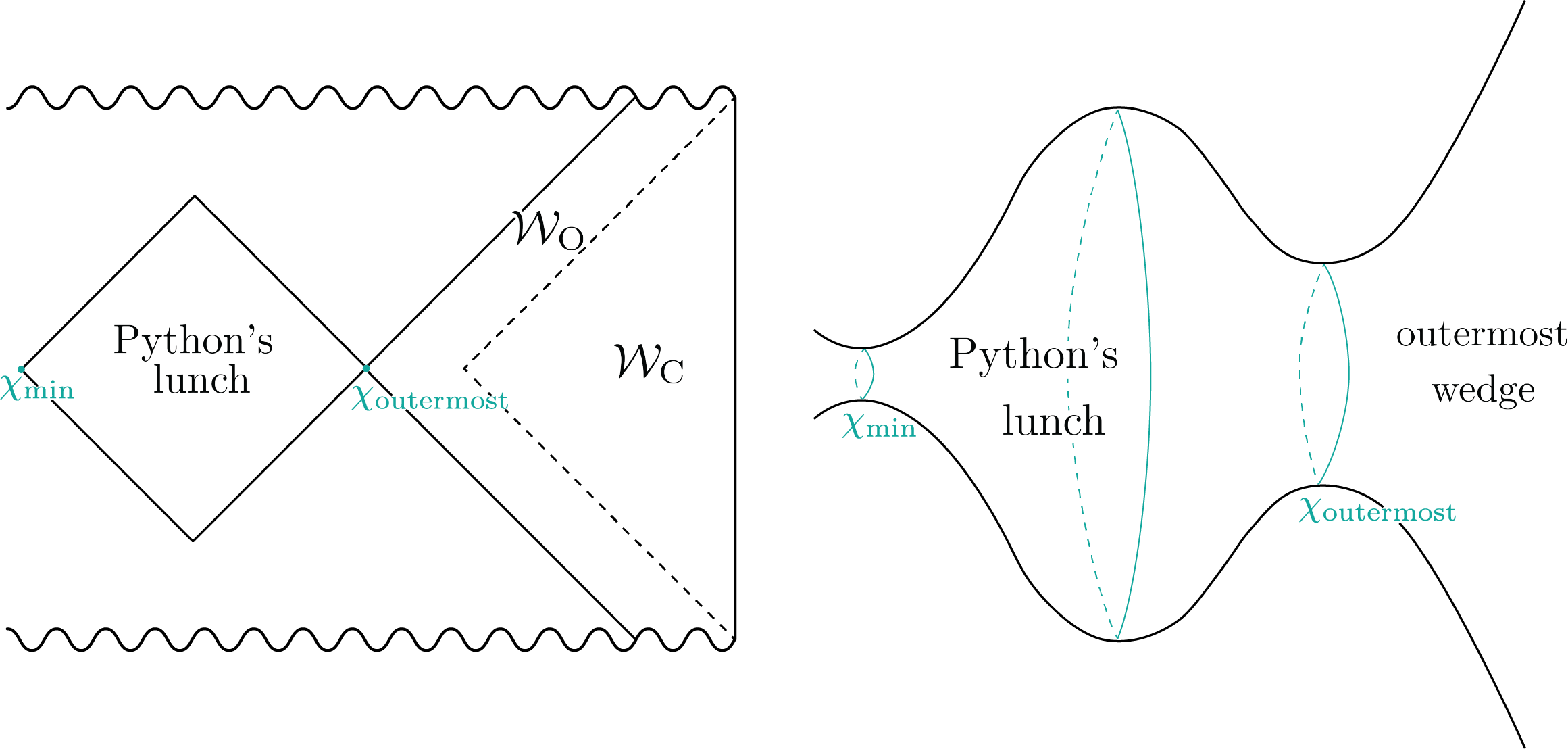}}
    \caption{A spacetime with a Python's Lunch. In the Penrose diagram on the left, the lunch, causal wedge ($\mathcal{W}_{\rm C}$) and the wedge of the outermost QES ($\mathcal{W}_{\rm O}$) are indicated. The left and right green dots are the minimal and outermost QES, respectively. The right panel depicts a timeslice.}
    \label{fig:PL}
\end{figure}
In many cases, such large curvatures create (quantum) trapped surfaces and (quantum) extremal surfaces. In~\cite{EngFol20}, unitary invariance of $S_{\rm vN}$ in the dual CFT was used to prove that classical extremal (and trapped) surfaces resulting from classical singularities that do not ``evaporate'' must lie behind horizons in AdS/CFT. As a warmup, we will prove that even upon inclusion of quantum corrections and under no restrictions on the behavior of the singularity, quantum extremal surfaces (QESs) generically lie behind horizons. The argument makes essential use of the geometrization of regions of high- and low-reconstruction complexity: the causal wedge, and more generally the wedge of the outermost QES, admits an efficient reconstruction --- i.e., subexponentially complex in $\log \dim {\cal H}_{\rm code}$, where ${\cal H}_{\rm code}$ is as before the code subspace on which reconstruction is taking place~\cite{EngPen21a}, while the complementary wedge --- the so-called Python's Lunch --- does not~\cite{BroGha19}. See Fig.~\ref{fig:PL}. In particular, acting on the Python's Lunch is an exponentially difficult (in $\log \dim {\cal H}_{\rm code}$) task in the boundary theory. If any QES lies in causal contact of the boundary, it is possible to act on the lunch by throwing in some timelike probe, as illustrated in Fig.~\ref{fig:QES-observer}. This is a simple task, and thus cannot happen: QESs --- minimal or otherwise --- must therefore lie outside of the causal wedge\footnote{For those readers who have previously seen the Quantum Focusing Conjecture (QFC) used to justify why QESs lie behind horizons~\cite{EngWal14,BouFis15a}, note that all such proofs to date have assumed cosmic censorship (in the form of global hyperbolicity). We prove this result using the holographic dictionary instead of using the QFC.} of the relevant region.\footnote{Any readers confused about the islands outside the horizon phenomenon~\cite{AlmMah19b,BouPen23} should recall that the horizons of~\cite{AlmMah19b,BouPen23} are not the horizons of the system for which the island is quantum extremal but are rather the horizons of a smaller system.}  

\begin{figure}
    \centering
    \includegraphics[scale=0.75]{{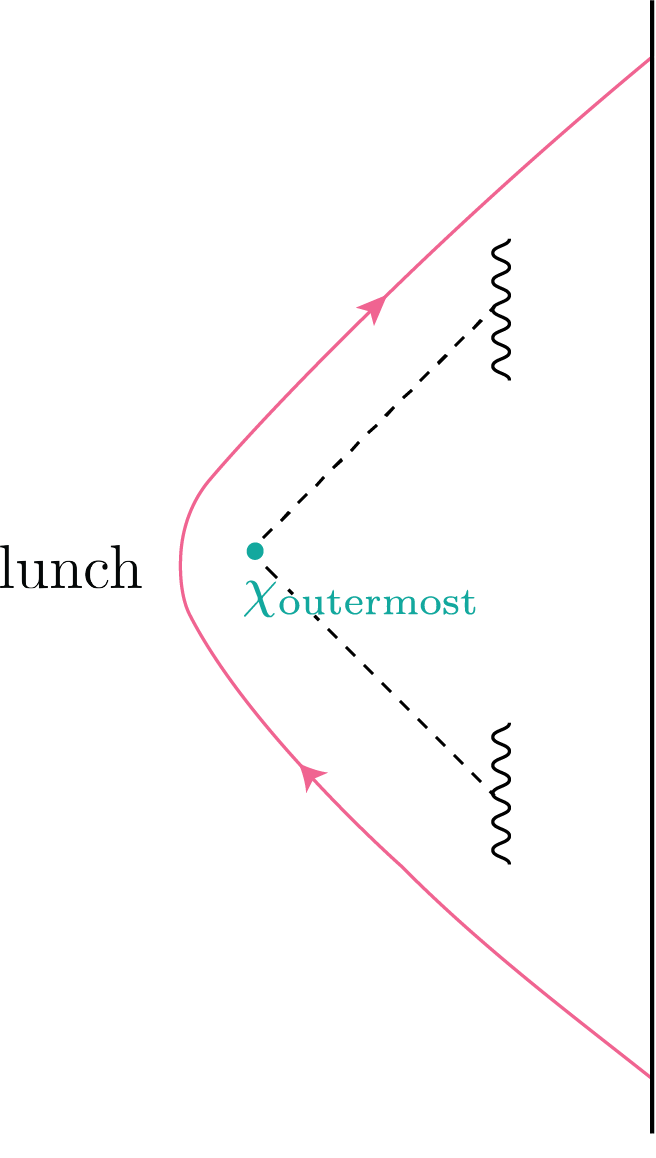}}
    \caption{A forbidden situation: a QES in causal contact with the boundary.}
    \label{fig:QES-observer}
\end{figure}

However, singularities are not always associated with (quantum) trapped  or (quantum) extremal surfaces. In these cases, the aforementioned result showing that extremal surfaces imply horizons will not be useful. However, if a singularity could be shown to be associated to a gravitational amount of pseudorandomness, then Cryptographic Censorship ensures that these singularities induce horizons. But how can we know which singularities are ``severe enough'' to accommodate gravitational pseudorandomness? As a first step towards this question, we give a precise definition of the size of a naked singularity $\Gamma$, denoted  $|\Gamma|$. This quantity measures the minimal amount of CFT time evolution required to prescribe (or equivalently, obtain) all information about the singularity. For a timelike singularity, we furthermore show that this quantity also measures how many modes within the semiclassical EFT can interact with the singularity. Of course, the dynamics of the singularity itself would not be semiclassical, but we argue that in the $\ell_{\rm Pl}\rightarrow 0$ limit, from the perspective of weakly backreacting probe fields, a timelike singularity might reasonably be conjectured to act primarily as a complicated UV-sensitive boundary condition for their evolution. Then we show that if $|\Gamma|$ does not go to zero as $\ell_{\rm Pl}\rightarrow 0$, or if it goes to zero sufficiently slowly, then the singularity is compatible with inducing gravitational pseudorandomness in the dynamics. With this in mind, we classify singularities into three groups: classical, semi-Planckian, and Planckian. We speculate that classical naked singularities, and possibly semi-Planckian ones, are forbidden by Cryptographic Censorship. Compatibility with gravitational pseudorandomness is of course necessary but not sufficient to guarantee actual gravitational pseudorandomness in the dynamics. We of course cannot give a top-down derivation that the unitary dynamics of a holographic CFT are approximately pseudorandom on such code subspaces, but we shall motivate this as a reasonable assumption for generic classical timelike singularities. Cryptographic Censorship, together with this expectation, is the basis of our subsequent conjecture of Quantum Cosmic Censorship: that classical (and perhaps semi-Planckian) singularities must be hidden behind horizons; only Planckian singularities may be visible to the asymptotic observer.

Let us now motivate the assumption of pseudorandomness as a model for singularity dynamics (which we have previously motivated for black holes; see also~\cite{KimPre22, YanEng23}). A natural source of pseudorandom dynamics associated to a singularity is due to the Planckian physics at the singularity itself --- what we described above as a UV-sensitive boundary condition. We expect such a boundary condition cannot be predicted by effective field theory. We will argue in more detail in Sec.~\ref{sec:counting} that this UV-sensitive boundary condition at the singularity could plausibly be the effect of pseudorandomness in the fundamental dynamics. 
We expect that in many cases, pseudorandomness is an accurate model for the boundary dynamics describing a singularity. When this criterion is satisfied, the theorem of~\cite{YanEng23} together with the argument outlined above guarantees that the spacetime has an event horizon.\footnote{Note that our argument does not guarantee that the singularity lies behind the event horizon, nor should it: we expect that small singularities can exist outside of horizons. Our techniques do suggest that it may be possible to limit the extent to which a singularity can lie outside of a horizon, but we leave a detailed investigation of this question to future work.} It immediately follows that \textit{whenever the fundamental dynamics of a singularity are well approximated by a pseudorandom unitary, the singularity incurs an event horizon.} 

It is worth mentioning another potential source of pseudorandom dynamics associated to a singularity: spacetime dynamics in the vicinity of the singularity might induce chaos in the probe fields. The paradigmatic example of chaos in classical gravity --- other than event horizons via the butterfly effect ---  is the BKL phenomenon~\cite{BelKha70a}. This phenomenon, which involves chaotically oscillating metric components, is expected to be generic in the descent to a spacelike singularity~\cite{BelKha70b}. BKL behavior is however \textit{only} established for spacelike singularities; it is possible that certain timelike singularities could induce similar chaos upon approach (see, e.g.,~\cite{ShaWan16}), though we remain agnostic on the plausibility of this potential source of pseudorandomness. 

We emphasize here that throughout this paper we will be considering pseudorandomness in the \textit{boundary} dynamics. Readers interested in understanding what happens in the situation of bulk pseudorandomness should refer to Sec.~\ref{sec:discussion}.

We now briefly summarize the paper. We begin in the next subsection by setting up notations and conventions. In Sec.~\ref{sec:QESsHorizons}, we prove that nonminimal QESs must lie behind a horizon. In Sec.~\ref{sec:cryptoCensor} we prove that when the time evolution in the fundamental description is sufficiently pseudorandom, the bulk spacetime contains an event horizon. We then apply this to the microcanonical window and a code subspace representing states with the same (leading order) geometry for times of $O(G_N^0)$. In Sec.~\ref{sec:singularities-pseudorandom}, we provide a definition that quantifies the size of (the naked part of) a singularity, and we show that classical (and semi-Planckian) singularities can accommodate gravitational pseudorandomness. By Cryptographic Censorship, this would then imply that such singularities must incur horizons. In Sec.~\ref{sec:singularities-pseudorandom}, we also speculate on how Cryptographic Censorship can be used to prove a version of cosmic censorship. We end with a conjecture for a new, inherently quantum gravitational form of cosmic censorship that may be consistent with known counterexamples to the standard WCCC, including the evaporating black hole. In Sec.~\ref{sec:discussion}, we end with a short discussion of our aspirations and a few possible future directions. In Apps.~\ref{app:proof-hol-learning} and \ref{app:proof-large-N}, we provide proof details for the statements of Sec.~\ref{sec:cryptoCensor}.

\subsection{Preliminaries}\label{sec:notation} 

\paragraph{Assumptions.} Unlike much of the AdS/CFT literature, we will not be imposing the assumption of global hyperbolicity on the bulk spacetimes under consideration, as this is tantamount to assuming weak (and strong) cosmic censorship.  
This means that, even in the $G_N\rightarrow 0$ limit, QFT in curved spacetime (including linearized gravitons) cannot be utilized for time evolution. However, we will still assume that QFT in curved spacetime is valid within causal diamonds that contain no singularities. We will further assume that any spacetime under consideration admits an initial data slice $\Sigma$ with an everywhere defined and nowhere vanishing continuous timelike unit normal; in particular, we require that $D[\Sigma]$ contains an open set containing $\Sigma$. Any singularity in the spacetime will be assumed to act as a UV-sensitive boundary condition that determines the relation between the algebras of different causal diamonds.\footnote{The singularity only has an effect on the relation between two diamonds $D_1$ and $D_2$ when the set $[J^+(D_1)\cup J^-(D_1)] \cap [J^+(D_2)\cup J^-(D_2)]$ contains part of the singularity. If this set does not contain a part of the singularity, the EFT can be used to relate the algebra of observables in $D_1$ and $D_2$.} Finally, we will take a conformal frame which is taken to be fixed at any value of $N$, including in the large-$N$ limit.

\paragraph{Assumptions on the semiclassical limit.}
We will also often implicitly use the framework of \emph{asymptotically isometric codes} for discussing the notion of ``geometric" states in the boundary theory. This framework was inspired by the work of \cite{FauLi22, LeuLiu22}. By ``a geometric state," we really have in mind that our bulk state is described by a sequence of boundary states, $\ket{\psi_N}\in \Hil_N$, labeled by an index $N$, which controls $1/G_N$ in the bulk. The Hilbert space $\Hil_N$ is that of the boundary theory at finite $N$. In the case where the boundary theory is $\mathcal{N}=4$ SYM, we have in mind that $N$ is the rank of the gauge group, but more generally it is just some parameter in the boundary theory.

We now state what we mean by an asymptotically isometric code for our work. These assumptions are only a slight modification of those listed in \cite{EngFol24b}.
\begin{assump}\label{ass:codespace}
    The bulk Hilbert space is given by a direct sum $\mathcal{H} = \mathcal{H}_{g_1} \oplus \mathcal{H}_{g_2} \oplus ...$, where $\mathcal{H}_{g_i}$ is the Hilbert space of QFT states defined perturbatively about a fixed background manifold $M_{g_i}$. Furthermore, the algebra of bounded operators on this Hilbert space, $\mathcal{B}(\mathcal{H})$, breaks up into a direct sum over bounded operators on each Hilbert space individually, $\mathcal{B}(\mathcal{H}) = \bigoplus_i \mathcal{B}(\mathcal{H}_{g_i})$.
\end{assump}

\begin{assump}\label{ass:isometry} 
There exists a sequence of bounded linear maps $V_N: \mathcal{H} \to \Hil_N$ such that for all $\ket{\psi}, \ket{\phi} \in \mathcal{H}$
\begin{align}
    \lim_{N \to \infty} \braket{\psi|V_N^{\dagger} V_N|\phi} = \braket{\psi|\phi}\,,
\end{align}
where $\Hil_N$ is the full boundary Hilbert space. From the results of \cite{FauLi22}, we will just need that for every bulk state $\rho$ on $\mathcal{B}(\mathcal{H})$ and every operator $\mathcal{O} \in \mathcal{B}(\mathcal{H})$, there is a sequence of states $\rho_N$ and boundary operators $\mathcal{O}_N \in \mathcal{B}(\Hil_N)$ such that
\begin{align}
    & \lim_{N \to \infty} \rho_N(\mathcal{O}_N) = \rho(\mathcal{O}),
\end{align}
and where $\rho_N$ and $\mathcal{O}_N$ are related by conjugation by $V_N$ to their bulk duals, $\rho = V_N^{\dagger} \rho_N V_N$ and $\mathcal{O}_N = V_N \mathcal{O} V_N^{\dagger}$. 
\end{assump}

Given these assumptions, we can more carefully  define what we mean by a boundary ``geometric" state in the context of an asymptotically isometric code. Namely, we say that a sequence of states $\lbrace \psi_N: \psi_N \in \Hil_N\rbrace$ is \emph{geometric} if there exists $\ket{\psi} \in \mathcal{H}$ such that $|| V_N \ket{\psi} - \ket{\psi_N}||_1 = 1/\exp(N)$. In such a situation, we will sometimes say that the states $\ket{\psi_N}$ labeled by $N$ are ``represented" by the asymptotically isometric code. 

In what follows, we will be concerned with whether or not bulk observables lie in the causal wedge. We will denote the algebra associated to the causal wedge as $\mathcal{N}_{\mathrm{cw}}$. This algebra is a subalgebra of $\mathcal{B}(\mathcal{H}) \supseteq \mathcal{N}_{\mathrm{cw}}$. Whether this inclusion is strict or not is the essential point of this work. If $\mathcal{B}(\mathcal{H}) \supsetneq \mathcal{N}_{\mathrm{cw}}$, then there exists a horizon by definition. 

In what follows, we will make various assumptions about the time evolution of the boundary theory. To conform with the standard expectations of AdS/CFT, we will assume that whenever we have an asymptotically isometric code, then bulk time evolution is boundary time evolution on the code subspace. We can encapsulate this in the following assumption. 
\begin{assump}\label{ass:bulkbndytimeevol} 
Let $H_{\rm bulk}$ be the bulk Hamiltonian acting on $\mathcal{H}$. Furthermore, let $H_N$ be the boundary Hamiltonian at finite $N$ acting on $\Hil_N$. Then
\begin{align}
    V_N^{\dagger} e^{-iH_N t} V_N = e^{-i H_{\rm bulk}t}.
\end{align}
\end{assump}

\paragraph{Notation and conventions.} We now proceed to introduce and review some language and notation from the quantum computing literature, which will be important for the formulation of Cryptographic Censorship. The rest of this section contains notation that will mostly be used in Sec.~\ref{sec:cryptoCensor}. We will use $\mathcal{H}$ to denote Hilbert spaces and for a pure state $\ket{\psi}$, we will use $\psi$ to denote the density matrix $\ket{\psi}\bra{\psi}$. When we write $x \gets \mathcal{X}$, we mean a variable $x$ is drawn from a distribution $\mathcal{X}$. For example, by $U \leftarrow \mu$, we mean to draw a unitary $U$ (from the unitary group $\mathbf{U}(\mathcal{H})$) using the Haar measure $\mu$; we will also use $\mu$ to denote the distribution of Haar random states, e.g., when we write $\ket{\psi}\gets\mu$. We use $\underset{x \gets \mathcal{X}}{\mathrm{Pr}}$ to denote taking the probability of some outcome occurring for $x$ drawn from $\mathcal{X}$. We will often work with a parameter $k$, which in the cryptographic literature is called a ``key''. This should be thought of as some small amount of random classical information that can be used to encode or decode data. Concretely, it is a bit string that is sampled uniformly from a \emph{key space} $\mK$, a set of bit strings. The dimension of $\mathcal{K}$ scales with a parameter $\kappa$, which is frequently referred to in the literature as the ``security parameter". We will always have in mind that we are scaling $\log \dim \mathcal{H} = \text{poly}(\secp)$ and will eventually take the $\secp \to \infty$ limit. A function $f(x)$ is said to be \emph{negligible} in its parameter $x$ if for any constant $c>0$, $|f(x)| \leq x^{-c}$ for any sufficiently large $x$. We will say that a quantity $f(x)$ is $O(g(x))$, denoted by $f(x)\sim O(g(x))$, if for some positive constants $c_1,c_2$ and for any sufficiently large $x$, $c_1 g(x)\leq |f(x)|\leq c_2 g(x)$.\footnote{In computer science literature this is denoted by $f = \Theta(g(x))$.}  Finally, we will call an operator $V$ \textit{approximately norm-preserving} if it approximately preserves state norms, i.e., for any state $\ket{\psi}$,  $\bra{\psi}V^{\dagger}V\ket{\psi}$ is exponentially close to $1$ for sufficiently large $\secp$. This of course includes any unitary operators. 

Our results will make heavy use of quantum algorithms, which we now define. A quantum algorithm ${\cal A}$ is a map from an input Hilbert space ${\cal H}_{\rm in}$ to some output set ${\cal C}$. Note that the set ${\cal C}$ can be a Hilbert space, a classical space, or some other set altogether. For example, a quantum algorithm could take as input a quantum state $\ket{\psi}\in {\cal H}_{\rm in}$ and output a classical bit representation of an operator ${\cal O}$ (e.g., the coefficients of the representation of ${\cal O}$ in some fixed choice of basis). A quantum algorithm ${\cal A}$ acting on some $ {\cal H}_{\rm in}$ is said to be \textit{efficient}, \textit{polynomially complex}, or equivalently \textit{computationally bounded} if it can be implemented by a quantum circuit of size no greater than polynomial in $\log \dim {\cal H}_{\rm in}$.\footnote{This involves some canonical choice of gates; however, this choice does not change the scaling in $\log \dim {\cal H}_{\rm in}$.} Here we will typically take ${\cal H}_{\rm in}$ to be the code subspace under consideration. We will often drop the reference to $\log \dim {\cal H}_{\rm in}$ when describing complexity or efficiency of such algorithms.

Consider now an efficient algorithm designed to take in a state $\ket{\psi}$ and output $0$ if it determines that $\ket{\psi}$ was sampled from distribution $\mu_1$ and $1$ if it determines $\ket{\psi}$ came from $\mu_2$. In this case, we say the two ensembles are computationally indistinguishable --- which we define precisely for general algorithms below --- if any efficient algorithm succeeds at this task with probability upper bounded by $1/2$ up to negligible corrections. It should be clear that this is equivalent to saying that the quantum algorithm outputs $1$ (or equivalently $0$) with almost the same probability when acting on both $\mu_{1}$ and $\mu_{2}$. 

Note that the authors of \cite{JiLiu18}, whose definitions of pseudorandom unitaries and states we will use, demand that pseudorandom ensembles are indistinguishable from the Haar ensemble even if the algorithm is allowed a reasonable (polynomial) number of copies of the state. Having multiple copies of the state is essential to allow the algorithm to probe coherent information about the quantum state.

With this intuition, we present the exact definition of computational indistinguishability of ensembles of states for general quantum algorithms:
\begin{defn}[Computational Indistinguishability]\label{defn:compindist}
Two ensembles of states, $\mathcal{S}_1$ and $\mathcal{S}_2$, in some Hilbert space $\mathcal{H}$, with size parameterized by a  parameter $\secp$, are \emph{computationally indistinguishable} if for any efficient quantum algorithm $\mathcal{A}$ and for all $m = \emph{poly}(\secp)$, 
\begin{align} 
    \left| \underset{\ket{\psi_1} \leftarrow \mathcal{S}_1}{\emph{Pr}} \left(\mathcal{A}(\ket{\psi_1}^{\otimes m}) =1 \right)-  \underset{\ket{\psi_2} \leftarrow \mathcal{S}_2}{\emph{Pr}} \left(\mathcal{A}(\ket{\psi_2}^{\otimes m}) =1 \right)\right| = \emph{negl}(\secp)~.
\end{align}
\end{defn} 

When we write $\Alg=1$, we are considering that the algorithm is given some binary task (e.g., deciding whether a state is pseudorandom or not), and outputs either $0$ or $1$; we write $\Alg=1$ to indicate that it outputs $1$.
 
The specific quantum algorithms that we will consider are learning algorithms; these are efficient algorithms that take certain data as input, implement some processing with the goal of learning some outcome, and finally produce their best guess for the outcome. In our particular context, the learning algorithms will aim to learn a circuit (or algebraic) representation of a unitary operator $U$. During the initial learning phase, the algorithms are given ``oracle access'' to some unknown unitary $U$, denoted by $\Alg^U$: this means that the algorithms have the ability to query $U$ a polynomial number of times on a polynomial number of (efficiently preparable) states $\ket{\varphi_i} \in \Hil$ to obtain $U\ket{\varphi_{i}}$ during its learning phase. During this phase, the algorithms can perform any polynomially complex computation between queries, on both the $U\ket{\varphi_{i}}$ from its queries thus far, and on its own registers. At the end of this phase, the algorithms' oracle access to $U$ is terminated, and they are tasked with producing a maximally accurate representation of $U$ or a particular aspect or feature of it: e.g., $U$ acting on some initial state $\ket{\psi_{\rm initial}}$ which would now be given to the algorithm. 

Let us briefly make a note on the different parameters and dimensions used throughout this paper. In the context of AdS/CFT, we will be working with code subspaces of dimension $\dim \Hil_{\rm code} = 2^\kappa$, where the security parameter $\kappa = \poly (N)$. For concreteness, we will consider $G_N \sim 1/N^2$, but the relation between $G_N$ and $N$ can be changed to the reader's liking. We will often use $N$ instead of $\kappa$ to parameterize functions, e.g., for negligible functions we will write $\eta(N)$.

Note that theorems and proofs in this paper are well above the physics level of rigor but fall short of a mathematician's level of rigor in a few places. We identify these gaps where possible and leave a fully rigorous formulation thereof to the curious mathematician.

\section{Warmup: QESs Lie Behind Horizons} \label{sec:QESsHorizons}
We will show in this section that if $\chi$ is a nonempty QES homologous to $\mathscr{I}$ in a spacetime $(M,g)$, then it lies behind both past and future event horizons. The quantum focusing conjecture (QFC)~\cite{BouZac16} with reflecting boundary conditions at $\mathscr{I}$ guarantees by the quantum singularity theorem~\cite{Wal10} that at least one null geodesic fired from $\chi$ is geodesically incomplete.\footnote{Strictly speaking, the restricted QFC \cite{Sha22} is sufficient to guarantee this.} However, in the absence of a stringent causality assumption such as AdS hyperbolicity, null geodesic incompleteness does not guarantee the existence of horizons. See Fig.~\ref{fig:excised-points} for an illustration. 
\begin{figure}
    \centering
    \includegraphics[scale=0.9]{{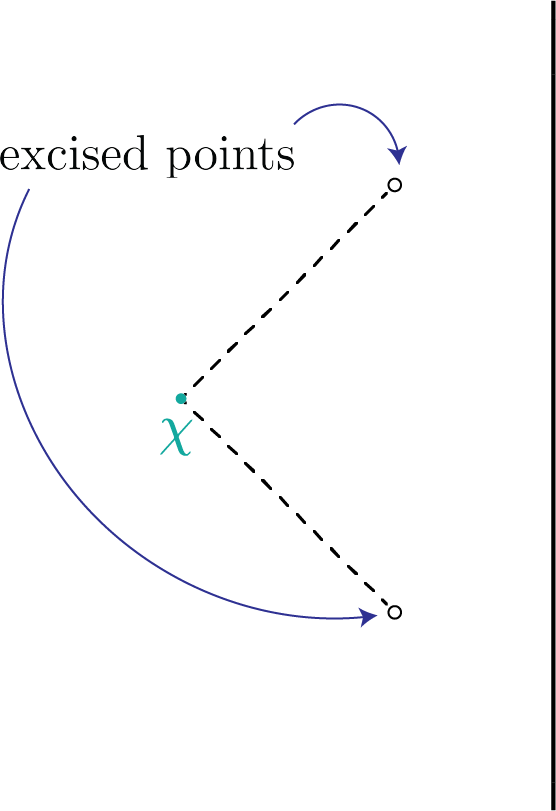}}
    \caption{An example where null geodesics fired from $\chi$ are geodesically incomplete, but there is no horizon.}
    \label{fig:excised-points}
\end{figure}

When $\chi$ is the minimal (and nonempty) QES for $\mathscr{I}$ in $M$, it immediately follows from the QES formula relating $S_{\rm gen}[\chi]$ to $S_{\rm vN}[\rho_{\mathscr{I}}]$ that bulk perturbations locally propagating from $\mathscr{I}$ (future- or past-directed) cannot modify $S_{\rm gen}[\chi]$ and thus cannot reach $\chi$: in this case, $\chi\cap I^{\pm}[\mathscr{I}]=\varnothing$.\footnote{Note that we have excluded the future horismos $\partial I^{\pm}[\mathscr{I}]$ from this conclusion, as $i^{\pm}$ are not precluded from null communication with $\chi$.} See  Fig.~\ref{fig:Sinv} for a review of the argument, originally presented in~\cite{EngWal14}. 

However, from the perspective of cosmic censorship, global minimality of a QES is of little relevance. We  expect that a similar result should apply to any QES independently of its global properties. In~\cite{EngFol20} some strides were made towards this end by CPT-conjugating the spacetime around the nonminimal QES $\chi$ to generate a new spacetime in which $\chi$ was minimal. This allowed~\cite{EngFol20} to prove that a large class of extremal surfaces and trapped surfaces must lie behind horizons. These results, however, relied on an assumption that we aim to relax in this section: that there can be no ``evaporating singularities'' (see~\cite{EngFol20} for the precise definition).

\begin{figure}
    \centering
    \includegraphics{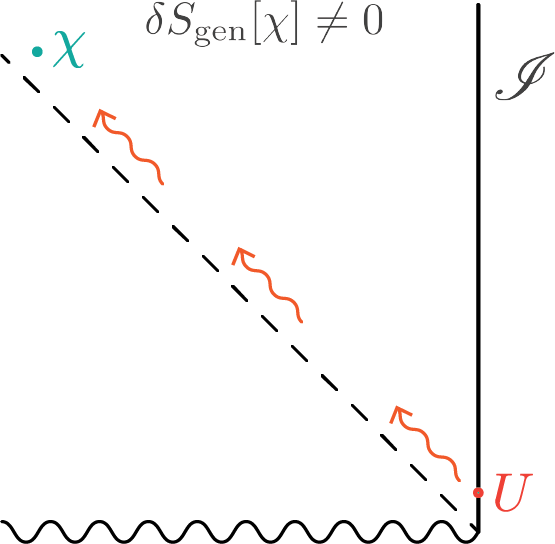}
    \caption{A QES outside the horizon, leading to a contradiction with the unitary invariance of von Neumann entropy: $S_{\rm vN}[\rho_{\mathscr{I}}] = S_{\rm vN}[U^{\dag}\rho_{\mathscr{I}} U]$.}
    \label{fig:Sinv}
\end{figure}

To relax this assumption, we will instead assume the Python's Lunch proposal~\cite{BroGha19, EngPen21b, EngPen23} in addition to the QES formula to argue that (quantum) extremal surfaces must generically be hidden behind both past and future horizons.\footnote{Note that the assumption of the Python's Lunch proposal applies to this section only; it will not be used in the proof of Cryptographic Censorship in Sec.~\ref{sec:proof-pseudorand-horizons}.} Or more precisely, if they lie outside a horizon, this distance must shrink to zero in the $G_N \rightarrow 0$ limit. Let us briefly review some facts about the Python's
Lunch proposal. A Python's Lunch $L$ with respect to a conformal boundary $\mathscr{I}$
is a region of spacetime defined by (at least) three mutually spacelike separated QESs homologous to $\mathscr{I}$:
$\chi_{\rm min}$, $\chi_{\rm bulge}$, and $\chi_{\rm aptz}$.
The first is the globally minimal QES, while $\chi_{\rm aptz}$ is a candidate
QES, but having larger generalized entropy than $\chi_{\rm min}$. The surfaces $\chi_ {\rm min}$ and $\chi_{\rm aptz}$ are 
so-called throats: QESs whose generalized entropy is roughly locally minimal in space
and locally maximal in time (see \cite{EngPen23} for details).
The bulge surface $\chi_{\rm bulge}$ is also a QES, but it is locally maximal in
time and space (see again \cite{EngPen23} for details).  The region $L$ is defined as the domain of dependence of any
spacelike slice bounded by $\chi_{\rm min}$ and $\chi_{\rm aptz}$, and the
Python's Lunch proposal conjectures that the CFT computational complexity required to
implement any operator in $L$ scales as 
\begin{equation}
\begin{aligned}
    \exp\left[\frac{ S_{\rm gen}[\chi_{\rm bulge}] - S_{\rm
    gen}[\chi_{\rm aptz}] }{ 2 } \right].
\end{aligned}
\end{equation}
We will only identify a Python's Lunch in a portion of spacetime that is
globally hyperbolic, and we will make the reasonable assumption that the presence of
potential spacetime singularities away from $L$ does not make operator reconstruction any easier.

Besides the Python's Lunch proposal, our assumptions in this section (and this section only) are as follows. 
 
\begin{paragraph}{Assumptions and Definitions.} We assume the following in addition to our general assumptions stated at the beginning of Sec.~\ref{sec:notation}:
\begin{itemize}
    \item $(M,g)$ has reflecting boundary conditions at each connected boundary component $\mathscr{I}$. 
    \item Smooth initial data: $(M, g)$ has a smooth inextendible timeslice $\Sigma_0$ containing
        two QESs $\chi_{\rm min}$ and $\chi$ homologous to $\mathscr{I}$, and
        where $\chi_{\rm min}$ is the QES of minimal
        generalized entropy that is homologous to $\mathscr{I}$.
    \item A generic condition: if $\chi'$ is a bulge QES and $\chi''$ a
        throat QES, then their associated generalized entropies satisfy  $S_{\rm gen}[\chi'] - S_{\rm gen}[\chi''] \sim
        O(G_N^{-1})$.
    \item In any globally hyperbolic region of spacetime that we consider, the (restricted) QFC holds, and the quantum maximin \cite{Wal12, AkeEng19b} and maximinimax \cite{BroGha19} prescriptions converge.
	\item We denote by $W_{O}[\chi]$ the outer wedge of a QES $\chi$ \cite{EngWal17b, EngWal18}.
    \item We define $t=0$ at $\mathscr{I}$ as $\Sigma_{0}\cap \mathscr{I}$.
    \item 
    If a point $p$ in the bulk is such that any future (past) directed causal
        curve from $p$ to $\mathscr{I}$ arrives at $\mathscr{I}$ at a boundary time that diverges to $+\infty$ ($-\infty$) in the $G_N \rightarrow 0$ limit, or if no curve to $\mathscr{I}$ exists, we say that $p$ is behind an ``effective future (past) horizon''. 
\end{itemize}
\end{paragraph}

\begin{restatable}{theorem}{QESbehindHors} Let $(M,g)$ be a spacetime as above. 
Assuming the Python's Lunch conjecture for complexity of operator
    reconstruction, the QES $\chi$ lies behind past and future effective horizons. \label{thm:warmup}
\end{restatable}
\begin{proof}
    Consider first working strictly within the globally hyperbolic spacetime
    given by the domain of dependence of $\Sigma_0$ within $M$, denoted as
    $D[\Sigma_0]$.
    By Proposition 1 of \cite{EngPen21a}, there always exists an outermost
    minimal QES $\chi_{\rm
    aptz}$ that lies in $W_{O}[\chi] \cap W_{O}[\chi_{\rm min}]$.\footnote{Outermost minimal means that there is no surface in the outer wedge $W_O[\chi_{\rm min}]$ homologous to $\chi_{\rm min}$ with less $S_{\rm gen}$.} This surface may or may not coincide with $\chi$ itself. Let now $H$ be a spatial slice with $\partial
    H= \chi_{\rm aptz} \cup \chi_{\rm min}$, which exists since $\chi_{\rm aptz} \subset W_{O}[\chi_{\rm min}]$,
    and let $L = D[H]$. See Fig.~\ref{fig:DoD} for an illustration.
    By the maximinimax construction in the appendix of
    \cite{BroGha19}, there is another QES $\chi_{\rm bulge}$ in $L$ that is of the
    bulge type, and so $L$ must be a Python's Lunch ($\chi_{\rm
    bulge}=\chi$ is allowed). $L$ must
    contain $\chi$, since there exists a choice of $H$ that contains
    $\chi_{\rm min}, \chi$, and $\chi_{\rm aptz}$. If $\chi$ is equal to $\chi_{\rm
    aptz}$ or $\chi_{\rm bulge}$, this is
    trivial. If $\chi$ is neither of these surfaces, then we construct $H$ by gluing together two
    pieces. The first piece is the subset of $\Sigma_0$ lying between
    $\chi_{\rm min}$ and $\chi$. The second is any spacelike slice between
    $\chi$ and $\chi_{\rm aptz}$ which exist by global hyperbolicity of
    $D[\Sigma_0]$ and the fact that $\chi_{\rm aptz}$ lies
    strictly in the outer wedge of $\chi$. 

        Let us now return to the full spacetime $M$. Define a code subspace
        $\mathcal{H}_{\rm code}$ consisting of all low energy perturbative quantum states that backreact on $\Sigma_0$ at subleading order in $G_N$. 
        By the genericity assumption, we have 
        $S_{\rm gen}[\chi_{\rm bulge}] -  S_{\rm gen}[\chi_{\rm aptz}]\sim
        O(G_N^{-1})$
        so by the Python's Lunch proposal,
        the complexity of reconstructing operators in $\mathcal{H}_{\rm code}$ 
        with support in $L$ scales as $\exp(cG_N^{-1})$ for some constant $c>0$. Assume now for
        contradiction that $\chi$ lies outside the effective past horizon. Then
        there exists a future-directed causal curve $\gamma$ from $\mathscr{I}$ to $\chi$ that is fired from $\mathscr{I}$ at a time $t\sim O(G_N^0)$.
        Then we can find a simple operator $\Phi$ that modifies the quantum fields at $\chi$. Just let $\Phi$ be an operator that injects a perturbative
        amount of matter
        with rocket boosters near $\mathscr{I}$, so that the matter travels along the
        curve $\gamma$. Because this matter reaches $\chi$ in a time that does
        not scale with $G_N$, this matter backreacts on $\Sigma_0$ only at
        subleading order, and so it acts within the code subspace. $\Phi$ can
        be constructed in two steps: an injection of matter near the boundary
        using HKLL \cite{HamKab05, HamKab06, HamKab06b} followed by an $O(G_N^0)$ amount of CFT time evolution.
        Neither is exponentially hard, and so we reach a contradiction with reconstruction complexity of operators with support in $L$. The
        argument excluding $\chi$ from intersecting the exterior of the effective future horizon
        proceeds \emph{mutatis mutandis}.
\end{proof}

\begin{figure}
    \centering
    \includegraphics[scale=0.7]{{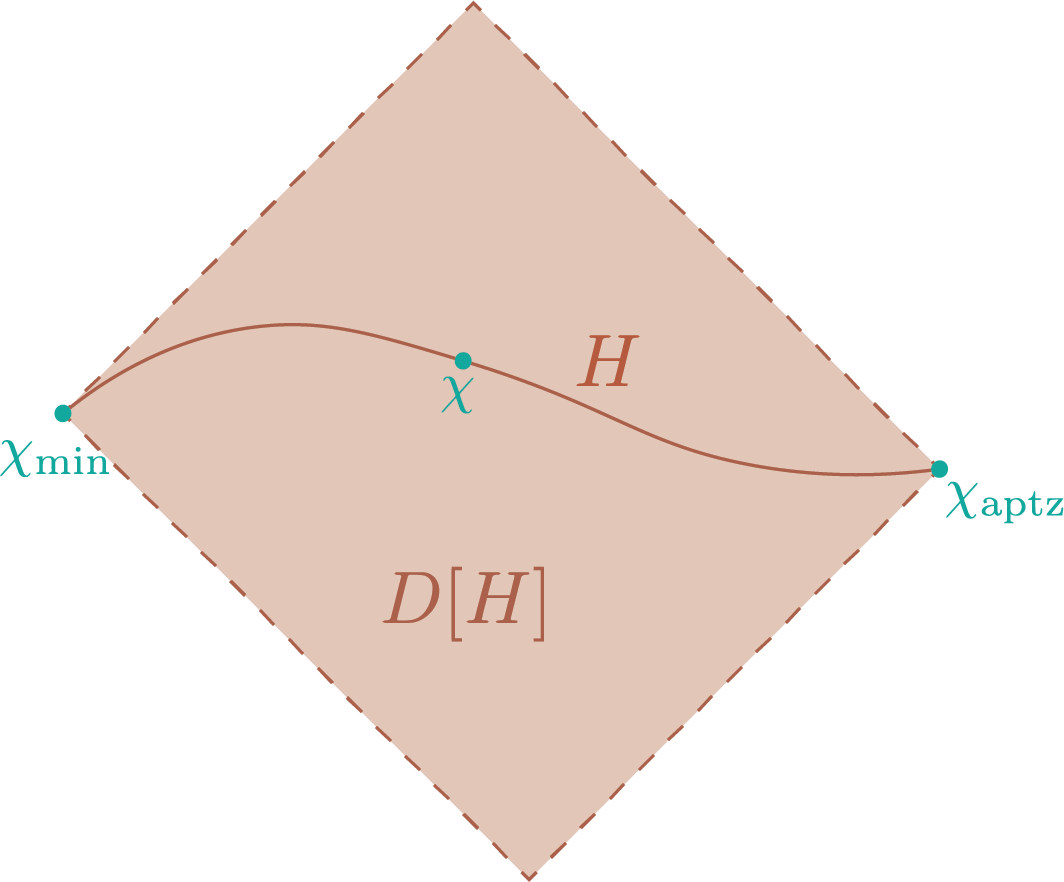}}
    \caption{The setup in the proof of Theorem~\ref{thm:warmup} where $\chi$ is nonminimal, and where $H$ is chosen to contain $\chi$.}
    \label{fig:DoD}
\end{figure}

Note that there are setups where islands \cite{AlmMah19b, BouPen23}, and thus QESs $\chi$, lie outside of
event horizons.  A universal aspect of these constructions is that the QES in question is a QES of a different subsystem than the one whose horizon is under consideration. 
We may also comment more specifically on each example: first, Ref.~\cite{AlmMah19b} found an island outside of the black hole horizon when 
the CFT at $\mathscr{I}$ was coupled to a reservoir. In this case boundary
conditions were not reflecting, which was assumed in our result. Furthermore, by virtue of the coupling, time evolution of the black hole system alone is no longer unitary. Thus, while a signal sent from $\mathscr{I}$ initially can be
turned near the boundary with a unitary, propagating it to $\chi$ requires a non-unitary operation
on the CFT, and so the entropy is expected to change. Additionally, 
decoupling the CFT from the auxiliary system produces a shockwave that alters the spacetime at leading order, 
which does push the island behind the horizon.
Next, in \cite{BouPen23} an island was found lying a distance of
$O(\sqrt{\ell_{\rm Pl}r_{\rm hor}})$ outside the horizon. In this case
the island was associated to angular momentum modes in the bulk, but qualitatively this construction is the same as Ref.~\cite{AlmMah19b} as far as it relates to the above theorem. 

\section{Crytographic Censorship} \label{sec:cryptoCensor}

In this section we introduce and prove our main result: that on any sufficiently large code subspace, (pseudo)random dynamics guarantee event horizon formation in typical (pseudo)random states with a semiclassical bulk dual. The precise statement (Theorem \ref{thm:pseudorandom-horizons}: Cryptographic Censorship) is given in Sec.~\ref{sec:proof-pseudorand-horizons}. In presenting and proving this theorem, we will be agnostic about the types of code subspaces that in a gravitational theory could correspond to this degree of pseudorandomness; in Sec.~\ref{sec:possible-subspaces} we will give two examples of code subspaces in gravity that are compatible with sufficiently pseudorandom dynamics. 

Intuitively, pseudorandom unitary ensembles are ensembles of unitaries that are \emph{computationally indistinguishable} from the Haar ensemble. As an example, upon evolution for some time longer than a scrambling time, a chaotic system is generally expected to become approximately pseudorandom. Our argument for the existence of a horizon hinges upon the assumption that the fundamental time evolution is well modeled by a unitary that is sufficiently pseudorandom on the code subspace. In a sense described below, we will require the unitary to be pseudorandom by an amount that scales inversely with $G_N$, which we will call \emph{gravitationally} pseudorandom.

How is this assumption of gravitational pseudorandomness sufficient to guarantee the existence of event horizons? We will outline the logic here while omitting many of the technicalities (which we leave to Apps.~\ref{app:proof-hol-learning} and~\ref{app:proof-large-N}); the rest of this section is devoted to the technical points necessary for the proof of Cryptographic Censorship and its application to gravitational code subspaces.

Recall that observables in the causal --- and more generally, simple \cite{EngWal17b, EngWal18, EngPen21a} --- wedge of the boundary are computationally simple to reconstruct. This is a non-trivial aspect of bulk reconstruction in AdS/CFT \cite{EngWal18, EngLiu23}.\footnote{For the likely limited number of readers who believe in the Quantum Extended Church-Turing (QECT) thesis outside of horizons but are skeptical of arguments relying on AdS/CFT, we note that the QECT thesis by itself also implies that causal observables must be simple.} We argue that, given some reasonable technical assumptions outlined in Sec.~\ref{sec:proof-pseudorand-horizons}, there exists a learning algorithm $\mathcal{A}_{CW}$ that can reconstruct the causal wedge of any geometric state in a code subspace of size $e^{O(G_N^{a})}$ for some $-1\leq a < 0$. However, a theorem of~\cite{YanEng23} (see Theorem~\ref{thm:learning} below for its precise statement) combined with results from measure concentration guarantees (roughly) that all efficient quantum algorithms must fail at predicting the expectation value of some operator $Q$, which we show exists in the large-$N$ limit, after pseudorandom time evolution. It then immediately follows that for the algorithm $\mathcal{A}_{CW}$, there exists an operator $Q$ that $\mathcal{A}_{CW}$ cannot reconstruct and therefore must lie outside of the causal wedge; in other words, an exterior to the causal wedge exists --- by extension, a boundary to the causal wedge exists. Since  the boundary of the causal wedge of $\mathscr{I}$ coincides with the event horizon, an event horizon must exist. In Sec.~\ref{sec:proof-pseudorand-horizons}, we will formalize this rough argument into a proof.

We start by briefly reviewing pseudorandom unitaries, as well as the appropriate learning result in \cite{YanEng23} that we will use, in Sec.~\ref{sec:PR-and-learning-result}. We then explain the required alterations and extensions necessary to apply the learning result to the holographic, $N \to \infty$ setting in Sec.~\ref{sec:complearnholog}. In Sec.~\ref{sec:proof-pseudorand-horizons} we state and prove our theorem of Cryptographic Censorship. Finally, in Sec.~\ref{sec:possible-subspaces} we comment on two code subspaces to which our theorem applies.

\subsection{Review: pseudorandomness and complexity of learning} \label{sec:PR-and-learning-result}  
Pseudorandom states and unitaries will play an important role in what follows, so we will briefly review these concepts. Within the context of quantum cryptography, pseudorandom states were first defined in \cite{JiLiu18}. Informally, an ensemble of quantum states is pseudorandom if the states are efficiently constructible, but indistinguishable from Haar random by any polynomially bounded quantum algorithm --- even if that algorithm is given polynomially many copies of the state. A pseudorandom state generator $G$ is a map from some space $\mathcal{K}$ into a Hilbert space $\mathcal{H}$. For each element $k \in \mK$, it outputs $G(k) = \ket{\psi_k} \in \mathcal{H}$.  The intuitive goal of the definition is to say that the action of feeding a random sample from $\mK$ to $G$ to get a state $\ket{\psi_k}$ is indistinguishable from sampling the Haar random distribution of states (in the computational sense defined in Sec.~\ref{sec:notation}). The technical definition is as follows: 

\begin{defn}[Pseudorandom State Ensemble (PRS); Definition 2 in \cite{JiLiu18}]\label{defn:PRS}

Consider an ensemble of states $\Psi=\{\ket{\psi_k} \in \Hil\}_{k \in \mathcal{K}}$, labeled by a parameter $k$ drawn from a key space $\mK$, where both $\Hil$ and $\mathcal{K}$ are parameterized by a security parameter $\kappa$. Then the ensemble $\Psi$ is a pseudorandom state ensemble if 
\begin{itemize}
    \item The states can be efficiently generated: there is an efficient quantum algorithm $G$ that generates the state $\ket{\psi_k}$ given $k$ as input. That is, for all $k \in \mathcal{K}$, $G(k) = \ket{\psi_k}$. 
    \item The states are computationally indistinguishable from Haar random:  for any polynomially many copies of $\ket{\psi_k}$ with the same $k \in \mathcal{K}$ and any efficient quantum algorithm $\mathcal{A}$, there exists a negligible function $\negl (\kappa)$ such that
    \begin{equation}
        \left|\Pr_{k\gets\mathcal{K}}\[\Alg\(\ket{\psi_k}^{\otimes \poly (\kappa)}\)=1\]-\Pr_{\ket{\psi}\gets\mu}\[\Alg\(\ket{\psi}^{\otimes \poly (\kappa)}\)=1\]\right|=\negl(\secp)~,
    \end{equation}
    where $\mu$ is the Haar measure on $\Hil$. 
\end{itemize}

\end{defn} 

We can define pseudorandom unitaries as in~\cite{JiLiu18} in a similar fashion:

\begin{defn}[Pseudorandom Unitary Ensemble (PRU); Definition 5 in~\cite{JiLiu18}]\label{defn:PRU}
    Consider an ensemble of unitaries,$\ \mathcal{U} = \lbrace U_k \in \mathbf{U}(\mH) \rbrace_{k \in \mK}$, labeled by a parameter $k$ drawn from a key space $\mK$, where both $\mK$ and $\mH$ are parameterized by a security parameter $\secp$. We say $\,\mathcal{U}$ is a pseudorandom unitary ensemble if 
    \begin{itemize}
        \item The unitaries are efficiently representable: there is an efficient quantum algorithm $Q$ such that for every $k \in \mK$ and any $\ket{\psi} \in \mathcal{H}$, $Q(k,\ket{\psi}) = U_k \ket{\psi}$.
        \item The unitaries are computationally indistinguishable from Haar random: for any efficient quantum algorithm $\mathcal{A}^U$ given oracle access to some $U$, there exists a negligible function $\negl (\kappa)$ such that
        \begin{align}
        \left | \underset{k \leftarrow \mK}{\mathrm{Pr}} \left[ \mathcal{A}^{U_k} =1 \right] -\underset{U \leftarrow \mu}{\mathrm{Pr}} \left[ \mathcal{A}^{U} =1 \right]  \right| = \mathrm{negl}(\secp)~.
        \end{align} 
    \end{itemize}
\end{defn} 
We note that technically, PRSs and PRUs are actually sequences of ensembles, $\PRSens = \{\PRSens_\kappa\}_\kappa$ and $\PRUens = \{\PRUens_\kappa\}_\kappa$, respectively, each parametrized by $\kappa \in \Nat$, and the guarantees in the above definitions are for large $\secp$. Since we will take $\kappa = \poly(N)$ in Sec.~\ref{sec:complearnholog} and later sections, this aligns well with the framework of asymptotically isometric codes outlined in Sec.~\ref{sec:notation}. For ease of notation, we will instead simply write $\PRSens$ and $\PRUens$, as in the above definitions, where the key $k$ should be understood to be drawn from a key space $\mathcal{K}_\kappa$ where $\kappa$ is sufficiently large. We will make an exception to this rule in some parts of the proof in App.~\ref{app:proof-large-N}, where it will be convenient to make the dependence on the security parameter explicit to take the $\kappa \to \infty$ limit. 

Finally, we remark that while explicit constructions of PRSs exist (and were given first in \cite{JiLiu18}), there are currently no known constructions of PRUs, although candidate constructions were proposed in \cite{JiLiu18}. We will not require any explicit constructions of either PRSs or PRUs in this work.

\paragraph{Pseudorandom states vs.\ $k$-designs.} 
For the reader familiar with the related notion of a $k$-design, we briefly comment on the two main differences with PRSs. First, the associated notion of indistinguishability is much stronger for $k$-designs: the states need to be indistinguishable from Haar random even to computationally \emph{unbounded} observers, e.g., algorithms of exponential complexity. However, $k$-designs only require that the moments of these ensembles are indistinguishable up to a certain value ($k$), whereas PRS ensembles are computationally indistinguishable for any $\poly (\kappa)$ number of moments. In this sense, PRSs are comparable to a high-order $k$-design but with a parametrically weaker notion of indistinguishability. Note that $k$-designs are less relevant for our purposes as they apply to unbounded observers, whereas our focus here is on bounded computations. 

~\\
We now state the learning result of \cite{YanEng23}, which was briefly reviewed at a lower level of detail in Sec.~\ref{sec:intro}. This result (and our Cryptographic Censorship result) also holds for Haar random $U$ and $\ket{\PRS}$, but we will focus on pseudorandom states and unitaries as we believe they present a more realistic characterization of the time evolution of physical systems. 

\begin{restatable}[Complexity of Learning Pseudorandom Unitaries, Theorem 2 of~\cite{YanEng23}]{theorem}{learning}\label{thm:learning}  
Let $\secp\in\Nat$ be the security parameter, and let $\PRUens$ be any pseudorandom unitary ensemble and $\PRSens$ be any pseudorandom state ensemble in a Hilbert space $\Hil$ where $|\Hil|=2^\secp$.
For any quantum algorithm $\Alg$ of $\poly(\secp)$ complexity that, given oracle access to the pseudorandom unitary $\PRU \gets\PRUens $, produces (a circuit representation of) an approximately norm-preserving operator\footnote{We note that~\cite{YanEng23} showed this result for a wider class of algorithms --- those producing any operator $\widehat{U}$ with column norms bounded above by $\approx 1$ --- which includes the algorithms here. We will not need this more general result in this paper.} $\widehat{\PRU}$ as its best guess for $\PRU$, there exists a constant $\const > 0$ independent of $\secp$ such that for any $\secp > \secp_0$,  the following quantity is bounded:
\begin{equation*}
\avgover{\substack{\PRU \gets\PRUens \\ \ket{\PRS}\gets\PRSens \\ \widehat{U} \gets \Alg^U}}\[F\(\widehat{\PRU}\ket{\PRS},\PRU\ket{\PRS}\)\]\leq 1-\const~.
\end{equation*} 
Here, $F(\widehat{U}\ket{\psi},U\ket{\psi}) = |\braket{\psi|\widehat{U}^\dagger U|\psi}|^2$ which is the fidelity when $\widehat{U}$ is unitary. In this context, it is approximately the fidelity since we are working with approximately norm-preserving operators. 
\end{restatable}  

This theorem shows that it is impossible for any efficient quantum learning  algorithm to use oracle access to a pseudorandom $U$ to produce a polynomially complex circuit or (succinct) algebraic expression that accurately implements $U$ on a pseudorandom state $\ket{\psi}$. As we will explain in more detail in the next section, we will be interested in applying this result to a sequence of code subspaces $\{\Hil_\code^{(\secp)}\}_\secp$ each of dimension $|\Hil_\code^{(\secp)}| = 2^\secp$, where we will take $\secp = \poly(N)$.

We now comment on two aspects of this theorem that will appear extensively in the subsequent section.  
First, Theorem~\ref{thm:learning} implies that, \emph{on average}, there exists a distinguishing operator that the learning algorithm fails to predict (see Appendix~B in~\cite{YanEng23} for an explanation).\footnote{As a difference of language,~\cite{YanEng23} writes this as the distinguishing ``POVM'' and lower bounds the distance between measurement outcomes obtained via this POVM.} In particular, when the fidelity between $\widehat{U}\ket{\psi}$ and $U\ket{\psi}$ is $1-O(1)$, then there exists an operator $\operator$ such that\footnote{Note that the quantum algorithm $\Alg^U$, as a result of, e.g., intrinsic randomness in the outcome of measurements during the learning process, generally outputs a $\widehat{U}$ that is not a deterministic function of $U$. Therefore, we average over the algorithm's output $\widehat{U}\gets\Alg^U$.}
\begin{equation}
    \avgover{\widehat{U}\gets \Alg^{U} }  \left|\tr(\operator\widehat{\PRU}\PRS (\widehat{\PRU})^\dagger)-\tr(\operator\PRU\PRS(\PRU)^\dagger)\right|\geq O(1)~.
\end{equation}
In the next section, we will use measure concentration to show that for \emph{nearly all} of the unitaries in the ensemble $\PRUens$, such a fidelity bound holds and thus there is a distinguishing operator.\footnote{The result on average in~\cite{YanEng23} only implies that a constant fraction of the unitaries in $\PRUens$ have such an operator, whereas we will show this holds for a fraction of the unitaries that goes to $1$ as $\secp\to\infty$.} This operator will play a crucial role in the proof of Cryptographic Censorship. 

Second, it immediately follows from Theorem~\ref{thm:learning} that
\begin{equation}\label{eqn:avg-fidelity-supremum} 
\sup_{\secp\to\infty} \avgover{\substack{\PRU\gets\PRUens\\ \ket{\PRS}\gets\PRSens\\ \widehat{U}\gets\Alg^U}}\[F\(\widehat{\PRU}\ket{\PRS},\PRU\ket{\PRS}\)\] \leq 1-\const~,
\end{equation}
which crucially depends on the upper bound of $1-\const$ being independent of $\secp$.   
This will be important in extending and applying this result to holography, where we will take the large $\secp$ (or equivalently, large $N$) limit, and want the guarantee that the large $N$ limits of these states also have fidelity at most $1-\alpha$, and thus a distinguishing operator exists.

\subsection{Complexity of learning in holography}\label{sec:complearnholog}

We now turn to the application of this learning result to holography. Since our goal is to constrain semiclassical gravitational phenomena, our focus will be on the large-$N$ limit in general and on quantities that scale with $N$ in this limit in particular. For concreteness, and without loss of generality, we use $G_N \sim 1/N^2$, and consider code subspaces $\Hil_\code$ of dimension $e^{O(G_N^{a})}$, where $-1\leq a < 0$. To this end, we will take the security parameter $\secp = \poly (N)$ such that our code subspace $\mathcal{H}_{\rm code}$ is of dimension $2^\kappa$.  Unitary operators that are pseudorandom on such code subspaces at any large but finite $N$ will be called gravitationally pseudorandom (defined precisely below); we will also resolve various technical challenges that follow from taking the large-$N$ limit (e.g., with regards to the existence of certain quantities). We begin with a formal definition of gravitational pseudorandomness.

\begin{defn}[Gravitationally Pseudorandom Unitary Ensemble (GPRU)]\label{defn:gravpsr} 
    Let $\mathcal{U} = \lbrace U_k \in \mathbf{U}(\mH) \rbrace_{k \in \mK}$ be a pseudorandom unitary ensemble, with key space $\mK$ parameterized by a security parameter $\secp$, where $\mH$ is the Hilbert space of a large $N$, holographic quantum system where  $N \sim G_N^{-1/2}$, with $G_N$ the bulk Newton's constant. We say that \,$\mathcal{U}$ is a \emph{gravitationally pseudorandom unitary ensemble (GPRU)} if $\kappa = \poly(N)$. 
\end{defn}   
As noted previously, at strictly infinite $N$ there is no notion of pseudorandom unitaries. When discussing operators we shall be primarily interested in asymptotic scaling at large $N$ rather than the actual operators at infinite $N$. When working in the actual infinite $N$ regime, we will only discuss limits of pseudorandom \textit{states}, using the asymptotically isometric codes framework outlined in Sec.~\ref{sec:notation}.

Recall that we aim to apply the results of \cite{YanEng23} to the holographic setup. The careful reader may have by now noticed two potential points of concern: first, Theorem~\ref{thm:learning} was proved on average for a PRU ensemble rather than for some typical instance; second, it is not obvious that in the large-$N$ limit, for a typical PRU, a distinguishing operator exists. Let us briefly elaborate on each concern before proceeding to state and prove our generalization of Theorem~\ref{thm:learning} that addresses both. 

On the first concern, the learning result of~\cite{YanEng23} is a statement about the \textit{average} fidelity between the evolution of a pseudorandom state (PRS) evolved with the actual PRU and the same state evolved with the algorithm's best guess. The averaging is over the unitary drawn from any PRU ensemble $\PRUens$, the state drawn from any PRS ensemble $\Psi$, and over the output of any quantum learning algorithm $\mathcal{A}$.

While much of the literature (see, e.g., for a few instances~\cite{HayPre07,HarHay13,KimTan20}) does model gravitational systems via randomly sampled states or dynamics, this approach is in some tension with the fact that there is just \textit{one} (fixed) unitary operator that describes the time evolution of ${\cal N}=4$ SYM. We therefore generalize the bound from an average to a \emph{typical} pseudorandom unitary and state by bounding the fluctuations of typical members of the ensemble, first for Haar random unitaries and states and then for pseudorandom unitaries and states. We emphasize that this means a learning bound applies to any \emph{single fixed} unitary and state that are both typical, as we state below in Theorem~\ref{thm:learning-holography}.

The second concern, involving the existence of a ``non-simple observable'' --- i.e., the distinguishing observable that the algorithm fails to predict --- for typical instances of the ensemble at large-$N$, is a bit more subtle. As a corollary to our theorem for learning typical pseudorandom unitaries at large but finite $N$, we prove, under a reasonable assumption on the limit of states as $N\to \infty$, that our theorem also bounds the fidelity at $N\to\infty$. Consequently, our theorem guarantees that for any simple boundary algorithm, there exists a \emph{semiclassical} observable that cannot be reconstructed by the algorithm.

~\\ 
\noindent We now state the holographic version of Theorem \ref{thm:learning}, incorporating the generalizations advertised above. We will subsequently use this version in Sec.~\ref{sec:proof-pseudorand-horizons} to prove Cryptographic Censorship. 

Consider a CFT with a large but finite UV energy cutoff, so that when $\HN$ is finite the dimension of the Hilbert space is finite.\footnote{We are free to consider other ways of regulating the dimension of the Hilbert space if needed.} We prove the theorem below for any large but finite value of $\HN$, under the assumption that the algorithm's output as a function of $U$ is (randomized) $K$-Lipschitz for some fixed constant $K$. We will describe this assumption in greater length in App.~\ref{app:proof-hol-learning}, but it boils down to a certain notion of continuity controlled by a parameter $K$ of the algorithm as a function of $U$.\footnote{It is possible that this assumption actually follows directly from our other assumptions; we leave determination of this question to future work.} We subsequently state and prove a corollary of this theorem for infinite-dimensional Hilbert spaces, under a class of assumptions. To increase readability, we postpone both of the proofs to Apps.~\ref{app:proof-hol-learning} and~\ref{app:proof-large-N}.

\begin{restatable}[Complexity of Learning for Typical Pseudorandom Unitaries]{theorem}{learningholography}\label{thm:learning-holography}   
Let $\secp=\poly(\HN )$ and consider a code subspace of dimension $2^{\kappa}$. Consider any (gravitationally) pseudorandom unitary ensemble $\PRUens$, any pseudorandom state ensemble $\PRSens$, and any constant $\epsilon>0$. For any (randomized) $K$-Lipschitz quantum algorithm $\Alg$ of $\poly(\HN)$ complexity that has oracle access to $U \gets \PRUens$ and produces an approximately norm-preserving operator $\widehat{U}$ as its best guess for $U$, and for any sufficiently large $\HN$, the following holds: for a fraction $\geq 1-e^{-\epsilon^2 2^\kappa /\left(48(1+K)^2\right)} - e^{-\epsilon^2 2^\kappa /192} - \negl'(N)$ of $\PRU\in\PRUens$ and $\ket{\PRS}\in\PRSens$, we have that:

\begin{enumerate} 
\item the following quantity is bounded 
    \begin{equation}\label{eqn:theorem-typical-fidelity} 
        \avgover{\widehat{U}\gets \Alg^{U} } \[ F\(\PRU\ket{\PRS}, \widehat{\PRU}\ket{\PRS}\) \] \leq 1-\const+2\epsilon + \negl (N)~, 
    \end{equation} 
where $\const$ is a constant independent of $\HN$, and
\item there exists an operator $\operator$ (which may depend on $U$, $\ket{\psi}$, and $\widehat{U}$) such that
    \begin{equation}\label{eqn:theorem-operator} 
       \avgover{\widehat{U}\gets \Alg^{U} }  \left|\tr(\operator\widehat{\PRU}\PRS (\widehat{\PRU})^\dagger)-\tr(\operator\PRU\PRS(\PRU)^\dagger)\right|\geq 2\const-4\epsilon - \negl(N)~.
    \end{equation}
\end{enumerate} 
\end{restatable}

\begin{proof}
See App.~\ref{app:proof-hol-learning}.
\end{proof}

We now extend this theorem to infinite-dimensional Hilbert spaces. For this purpose, in what follows, we reinstate the sequences of PRSs and PRUs to make the $N$-dependence explicit. We now assume that for every $N$ we have a code subspace $\mathcal{H}_N$ on which the time evolution operator $e^{-iH_Nt}$ of the
CFT, for some $t$ (possibly scaling with $N$), has the same matrix elements on the code
subspace as another unitary operator $U_N: \mathcal{H}_N
\rightarrow\mathcal{H}_N$, up to exponential corrections. As in Assumption \ref{ass:isometry}, we also assume that there exists an asymptotically isometric code with map $V_N$ from the bulk Hilbert space, $\mathcal{H}$, to $\mathcal{H}_N$. Furthermore, we
assume there exists a
(gravitationally) PRU ensemble $\{\mathcal{U}_N\}_{N}$ and some $\HN_0$ such that for all $N\geq N_0$, the ensemble $\mathcal{U}_N$ contains 
$U_N$ as a typical member.

Next, we assume there exists a sequence of pseudorandom states $\{\ket{\psi_N} \in \Hil_N\}_\HN$ that are represented on the code. Namely, there exists a $\ket{\xi} \in \mathcal{H}$ such that for every $\ket{\psi_N}$ in the sequence with sufficiently large $N$, $\ket{\psi_N} = V_N \ket{\xi}$ up to exponentially small corrections in the one-norm. Finally, letting the learning algorithm's best guess for
$U_N$ be $\widehat{U}_N$, we assume that $\widehat{U}_N\ket{\psi_N}$ is also represented by the asymptotically isometric code, i.e., there exists a $\hat{\xi}$ such that $\widehat{U}_N\ket{\psi_N} = V_N \ket{\hat{\xi}}$. We also assume that the subset of states at finite $N$ that are the fraction of states that are typical converge to a subset of states in the infinite-$N$ Hilbert space that is of measure $1$ (as the fraction of typical states limits to $1$ as $N\to\infty$ in Theorem~\ref{thm:learning-holography}).

\begin{restatable}[Distinguishing Operator at $N\to\infty$]{corollary}{largeNlimitredo}\label{cor:large-N-limit}  
Let $\{\PRU_N\ket{\PRS_N}\}_\HN$ and $\{\widehat{\PRU}_N\ket{\PRS_N}\}_\HN$ be two sequences of states that are represented on the code, namely there exist $\ket{\xi}$ and $\ket{\hat{\xi}}$ such that $\PRU_N\ket{\PRS_N} = V_N \ket{\xi}$ and $\widehat{\PRU}_N\ket{\PRS_N} = V_N \ket{\hat{\xi}}$. Furthermore, assume that $\{\PRU_N\ket{\PRS_N}\}_\HN$ and $\{\widehat{\PRU}_N\ket{\PRS_N}\}_\HN$ satisfy the assumptions stated above, then both 1. and 2. in Theorem~\ref{thm:learning-holography} also hold for the states $\ket{\xi}$ and $\ket{\hat{\xi}}$. In particular, by taking the constant $\epsilon$ to be $< \const/2$, we have for any $\ket{\hat{\xi}}$ the existence of an operator $\operator\in \mathcal{B}(\mathcal{H})$ such that $$\avgover{\widehat{\xi}} \left|\tr(\operator\hat{\xi})-\tr(\operator\xi)\right| \geq O(1)~.$$
\end{restatable}

\begin{proof}  

See App.~\ref{app:proof-large-N}.

\end{proof}

\subsection{Pseudorandomness implies horizons}\label{sec:proof-pseudorand-horizons}

We are now ready to argue that when time evolution in the boundary theory is sufficiently (i.e., gravitationally) pseudorandom, the bulk has an event horizon. Recall that the logic in this argument is that for a typical pseudorandom time evolution operator $U$, (a) efficient quantum learning algorithms cannot accurately compute the time evolution of $\langle Q \rangle$ for some operator $Q$ that may depend on $U$, and (b) there exists an efficient quantum algorithm that can accurately compute the time evolution of all local operators --- i.e., operators in the causal wedge --- under $U$. Point (a) follows directly from Corollary~\ref{cor:large-N-limit}. 

Point (b) is in a sense obvious holographically given the developments on complexity of reconstruction in AdS/CFT. The general expectation following from~\cite{EngPen21a} is that a simple quantum algorithm should be able to reconstruct the causal wedge. We might worry that the existence of a naked singularity complicates this story, but there is a simple physical argument suggesting otherwise. To measure an operator in the causal wedge, we can send rocket ships from the conformal boundary that travel into the bulk, carry out measurements anywhere in the causal wedge, and then travel back out to the boundary with the result. Since we are only interested in measuring operators that survive at strictly infinite $N$, the complexity to implement these rocket ships from the CFT should not scale exponentially with $N$.\footnote{One might worry that the bulk measurements take too long for the rocket ships to be able to get back out to the boundary. Again, this should not be the case for bulk operators that survive the infinite $N$ limit.} Under the assumption that there exists an efficient process that can produce a guess $V$ for the time evolution operator $U$ using these measurements, we can then combine this with point (a) above to formulate our theorem of Cryptographic Censorship, Theorem~\ref{thm:pseudorandom-horizons} below. 
\begin{figure}
    \centering
    \includegraphics[scale=0.7]{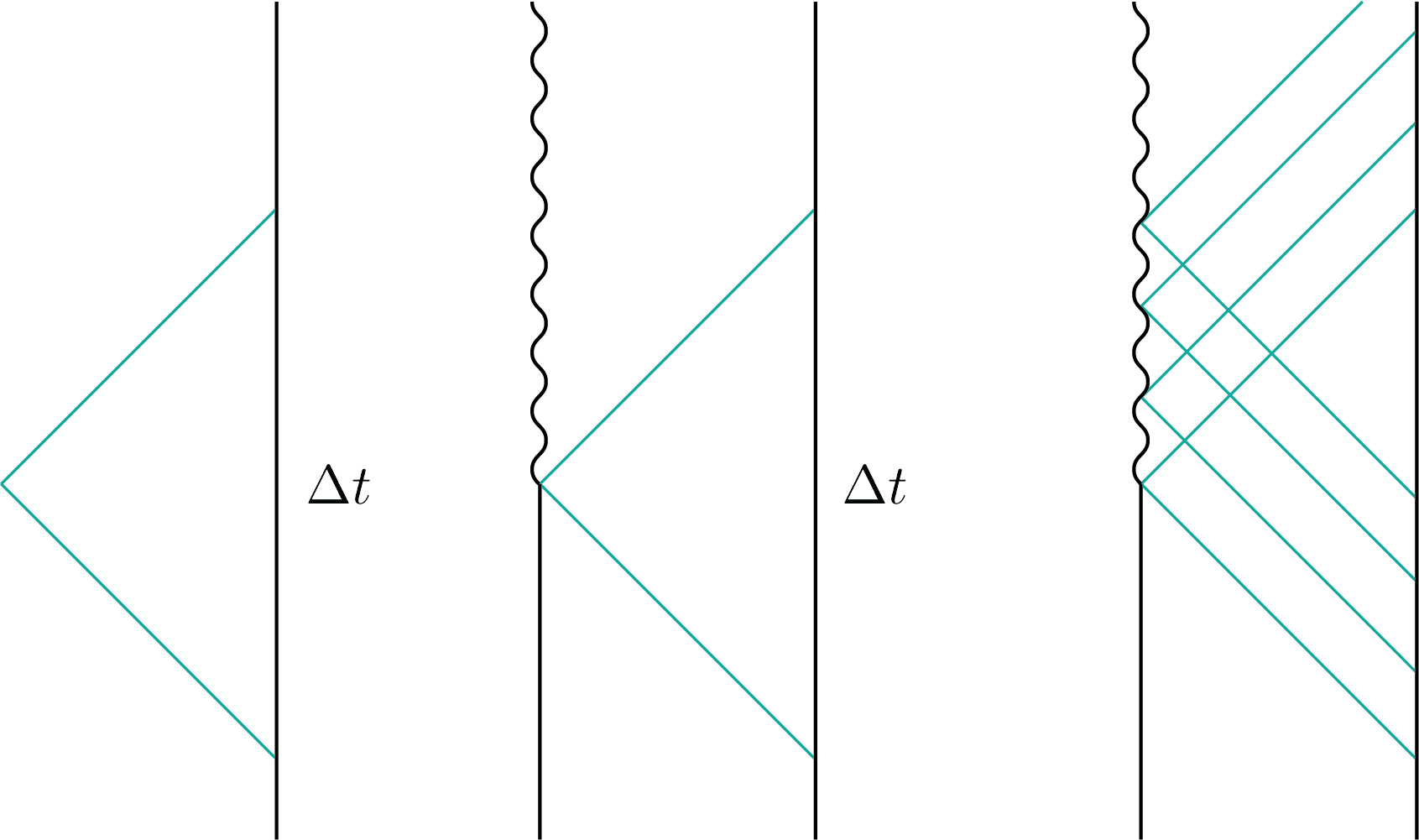}
    \caption{In the leftmost and middle panels, the reconstruction of operators inside the green wedge (which is globally hyperbolic) can be done using the HKLL protocol, which should be simple, insofar as we can compute the expectation values of local operators in the boundary theory on the time band $\Delta t$. We can keep doing this for subsequent time bands until we have reconstructed the full causal wedge, as in the rightmost panel. We use this as an ingredient to motivate the existence of an algorithm that can output an operator $V$ which reproduces the time evolution of operators in the causal wedge.}
    \label{fig:HKLL-argument}
\end{figure}

In many cases, we can motivate this assumption --- the existence of an efficient learning algorithm that can produce such a guess $V$ --- from a computer science point of view. For instance, consider first a domain of dependence corresponding to a boundary time band. If we have access to the time-evolved expectation values of single-trace operators in that time band, we can use the HKLL protocol \cite{HamKab05, HamKab06, HamKab06b} to reconstruct this domain of dependence. See Fig.~\ref{fig:HKLL-argument}. 
In the case where the causal wedge can be generated from a union of such wedges, and if there exists a simple algorithm that can aggregate all these reconstructed wedges into a guess for the time evolution operator, then we can efficiently reconstruct the full causal wedge. There is in fact a learning algorithm that comes close to providing such a guess $V$, as we will explain in more detail below; this thus provides strong motivation that learning the causal wedge is simple, even in the presence of such naked singularities. 

However, it is possible that there exist naked singularities with unusual topologies where this reasoning falls short and we cannot use this `progressive HKLL' protocol combined with the aforementioned existing learning algorithm. An illustration of such a hypothetical singularity is given in Fig.~\ref{fig:cylindrical-sing}, where the interior of a cylindrical singularity does not lie inside any HKLL wedge. In such cases, however, we revert to our initial holographic reasoning to motivate our assumption: we can still use rocket ships to measure operators in the entire causal wedge.

For the case where HKLL in progressive wedges does work, we will now illustrate more precisely that under a mild technical assumption, there exists a quantum learning algorithm satisfying the assumptions of Theorem~\ref{thm:learning} that, given oracle access to $U$, can output an operator $V$ where $V\ket{\psi}$ can accurately reproduce the expectation values of all simple operators in the state $U\ket{\psi}$. Applicability to our strengthened version, Theorem~\ref{thm:learning-holography}, does not immediately follow and would require techniques from measure concentration. This falls outside the scope of the current work, but we expect that this may go through if the algorithm in~\cite{HuaChe22} satisfies certain assumptions. Readers willing to accept the existence of the relevant algorithm --- i.e., an efficient algorithm for reconstructing the causal wedge --- can skip directly to the statement of Cryptographic Censorship, Theorem \ref{thm:pseudorandom-horizons} below.
\begin{figure}
    \centering
    \includegraphics{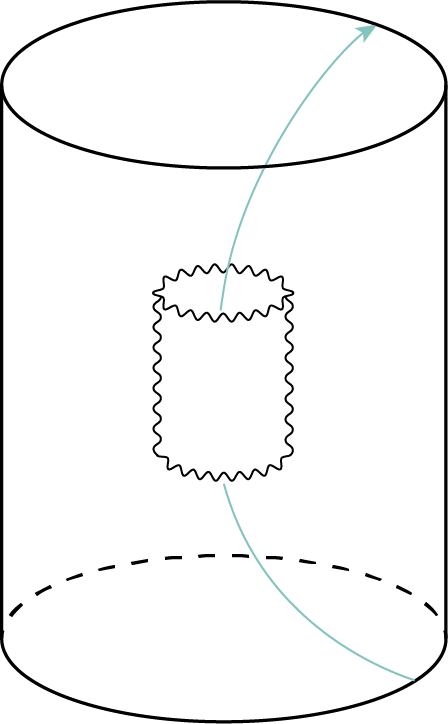}
    \caption{The region of spacetime inside this cylindrical naked singularity does not lie inside of any HKLL wedge. It is possible, however, for a rocket ship to enter the region behind the singularity and come back out to the boundary.}
    \label{fig:cylindrical-sing}
\end{figure}

\paragraph{Efficiently learning the causal wedge.} 
The starting point for our argument for (b) is the learning algorithm of~\cite{HuaChe22}, which we review  below, together with the assumption referenced above that allows us to bridge the gap between~\cite{HuaChe22}'s work and our Theorem~\ref{thm:learning}. 

Given an $n$-qubit system with $n$ large, the authors of \cite{HuaChe22} showed that there exists an efficient learning algorithm that can learn a \emph{bounded degree} representation of any unknown operator ${\cal O}^{(\mathrm{unk})}$. The bounded degree representation ${\cal O}^{(k)}$ takes the form of
\begin{align}\label{eqn:bounded-degree-operator}
    {\cal O}^{(k)} = \sum_{P: |P| \leq k } \alpha_P P~,
\end{align}
where $P$ is a Pauli operator given by a string of Pauli operators $X$, $Y$, $Z$, or $I$ at each site, and $k$ is a constant (independent of $n$). By $|P|$ we mean the number of qubits that are acted on non-trivially by $P$. A bounded-degree operator is then just a sum of individual Pauli operators with $|P|$ not scaling with $n$ in any way. 

The algorithm of \cite{HuaChe22} works by forming a guess for the coefficients of the unknown operator ${\cal O}^{(\mathrm{unk})}$ in the Pauli basis, using a polynomial or subexponential number of measurements involving the operator. The authors show that this guess ${\cal O}^{(k)}$ for the operator reproduces the correct expectation values in certain distributions of quantum states. Quantitatively:

\begin{restatable}[Theorem 13 of~\cite{HuaChe22}]{theorem}{hcp}\label{thm:HCP13} 
There exists an algorithm that, for any $n,\varepsilon, \delta >0$, any unknown $n$-qubit observable ${\cal O}^{(\mathrm{unk})}$ and any $n$-qubit state distribution $\mathcal{D}$ that is invariant under single qubit gates, given training data $\lbrace \rho_{\ell}, \mathrm{tr} \left({\cal O}^{(\mathrm{unk})}\rho_{\ell}\right) \rbrace_{\ell=1}^M$ of size 
\begin{align}
    M = \log \left( \frac{n}{\delta}\right) 2^{O\left(\log \left(\frac{1}{\varepsilon}\right) \log (n) \right)}~,
\end{align}
the algorithm is of $\poly(M)$ complexity and can learn a function $h(\rho)$ that achieves a prediction error
\begin{align}\label{eqn:HCP-prediction-error}
    \avgover{\rho\,\gets \mathcal{D}}  \left|h(\rho) - \mathrm{tr} \left({\cal O}^{(\mathrm{unk})} \rho \right)\right|^2 \leq \varepsilon \left |\left| {\cal O}^{(\mathrm{unk})} \right|\right|^2,
\end{align}
with probability at least $1-\delta$. Here, $\left |\left| {\cal O}^{(\mathrm{unk})} \right|\right|$ is the operator norm of ${\cal O}^{(\mathrm{unk})}$. 
In particular, the algorithm produces a bounded-degree operator ${\cal O}^{(k)}=\sum_{P:|P|\leq k} \alpha_P P$ and thereby learns $h(\rho):=\tr({\cal O}^{(k)}\rho)$. 
\end{restatable} 
\noindent One example of such a distribution $\mathcal{D}$ is the distribution of Haar random states, which is of particular relevance to our setting. 

In general, a bounded-degree operator is any operator of the form in \eqref{eqn:bounded-degree-operator}.
In the context of AdS/CFT, we expect that bounded-degree operators correspond to local operators in the CFT. In particular, expectation values of these operators fully determine the causal wedge of the entire boundary in the large-$N$ limit. Recall that our aim is the illustration of an algorithm that falls under the purview of Theorem~\ref{thm:learning} and can learn a $V$ that approximates the time-evolved expectation values of bounded-degree operators under $U$. Taking ${\cal O}^{(\mathrm{unk})}=U$, the algorithm in the above theorem (given the specified training data) learns an operator $U^{(k)}$ such that
\be\label{eqn:crypto-censor-HCP-thm13} 
\tr(U^{(k)}\rho)\approx \tr(U\rho)
\ee 
(on average over $\rho$). Note that this algorithm requires access to training data comprised of expectation values of the unknown operator ${\cal O}^{(\text{unk})}=U$ in a collection of states $\lbrace \rho_{\ell} \rbrace$. We expect that there exists an efficient algorithm to estimate the expectation values of $U$ in this collection of efficiently preparable states given oracle access to $U$ so that the learning algorithm in Theorem~\ref{thm:HCP13} then falls under the purview of Theorem~\ref{thm:learning}. 

In order to obtain a guarantee that this $U^{(k)}$ is such a $V$, we will make a mild assumption that \eqref{eqn:crypto-censor-HCP-thm13} holds when multiple operators are concatenated. In general, we may expect that for any operators ${\cal O}_1$, ${\cal O}_2$, ${\cal O}_3$, using the learning algorithm in the theorem above to produce ${\cal O}_1^{(k)}$, ${\cal O}_2^{(k)}$, ${\cal O}_3^{(k)}$, the following holds:
$$\tr({\cal O}_1^{(k)}{\cal O}_2^{(k)}{\cal O}_3^{(k)}\rho)\approx \tr({\cal O}_1 {\cal O}_2 {\cal O}_3\rho)~.$$
However, to bridge the gap to Theorem~\ref{thm:learning}, it suffices that the $U^{(k)}$ produced by the algorithm satisfies the following for any bounded-degree operator ${\cal O}$:
\be \label{eq:HCP-assumption}
\tr((U^{(k)})^\dagger {\cal O} U^{(k)} \rho)\approx \tr(U^\dagger {\cal O} U \rho)~,\ee
where the error scales with $1/N$ if the algorithm complexity is subexponential in $N$, which is a complexity to which Theorem~\ref{thm:learning} (and its strengthening Theorem~\ref{thm:learning-holography}) can be applied.\footnote{Note that although in previous sections we have considered polynomial complexity algorithms, we expect our results there to also apply to algorithms of subexponential complexity.}  
The upshot is that under the assumption \eqref{eq:HCP-assumption} above, there is an efficient algorithm that obeys the assumptions of Theorem~\ref{thm:learning}\footnote{Note that strictly speaking, Theorem~\ref{thm:learning} only applies to this algorithm if $V$ is approximately norm-preserving. By taking e.g.\ $O=I$ in \eqref{eq:HCP-assumption} (which holds for most $\rho$ in the Haar distribution), we expect that the guess $V$ should satisfy this condition.} and outputs a guess $V=U^{(k)}$ for the time evolution unitary $U$, which correctly predicts expectation values of local operators in the CFT, up to polynomial corrections in $1/N$.

\medskip
To summarize: we use the above (and expectations from holography) to motivate that there exists an algorithm, which we will call $\Alg_{CW}$, that can always learn time evolution within the causal wedge but falls under the purview of Theorem~\ref{thm:learning}, which guarantees that for pseudorandom time evolution, this algorithm cannot learn the time evolution of \textit{all} operators. This forms the basis of our theorem of Cryptographic Censorship. Under the assumptions above, which simply boil down to the existence of $\Alg_{CW}$, together with the assumptions of Section 1.1, we have the following theorem: 

\begin{restatable}[Cryptographic
Censorship]{theorem}{PRhor}\label{thm:pseudorandom-horizons} 
As in Assumption \ref{ass:bulkbndytimeevol}, let $U_{\mathrm{fund},N} =V_N U_{\rm bulk} V_N^{\dagger}$ where $U_{\rm bulk}$ is the time evolution operator of the bulk quantum fields acting on $\mathcal{H}$. Furthermore, assume that $U_{\mathrm{fund},N}$ is well approximated on $V_N \mathcal{H} \subseteq \Hil_N$  by a unitary $U_N$ which is typical in a gravitationally pseudorandom unitary ensemble, $\mathcal{U}_N$. Assume also (as motivated above) that there exists an efficient algorithm for causal wedge reconstruction that is (randomized) $K$-Lipschitz for some $K>0$, and which outputs (approximately) norm-preserving operators. Let $\ket{\psi_N(t)}$ be a state at any time $t$ whose complete time evolution is dual to a strongly causal geometric
bulk $(M, g)$.\footnote{When we say that $\ket{\psi_N(t)}$ is dual to a geometric state, we have in mind the definition in Ass.~\ref{ass:isometry}.}
Under the assumption that $U_N \ket{\psi_N(t)}$ is itself the image of a bulk state $\ket{\psi(t')}$ under $V_N$, then if there exists a single time $t$ such that $U_N\ket{\psi_N(t)}$ is typical in the pseudorandom state ensemble $\Psi = \mathcal{U}_N\ket{\psi_N(t)}$,\footnote{Recall that the typical states constitute a $ \lim_{N\to\infty}1-e^{-\epsilon^2 2^\kappa /\left(48(1+K)^2\right)} - e^{-\epsilon^2 2^\kappa /192} - \negl'(N)=1$ (recall $\kappa = \poly (N)$) fraction of states in $\PRSens$, and similarly for typical $U \in \PRUens$.} then there is an event horizon in $(M,g)$.
\end{restatable}

\begin{proof} 
Let $U_N$ be the pseudorandom unitary approximating $V_Ne^{-iH_{\rm bulk} \Delta t}V_N^{\dagger}$ over a time $\Delta t$. Let the fundamental quantum state of the system be $\ket{\psi_N(t)}$ at some initial time $t$, and $\ket{\psi_N(t')} = U_N\ket{\psi_N(t)}$ at some time $t' = t+ \Delta t$. Let the bulk duals be $\ket{\psi(t)}$ and $\ket{\psi(t')}$ such that $\ket{\psi_N(t)} = V_N\ket{\psi(t)}$ and $\ket{\psi_N(t')} = V_N \ket{\psi(t')}$. See Fig.~\ref{fig:thm5}. By assumption, we can find a $t$ such that $\ket{\psi_N(t')}$ is typical in the PRS ensemble $\Psi = \mathcal{U}_N \ket{\psi_N(t)}$.
\begin{figure}
    \centering
    \includegraphics[scale=0.8]{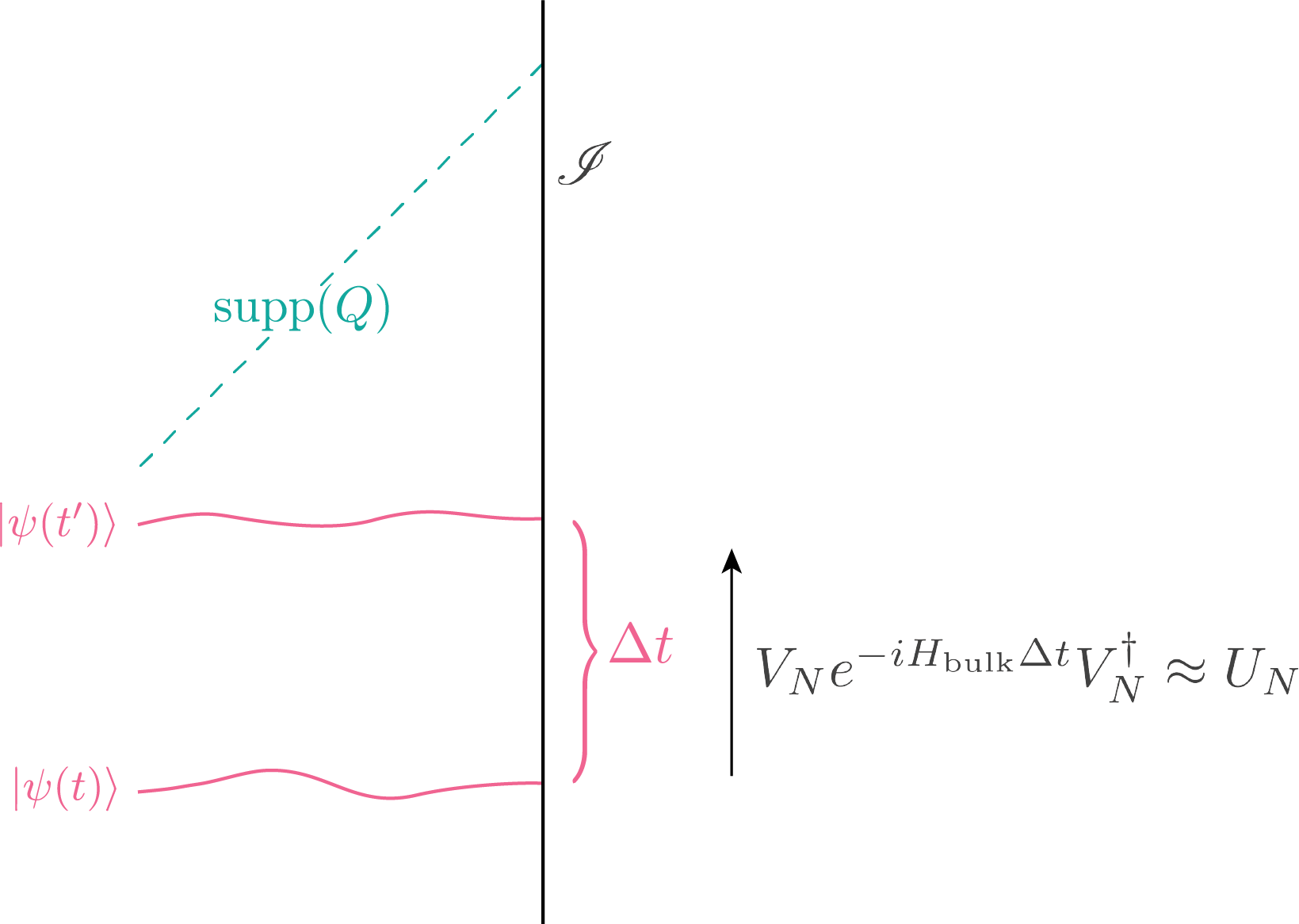}
    \caption{By assumption, the time evolution operator $U_{\rm fund, N} = e^{-i H_N \Delta t}$ on $\Hil_{\rm code}$ over a time $\Delta t$ is well approximated by some $U_N \gets \PRUens_N$, such that $\ket{\psi_N (t')}$ is typical in the PRS ensemble $\Psi = \PRUens_N \ket{\psi_N (t)}$. In this case, Theorem~\ref{thm:pseudorandom-horizons} guarantees the existence of a distinguishing operator $Q_*$ that has support outside the causal wedge: there must be a horizon (indicated by the green dotted line).}
    \label{fig:thm5}
\end{figure}

Consider a particular (possibly randomized) algorithm, $\mathcal{A}_{CW}^{U_N}$, which by the reasoning and assumptions above based on Theorem \ref{thm:HCP13} we can choose to correctly predict the values of all operators $\mathcal{O}$ that are in the causal wedge algebra $\mathcal{N}_{\rm cw}$ for any output $\widehat{U}_N \gets \Alg_{CW}^{U_N}$. By Corollary~\ref{cor:large-N-limit}, on average over the $\widehat{U}_N$ there exists a distinguishing bulk operator, $Q_* \in \mathcal{B}(\mathcal{H})$, which distinguishes between the bulk states $\ket{\psi(t')}$ and $\ket{\widehat{\psi}(t')}$ such that $U_N \ket{\psi_N(t)} = V_N \ket{\psi(t')}$ and $\widehat{U}_N\ket{\psi_N(t)} = V_N \ket{\widehat{\psi}(t')}$. Thus there must exist at least one $\widehat{U}_N$ for which the corresponding $Q_{*}$ distinguishes between $\ket{\psi(t')}$ and $\ket{\widehat{\psi}(t')}$. Note that since $\mathcal{B}(\mathcal{H})$ is the set of all bulk quantum field theory operators, we are guaranteed that $Q_*$ is itself a good semiclassical operator. Since, by construction, $\mathcal{A}_{CW}^{U}$ can correctly predict all causal wedge operators, the operator $Q_*$ must have support in a bulk spacetime region that is outside the causal wedge. Thus a horizon exists.
\end{proof}

We end this section by pointing out that we in fact get a more general result: the proof of Theorem~\ref{thm:pseudorandom-horizons} goes through for any time evolution operator $U$ that cannot be predicted (in the sense of Theorem~\ref{thm:learning-holography}) by efficient quantum algorithms. In particular, the diagnostic does not have to be pseudorandomness: our theorem holds for any $U$ that is exponentially complex to learn.

\subsection{Possible subspaces with pseudorandom time evolution}\label{sec:possible-subspaces}

We now discuss possible code subspaces which accommodate gravitationally pseudorandom dynamics. We will outline two cases of interest. First, we consider a microcanonical window about some energy above the black hole threshold. Second, we will argue that there can exist a code subspace where all the states have the same classical geometry, with low energy modes on top of it, and for which it is possible to obtain gravitational pseudorandomness within an $O(1)$ time on the boundary.

\subsubsection*{The microcanonical window} \label{sec:microcanonical}

An obvious candidate for a subspace in a holographic quantum system which admits pseudorandom time evolution is a microcanonical window about some energy $E_0 \sim N^2$ with energy width $\Delta E\sim O(N)$. The microcanonical ensemble defined by this window is known to have a stationary two-sided black hole dual~\cite{Mar18}. Here we are interested in applying our result to any state in this energy window with a semiclassical geometric bulk description. 

For holographic theories that exhibit maximal chaos, it is an interesting question how large $T$ needs to be for $U(T) = e^{-iHT}$ to be well approximated by a pseudorandom unitary ensemble. A related (but different) question is how large $T$ needs to be in a maximally chaotic system for the unitary to be well approximated by a unitary $k$-design. For $k=2$, this question was answered in \cite{RobYos16} by directly relating out-of-time-order (OTO) correlators to the question of how close a given ensemble is to a unitary $2$-design. The authors found that if the (square of the) $2$-point OTO correlator is small when averaged over all possible operators, then the time evolution is well modeled by a $2$-design.\footnote{The authors proved this statement for an $n$-qubit system and the average is over all Pauli operators on this $n$-qubit Hilbert space, which provide a basis for all operators.}

As discussed above, $k$-designs are a different but related notion from pseudorandom unitary ensembles. That said, this intuition suggests that the time $T$ at which $U(T)$ becomes well modeled by a pseudorandom unitary is at least the scrambling time for chaotic quantum systems, since this is the time when the OTO is expected to match those of a 2-design. Note that if $U(T)$ is ever going to be modeled by a pseudorandom unitary, it better be at times $T$ sub-exponential in the entropy, otherwise $U$ will not be efficiently generated. Thus, the expectation is that $U(T)$ should become pseudorandom on the microcanonical subspace for sufficiently large $T$ larger than the scrambling time. It would of course be nice to know what exactly this timescale is in certain systems. 

Applying Theorem \ref{thm:pseudorandom-horizons} to this subspace, we then see that for any state in the microcanonical window, $\ket{\psi_0}$, that is both dual to a geometric state and is also typical in the pseudorandom state ensemble, then \textit{with a probability that is negligibly (in $N$) close to one there must exist a horizon in the bulk dual to $\ket{\psi_0}$.} It is natural to ask what \textit{fraction} of the states in the microcanonical window this corresponds to. We cannot give a definitive answer to this question in the case of pseudorandom ensembles; let us just note that for Haar random unitaries and states, this would include all but a doubly-exponentially small fraction of the states. However, since it is not clear whether typical Haar random states can have semiclassical geometric duals,\footnote{ Note that if the fraction of geometric states which violate our theorem is much larger than doubly exponentially small, then there would need to be a conspiracy between having a geometry and the ability to reconstruct all operators. Since we expect spacetimes with horizons to be much more ``numerous'' than those without, this would be rather surprising.} we prefer to use pseudorandom ensembles. It is not implausible that a similar fraction of the states in some pseudorandom ensemble would obey this property, but we shall remain agnostic on this point in the absence of a proof of this statement.

\subsubsection*{A code subspace for $O(1)$ time bands on the boundary }\label{sec:O1timeband}
We will now consider a different subspace, where we restrict to $O(1)$ time bands on the boundary in order to preserve the geometry of states in our code subspace within exponentially suppressed errors.
We will start by making some assumptions that will apply to what follows, as well as in Sec.~\ref{sec:singularities-pseudorandom}. These assumptions will describe how we construct a code subspace to describe time evolution of probe fields for $O(1)$ times around geometries with naked singularities. 

Consider a CFT labeled by a parameter $N$, such that the theory in the $N\rightarrow
\infty$ limit is dual to Einstein gravity coupled to matter.\footnote{We always assume
that any effective 't Hooft coupling in the theory has been taken to infinity.} 
By this we mean that the theory participates in an asymptotically isometric code given by $V_N: \Hil \to \widehat{\Hil}_N$, where $\Hil$ is the bulk Hilbert space of Einstein gravity perturbatively coupled to matter and $\widehat{\mathcal{H}}_N$ is the finite-$N$ Hilbert space of the CFT.  
We are now going to consider a finite-dimensional subspace
$\mathcal{H}_N$ of $\widehat{\mathcal{H}}_N$ that we will assume to have certain properties. This 
will formalize assumptions often implicitly made in the study of AdS/CFT, but
since we will (1) study somewhat exotic setups, and (2) apply certain rigorous theorems, we are careful to state our assumptions.

Take $\ket{\psi^{i}_{N}}$ to be a basis of $\mathcal{H}_N$ labeled by $i$.
We will assume that there exists a one-parameter family of unitary operators
\begin{equation}
\begin{aligned}
    U_N(t): \mathcal{H}_N \rightarrow \mathcal{H}_N
\end{aligned}
\end{equation}
such that for all $t \sim O(N^0)$,
\begin{equation}
\begin{aligned}
    V_N\ket{\psi^i(t)} = U_N(t) \ket{\psi^{i}_{N}}
\end{aligned}
\end{equation}
with $\ket{\psi^i(t)} \in \Hil$ a bulk state, time-evolved using the bulk Hamiltonian. In other words, $\ket{\psi^i(t)}$ is a good bulk semiclassical state. Next, we will assume that if
$e^{-iH_Nt}$ is the true time evolution operator in the full CFT at finite $N$, then for
any fixed $t\sim O(N^{0})$ and any state $\ket{\varphi}\in \widehat{\mathcal{H}}_N$,  
\begin{equation}
\begin{aligned}
    \bra{\varphi} e^{-iH_N t}\ket{\psi^{i}_N} = \bra{\varphi} U_N(t)\ket{\psi^{i}_N} +
    O(e^{-c G^{a}})\label{eq:Uapprox}
\end{aligned}
\end{equation}
for some $c>0$, $a<0$. For $\ket{\varphi}$ in the orthogonal complement of $\mathcal{H}_N$, the first term is identically zero, so our statement implies that at large $N$, CFT time evolution for an $O(1)$ amount of time preserves our code
subspace $\mathcal{H}_N$, up to exponentially small corrections. 
This is just the usual intuition that different geometries have exponentially
small overlaps, so unless we work to exponential precision, we can consider
time evolution to preserve our code subspace. Instead of working with $e^{-iHt}$, we replace this operator by a unitary that genuinely acts on $\mathcal{H}_N$. 

Then, under the above assumptions, Theorem \ref{thm:pseudorandom-horizons} implies the following. Consider an asymptotically AdS spacetime $(M,g)$ where the boundary dynamics in an $O(1)$ time band evolve according to a unitary $U_{\rm fund}=e^{-iH_N t}$ and as in the above, can be approximated by $U_N(t)$ which is a typical gravitationally pseudorandom unitary (as defined in Def.~\ref{defn:gravpsr}). Let the state dual to $(M,g)$ be typical in a pseudorandom state ensemble. Then there must be an event horizon in $(M,g)$.

This is the version of Cryptographic Censorship that we will use in the next section.

\section{From Cryptographic to Cosmic Censorship} \label{sec:singularities-pseudorandom}

In Section \ref{sec:O1timeband}, we adapted our main result --- Cryptographic Censorship --- to code subspaces of states that have the same classical geometry and undergo gravitationally pseudorandom boundary dynamics in an $O(1)$ time band. In this section we will argue that a potential example of such a code subspace is a spacetime with a `large' naked singularity.  
Our goal is to connect Cryptographic Censorship to the more familiar notion of cosmic censorship, which is roughly the expectation that naked singularities formed from collapse must be hidden behind horizons. 

Cryptographic Censorship guarantees that singularities that are well-described by gravitationally pseudorandom dynamics must incur event horizons. The central task in connecting Cryptographic Censorship to cosmic censorship is thus understanding whether --- or when --- singularities are associated with (sufficiently) pseudorandom dynamics. If generic singularities formed from matter collapse do indeed require pseudorandom time evolution, then they cannot exist in horizonless spacetimes. This version of censorship of naked singularities is distinct from the original version of cosmic censorship: even when a singularity is associated with pseudorandom dynamics, our result does not imply that the singularity is completely hidden.  This is not a bug but a feature, since certain naked singularities do exist: e.g., the Gregory--Laflamme (GL) instability, which occurs when a black string in $D>4$ pinches off and fragments into multiple black holes~\cite{GreLaf93, GreLaf94, LehPre10}.

In the following, we take some first steps towards connecting a version of cosmic censorship to Cryptographic Censorship, culminating in Sec.~\ref{sec:QCC}, where we state a cryptographically-inspired formulation of quantum cosmic censorship. We begin by defining a mathematically precise notion of the size of a singularity, and specifically identify how large the singularity must be in the $G_{N}\rightarrow 0$ limit to accommodate gravitationally pseudorandom fundamental dynamics. This definition usefully distinguishes between a Planckian and a
classical naked singularity (and similarly for non-naked singularities of large and small sizes). This is in itself useful, as it separates naked singularities into classes that are expected  to be excluded and those that we expect may be permitted in a quantum theory of gravity (e.g., evaporating black holes or Gregory--Laflamme instabilities).\footnote{There is some extant literature dedicated to investigating the size of a singularity; see, e.g., \cite{Emp20}. These approaches differ from ours, as we are motivated by quantum computational considerations in the fundamental quantum description of the system.}
Our second step uses this quantification to provide some evidence that large classical singularities can indeed be associated with gravitationally pseudorandom dynamics and thus incur horizons by Cryptographic Censorship. Finally, we propose a Quantum Weak Cosmic Censorship Conjecture, which formally requires large naked singularities to be atypical. 

\subsection{The size of singularities} \label{sec:size-of-sings}
We are interested in an unambiguous and diffeomorphism-invariant way of measuring the extent of a naked singularity. Since a naked singularity can only be distinguished from a cloaked one via access to $\mathscr{I}$, it is clear that this must be defined asymptotically. A definition on $\mathscr{I}$ would be automatically diffeomorphism-invariant and would more readily admit a straightforward dual CFT interpretation. To find this criterion, we will leverage the definition of a naked singularity --- i.e., causal visibility to $\mathscr{I}$ --- to define what will roughly correspond to the time evolution in the CFT that is necessary, in the Schr\"odinger picture, to reconstruct near-singularity physics. 

To make this precise, we first use an extant mathematical construction to single out the singularity as a set, a nontrivial task given that the singularity is not part of the spacetime manifold. Second, we identify the time intervals on $\mathscr{I}$ that are causally-separated from this set, and finally find the minimal boundary time bands that are necessary for a full reconstruction of the singularity. The spacetime volume of these time bands in a fixed conformal frame is essentially our definition of the size of the singularity in that frame. There is thus a conformal frame ambiguity in our definition of singularity size. This, however, is immaterial for our desired application: we shall only care about scaling of the singularity size with $G_N$, which is independent of the choice of conformal frame. Using this scaling, we separate the singularities into three categories: classical, semi-Planckian, and Planckian. The first type survives the $G_N\rightarrow 0$ limit, while the other two do not. These latter two sizes are distinguished via the power of $G_N$ by which they vanish. Let us now turn to the details, starting with an identification of the set comprising the singularity. 

\paragraph{The Singularity Set.} The primary difficulty in this task is that singularities are not locations in some geodesically-incomplete
manifold $M$ but are rather associated to the endpoints of incomplete inextendible geodesics in $M$. However, as shown by~\cite{Ger68}, it is possible to treat singularities as if they
are made up of a set of points $\Gamma$, although this set does not lie in $M$.\footnote{$\Gamma$ also comes with a natural choice of
topology \cite{Ger68}.}
We will briefly review this construction, as it is instrumental to our classification of singularities. The idea is to build an equivalence class of geodesics associated to an endpoint. Consider a geodesic $\gamma_{p,v}$ in $M$ fired from a starting point $p\in M$ along a non-vanishing tangent vector $v^{a}$ at $p$. We will continue $\gamma_{p,v}$ so it is inextendible in the direction of $v^a$ (with starting point $p$).
We then identify two distinct geodesics $ \gamma_{p,v} \sim \gamma_{\tilde{p},\tilde{v}}$ if $\gamma_{p, v}$
enters and remains inside every tubular neighborhood of
$\gamma_{\tilde{p},\tilde{v}}$. Explicitly, if $\lambda \in [0, \lambda_0)$ is an affine parameter for $\gamma_{p, v}$, then for every given tubular neighborhood $U$ of $\gamma_{\tilde{p},\tilde{v}}$, there exist a $0\leq \hat{\lambda}<\lambda_0$ such that $\gamma_{p, v}(\lambda) \in U$ for all $\lambda \in (\hat{\lambda}, \lambda_0)$. See Fig.~\ref{fig:Tubular-region}. 
\begin{figure}
    \centering
    \includegraphics[scale=0.9]{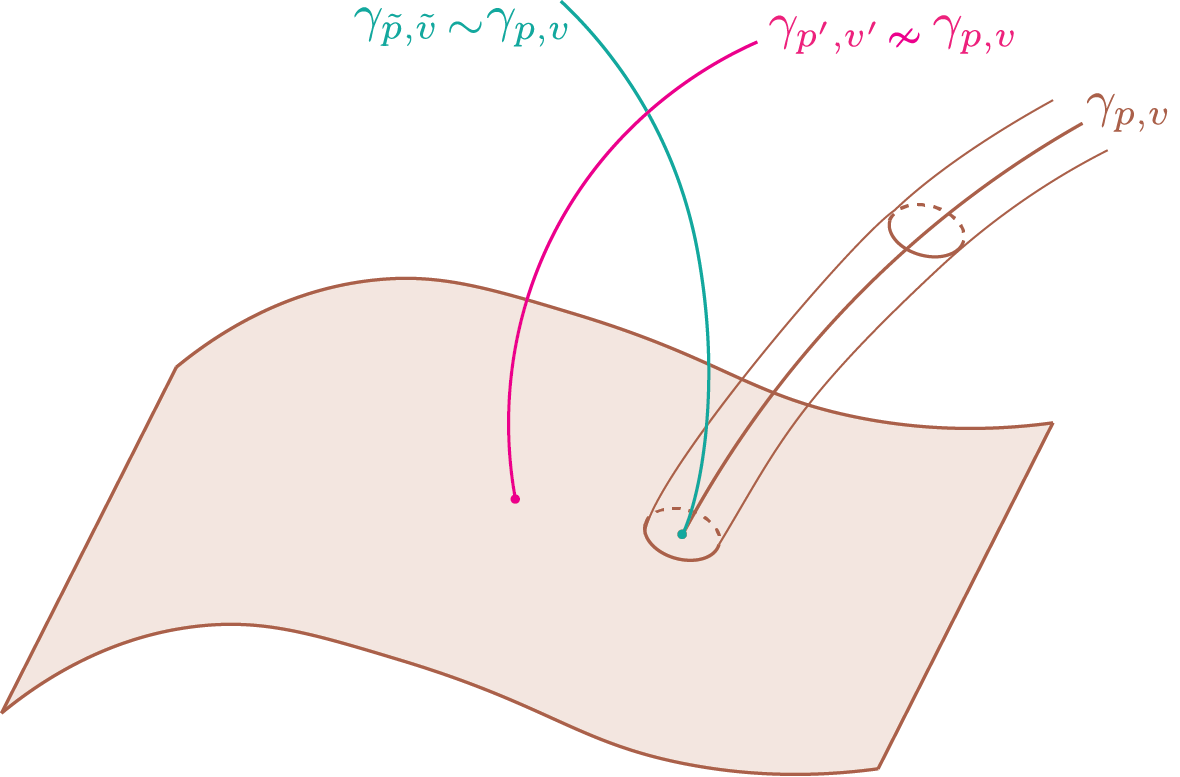}
    \caption{Example spacetime where the colored sheet represents a singularity.
    The two geodesics $\gamma_{p, v}$ and $\gamma_{ \tilde{p}, \tilde{v} }$ by definition end on the same singular point, since $\gamma_{\tilde{p}, \tilde{v}}$ must enter any tubular neighborhood of $\gamma_{p, v}$, and vice versa.}
    \label{fig:Tubular-region}
\end{figure}
The equivalence classes
under $\sim$
then form the points in $\Gamma$, and we say that $\gamma_{p, v},
\gamma_{\tilde{p}, \tilde{v}}$ are geodesics that end on the singular point $[\gamma_{p, v}] \in
\Gamma$. Importantly, if $p_s \in
\Gamma$,  the future $I^+(p_s)$ and past $I^-(p_s)$ of $p_s$ in
$M$ is well-defined even though $p_s$ is not in $M$. Namely, we say that $p \in I^+(p_s)$ if
there is a future-directed timelike curve
$\gamma$ through $p$ and with past ``endpoint'' on $p_s$, where
the latter means that $\gamma$, sufficiently far in its past direction, enters
and remains within any tubular neighborhood of any representative geodesic in $p_s$.
Armed with a set identified as the singularity, we now extract the part of the singularity that is naked with respect to a
given conformal boundary $\mathscr{I}$ as follows:
\begin{equation}
\begin{aligned}
    \Gamma_{n(\mathscr{I})} = \{p_s \in \Gamma | I^+(p_s) \cap \mathscr{I}\neq
    \varnothing \text{ and } I^-(p_s) \cap \mathscr{I}\neq
    \varnothing\}.
\end{aligned}
\end{equation}
If this is nonempty, we say that $\Gamma$ is naked with respect to
$\mathscr{I}$.

\paragraph{The Singularity Size.} We now proceed to classify the size of any (possibly improper) subset of a singularity. 
Let $\hat{\Gamma}$ be a
subset of $\Gamma$; below we will mostly be
interested in taking $\hat{\Gamma}=\Gamma_{n(\mathscr{I})}$, but it will be useful
to take our definition to be more general. Consider the sets
\begin{equation}\label{eq:Tdef}
\begin{aligned}
    T_{+} &= \{p \in \mathscr{I} | ~\exists \text{ an achronal past-directed    null geodesic from } p \text{ to } \hat{\Gamma} \}~, \\
    T_{-} &= \{p \in \mathscr{I} | ~\exists\text{ an achronal future-directed
    null geodesic from } p \text{ to } \hat{\Gamma} \}~. 
\end{aligned}
\end{equation}
Fig.~\ref{fig:sings-examples} illustrates some examples of these sets for both naked singularities and ordinary singularities behind horizons. 
\begin{figure}
   \centering
   \subcaptionbox{ \label{fig:timesing}}[.25\textwidth]{\includegraphics[scale=0.65]{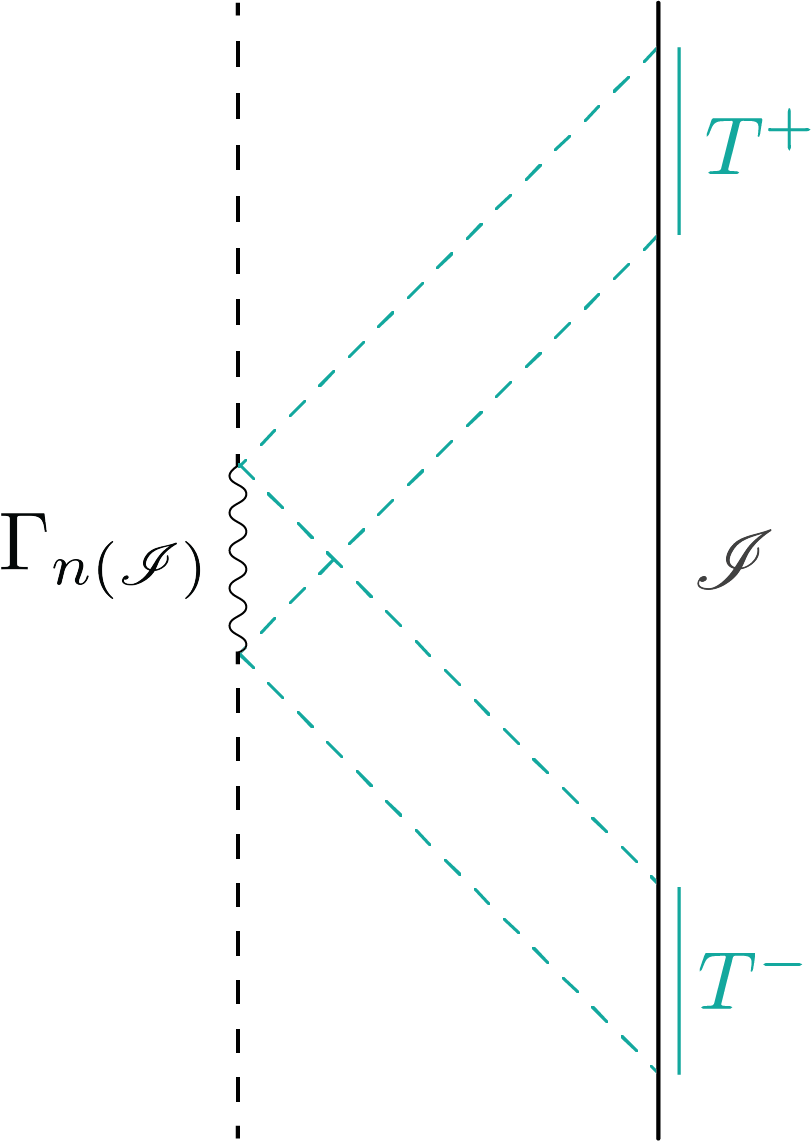}}
   \hfill
   \subcaptionbox{ \label{fig:spacesing}}[.25\textwidth]{\includegraphics[scale=0.65]{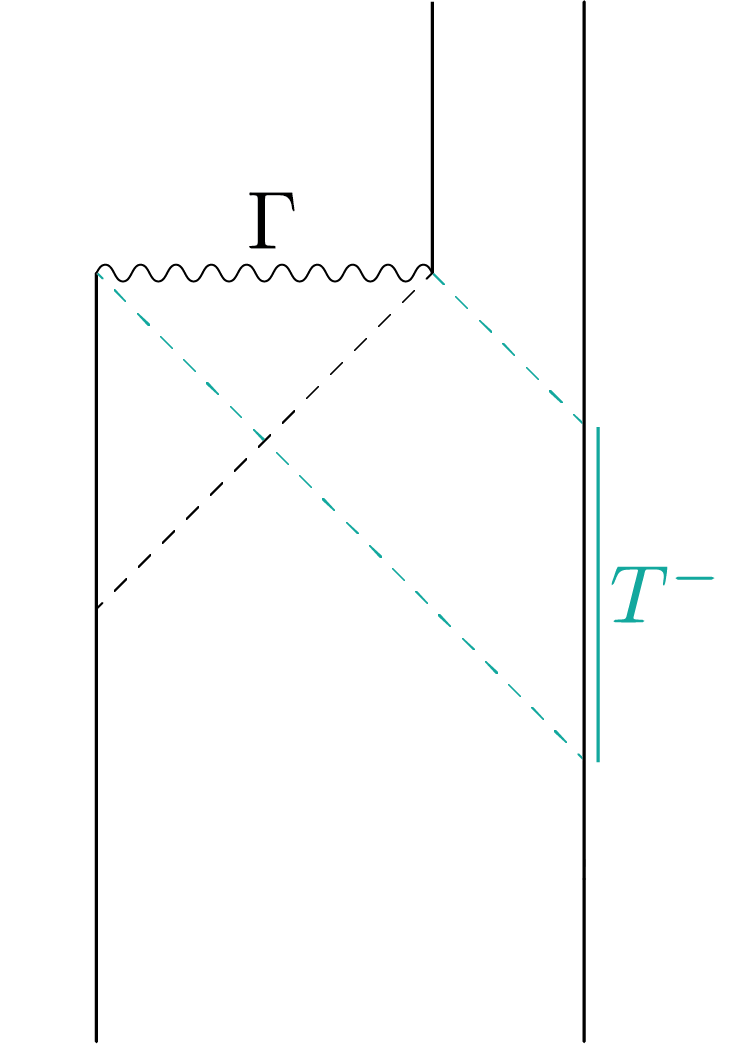}}
   \hfill
  \subcaptionbox{ \label{RN-AdS}}[.25\textwidth]{\includegraphics[scale=0.65]{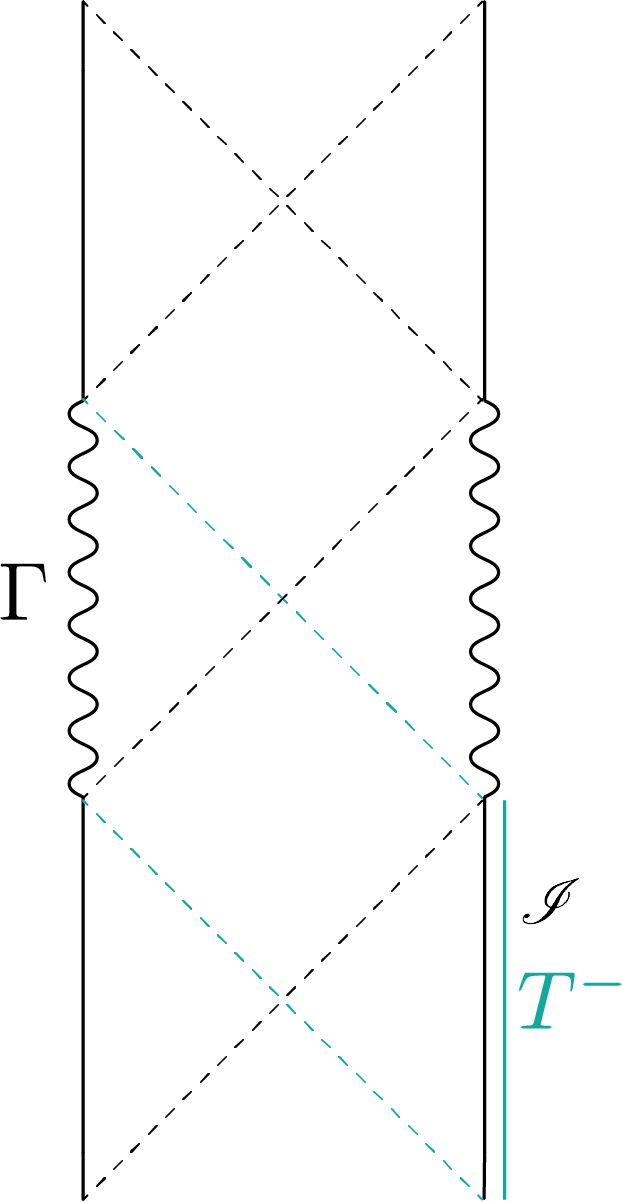}}
   \caption{Some examples of spacetimes with (naked) singularities. In (a) the singularity is timelike and classically naked. Panel (b) depicts the evaporating black hole. The non-naked part of the singularity is classical and behind a horizon. For the naked part, see Fig.~\ref{fig:endpoint-evap-non-example}. Panel (c) is Reissner--Nordstr\"om-AdS; here $|\Gamma|_{\mathscr{I}} = |T_-|$; note, however, that $\Gamma_{n(\mathscr{I})} = \varnothing$, as $I^{\pm} \cap \mathscr{I} = \varnothing$. 
   \label{fig:sings-examples}}
\end{figure}

We will now define the size of $\hat{\Gamma}$ according to properties of $T_+$, $T_-$. 
We assume that $(M, g)$ is conformal to the Einstein static
universe on each conformal boundary, and we choose the unique conformal frame where the
boundary metric is identical to the Einstein static universe in which the spatial
spheres have radius equal to the AdS radius.\footnote{For asymptotically locally AdS spacetimes where $\mathscr{I}$ has another conformal class, we can also define $|\hat{\Gamma}|$, once we pick some preferred conformal frame.} This fixes a unique Lorentzian
metric on $\mathscr{I}$, and when $T_+$ is nonempty, we define $|T_+|$ as
the absolute value of the spacetime volume of $T_+$ in this metric; similarly for
$T_-$. As examples, if $T_+$ is a timeslice or a union of a finite number of
points, then $|T_+|=0$, since these sets are measure zero.
We can now define the size $|\hat{\Gamma}|_{\mathscr{I}}$ of $\hat{\Gamma}$ with
respect to $\mathscr{I}$ as follows:
first, if $T_+, T_-$ are both empty, then $|\hat{\Gamma}|_{\mathscr{I}}=0$. Second, if only
one of $T_{+}, T_-$ is nonempty, say $T_+$, then we define
$|\hat{\Gamma}|_{\mathscr{I}}=|T_+|$. Third, if both $T_{\pm}$ are nonempty, we define
$|\hat{\Gamma}|_{\mathscr{I}} = \min(|T_+|, |T_-|)$. This finally leads us to
our definition of classical and Planckian singularities, as well as an
in-between regime, which we term semi-Planckian:

\begin{defn} \label{defn:singsize}
    Let $(M, g)$ be a geodesically incomplete spacetime of dimension $D=d+1$ with a singularity
    $\hat{\Gamma}$, and let \begin{equation}
    \begin{aligned}        |\hat{\Gamma}|_{\mathscr{I}} \sim O( \ell_{\rm Pl}^{a}).
    \end{aligned}
    \end{equation}
    If $a\leq0$, we say that $\hat{\Gamma}$ is a classical singularity, while if
    $ 0<a<d-1$, we say that $\hat{\Gamma}$ is a semi-Planckian singularity. Otherwise, it is a Planckian singularity. 
\end{defn}

\paragraph{Singularity Size from CFT Time Evolution.} If $\ket{\psi(t)}$ is the state whose complete past and future history is dual to $(M,g)$, then without access to the time evolution operator, we can only reconstruct operators at points in spacetime that are not causally related to the boundary time slice on which $\ket{\psi(t)}$ lives --- i.e., the Wheeler--de Witt patch $W_{t}$.\footnote{This assumes for simplicity that there is only one asymptotic boundary, and that the state of the CFT on that boundary is pure --- so the entanglement wedge of $\ket{\psi}$ is the entire bulk. If there are multiple boundaries or the state on one boundary is mixed for a different reason, we simply take the intersection of $W_{t}$ with the entanglement wedge of $\ket{\psi}$.} Let us now specialize to a spacetime with a naked singularity $\Gamma$. For some fixed $t$, generally only a (possibly empty) subset of $\Gamma$ will be contained in $W_{t}$: to learn everything about $\Gamma$ (for example by reconstruction of all near-singularity correlators), a certain amount of boundary time evolution is necessary. More precisely, if $T_{-}$ is nonempty, and we are evolving forwards from the past, then the past-most boundary of $T_{-}$ corresponds to the earliest time that some part of $\Gamma$ is encoded in the Wheeler--de Witt patch $W_{t}$.  As we evolve forward through $T_-$, new pieces of
$\Gamma$ enter $W_t$ and are encoded in the time-evolved state. Finally,
once we evolve past the futuremost boundary of $T_-$, no new parts of
$\Gamma$ enter $W_t$ that we have not already seen, and so
further forward time evolution reveals nothing new about the singularity.
\begin{figure}
    \centering
    \includegraphics[scale=1.1]{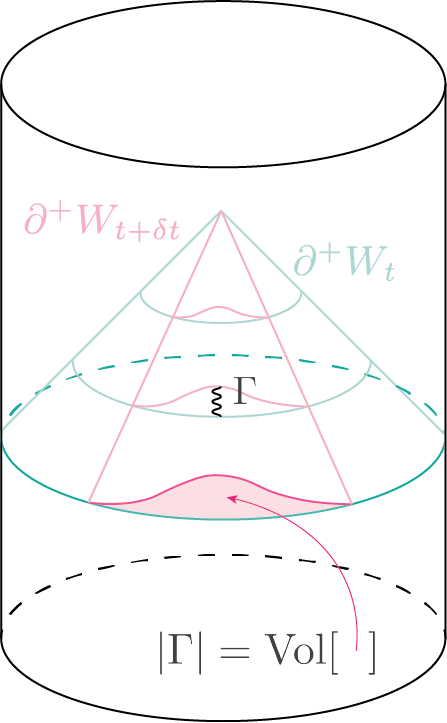}
    \caption{To get all information about the small singularity $\Gamma$, we only need to evolve the state forwards in a small spatial subregion of the boundary Cauchy slice, such that the Wheeler--de Witt patch now includes the singularity. Note that, as our artistic capabilities are limited, we here only indicate the future boundaries of the Wheeler--de Witt patches for times $t$ and $t+\delta t$.}
    \label{fig:small-sing}
\end{figure}
Thus, everything about the naked singularity can be learned by carrying out CFT
time evolution through the band $T_-$ (plus bulk reconstruction not using the
time-evolution operator). For a singularity
where $T_+$ is also nonempty, everything about the
singularity could alternatively have been learned by evolving through $T_+$. So
as advertised, $|\Gamma|_{\mathscr{I}}=\min\{|T_-|, |T_+|\}$ measures the minimal amount of CFT time evolution
required to learn everything about the naked singularity. Note that $|\Gamma|_{\mathscr{I}}$ has units of spacetime
volume rather than time. This is a feature: if
all information about the singularity is available upon evolution of the state forwards in a small spatial subregion of a
boundary Cauchy slice, then this should be
accounted for in the smallness of a singularity. See Fig.~\ref{fig:small-sing}. Note also that if $|\Gamma|_{\mathscr{I}}$ is large, this always is caused by a large amount of time evolution, since timeslices have bounded volume.

Let us now make a few remarks on our definition, starting with a question. Under what circumstances do Planckian or semi-Planckian singularities exist?   The most natural case is that we have $\lim_{\ell_{\rm Pl} \rightarrow
0}|\hat{\Gamma}|_{\mathscr{I}}=0$ because the CFT state acquires all information
about the singularity in a CFT timestep that shrinks to zero as $\ell_{\rm Pl}\rightarrow
0$. An example would be the singularity in the evaporating black hole,
illustrated in Fig.~\ref{fig:endpoint-evap-non-example}, or a naked timelike singularity $\hat{\Gamma}$ lying behind another naked
timelike singularity $\tilde{\Gamma}$, like in Fig.~\ref{fig:hidden-sing}. 
    \begin{figure}
        \centering
        \subcaptionbox{ \label{fig:endpoint-evap-non-example}}[.4\textwidth]{\includegraphics[scale=0.75]{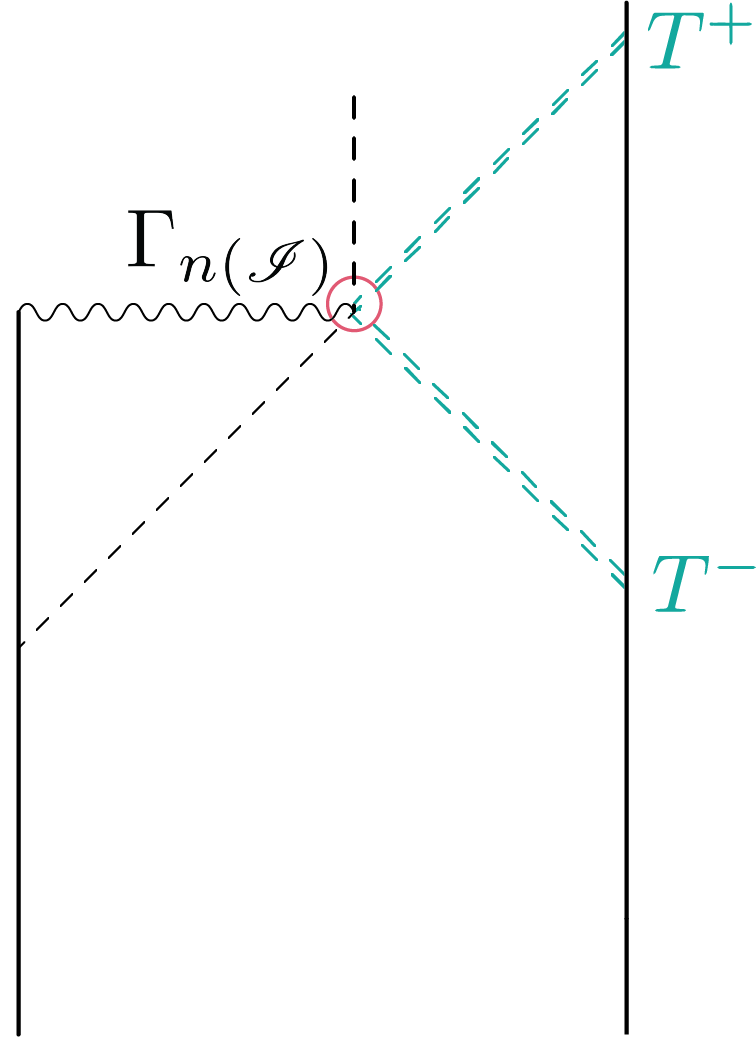}}
        \hspace{1cm}
        \subcaptionbox{ \label{fig:hidden-sing}}[.4\textwidth]{\includegraphics[scale=0.8]{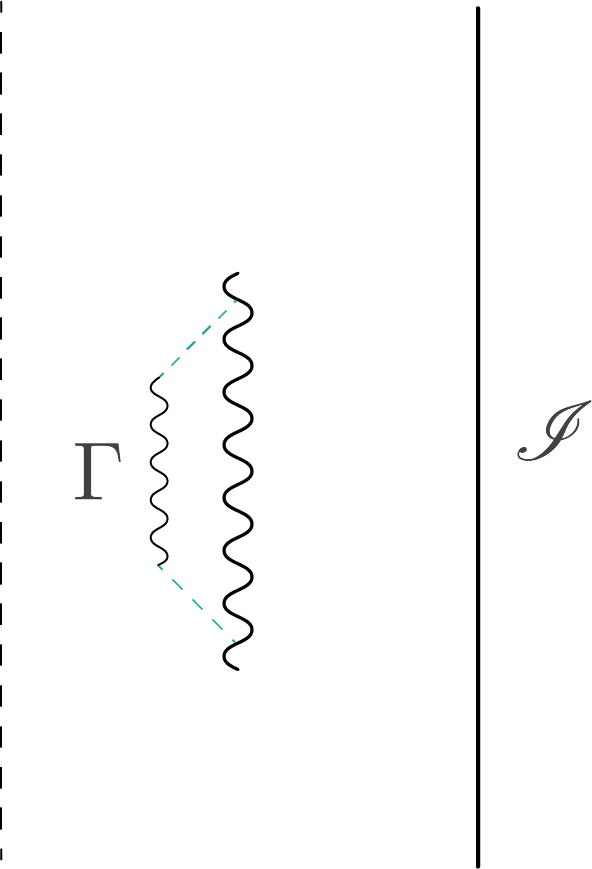}}
        \caption{Two examples of $\lim_{\ell_{\rm Pl} \rightarrow
0}|\hat{\Gamma}|_{\mathscr{I}}=0$: (a) For the naked part of the singularity at the endpoint of black hole evaporation, we expect that  $T_{\pm}$ are just timeslices in the $G_N \to 0$ limit. (b) A `hidden' naked timelike singularity; here, $T_{\pm} = \varnothing$.}
    \end{figure}
In this latter case,
$\hat{\Gamma}$ is such that $T_+ =
T_-= \varnothing$, since all causal geodesics from the
boundary to $\hat{\Gamma}$ are necessarily chronal. This latter example shows that
we probably should not assign much meaning to the size of distinct connected
components,  and that we instead just should look at the size of the full naked
singularity $\Gamma_{n(\mathscr{I})}$, which by definition contains all connected components.

\subsection{Accommodating gravitational pseudorandomness} \label{sec:counting}

Our motivation for classifying a classical singularity via the scaling $a\leq 0$ is simple: such singularities are visible for a boundary time that diverges or remains $O(1)$ in the $G_N \rightarrow 0$ limit. For the remaining range of $a$, we have chosen to separate semi-Planckian singularities (with $0<a<d-1$) from Planckian singularities (with $a\geq d-1$). In both cases, $|\hat{\Gamma}|_{\mathscr{I}}\rightarrow 0$ as $G_N \rightarrow 0$, but as we will show below, singularities with $a<d-1$ are compatible with an EFT code subspace which is sufficiently large that it can accommodate gravitational pseudorandomness, while singularities with $a\geq d-1$ are not. 

We will now show that singularities with $a<d-1$ are compatible with gravitational pseudorandomness. Consider an asymptotically AdS spacetime with a timelike naked
singularity $\Gamma=\Gamma_{n(\mathscr{I})}$. We will work with large (but possibly finite) $N$ and demand only that effective field theory as defined by a bulk UV cutoff be valid in the asymptotic region. At finite $N$ we will consider the failure of EFT due to large UV corrections as the diagnostic of a singularity.

We will decompose our bulk state into modes with energy below the cutoff, and we will focus on modes localized near the boundary, some of which fall into the naked singularity. By unitarity, there must also be outgoing occupied modes as well, which can be traced back to the singularity. However, the relation between the modes falling into the
singularity and the modes coming out is dependent on UV physics: it is inaccessible
within low energy EFT. So while the low energy modes in the asymptotic region could be semiclassical as $G_N\rightarrow 0$, the time evolution operator that relates them is not known in the bulk EFT --- it cannot be expressed in
terms of the low-energy operators. See Fig.~\ref{fig:toy-model}.
\begin{figure}
    \centering
    \includegraphics[scale=0.75]{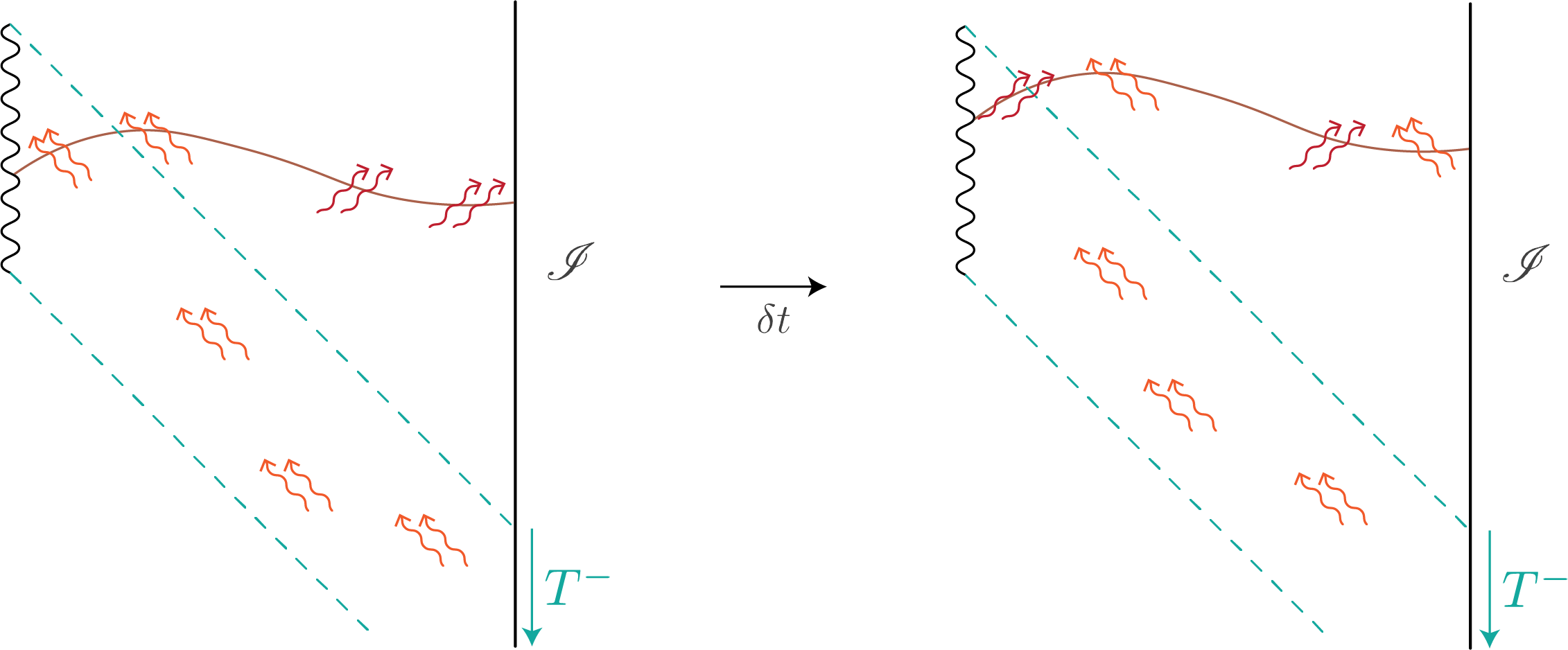}
    \caption{From the low-energy perspective, there is an unknown, UV-sensitive boundary condition assigned at the singularity.}
    \label{fig:toy-model}
\end{figure}
From the low-energy perspective, there is a  boundary condition assigned at the
singularity whose precise nature is unknown.\footnote{Since quantum gravity is expected to lack global
symmetries \cite{Zel76, AbbWis89, KalLin95b}, any approximate symmetries emergent in the semiclassical EFT should be badly violated by this boundary condition.} 

In what follows, we will be agnostic about the UV-sensitive boundary condition that implements pseudorandomness in the fundamental theory. Our approach instead will be to show that once the singularity is of classical or semi-Planckian size, then there are enough EFT modes in a code subspace in which all of the states have the same leading order geometry to accommodate gravitational pseudorandomness. On the other hand, if the singularity is Planckian, we show that we do not expect this. As our analysis is asymptotic, we do not require a sharply defined geometry near the singularity.

We now give a parametric estimate of the number of modes within semiclassical EFT that
can potentially get scrambled by the naked timelike singularity. That is, we want to know the number of EFT modes that are
affected by the highly UV-sensitive boundary condition. Assume that
$T_+, T_-$ are nonempty, and assume without loss of generality $|\Gamma_{n(\mathscr{I})}|_{\mathscr{I}} = |T_-|$. 
Near
$\mathscr{I}$, the leading part of the metric can be written
\begin{equation}\label{eq:nbmetric}
\begin{aligned}
d s^2 = \frac{ -d t^2 + d \rho^2 + \sin^2 \rho d \Omega^2 }{ \cos^2 \rho
    } + \ldots ~,
\end{aligned}
\end{equation}
where we work in units where the AdS radius is unity, and where $\mathscr{I}$
lies at $\rho = \frac{ \pi }{ 2 }$. For pedagogical simplicity we take the matter in our theory to be a conformal massless
scalar and assume that the system is spherically symmetric, so that $T_{-}$ can be taken to be a time band $[-\Delta t, 0] \times
S^{d-1}$ without loss of generality.  We can now count roughly how many ingoing single-particle modes with energy below the cutoff can be sent from $T_{-}$
into the singularity. As is clear from Fig.~\ref{fig:mode-counting}, this simply corresponds to the number of orthogonal modes
in the strip given by $\rho \in \left[\frac{ \pi }{ 2 }-\Delta t,
\frac{ \pi }{ 2 }\right]$ at $t=0$.\footnote{Greybody factors cause some modes to scatter off background curvature before they get close to the singularity, but this should not change the parametric counting.} 
\begin{figure}
    \centering
    \includegraphics[scale=0.8]{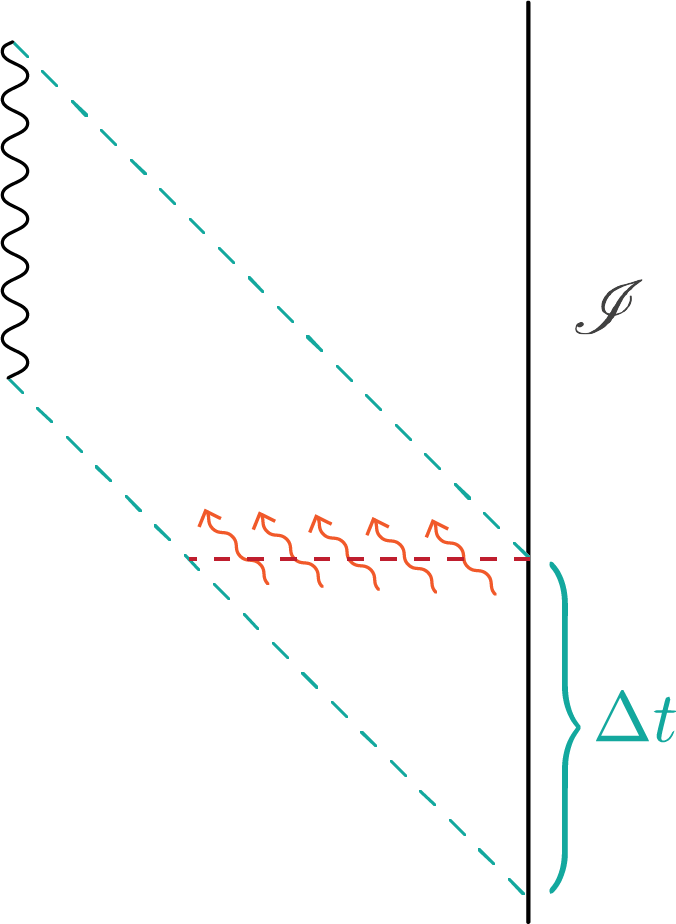}
    \caption{To count the ingoing modes sent from $T_-$, which here is the time band $[-\Delta t, 0]\times S^{d-1}$, we can simple count the number of modes in the strip indicated by the red dotted line.}
    \label{fig:mode-counting}
\end{figure}
As long as $\Delta t$ is less than roughly $G_N M \sim O(1)$ where $M$ is the spacetime mass, we are always restricted to an asymptotic range of $\rho$ where the metric is approximately pure AdS, so the number of modes can be counted using the pure AdS metric. We always take $G_{N}$ small with $G_N M = \text{fixed}$, and we are interested in a parametric counting that can only change if $\Delta t \rightarrow 0$ sufficiently fast in $G_{N}$ 
so we can without loss of generality restrict to this regime. Assume now that $\tilde{u}_{i}$ is a collection of orthonormal flat-space positive
frequency ingoing wavepackets with mean frequency $\omega_i$ in the metric $-d t^2 + d \rho^2
+ \sin^2 \rho d \Omega^2$; furthermore, we will assume that they are well localized
within $\rho \in (\frac{ \pi }{ 2 }-\Delta t, \frac{ \pi }{ 2 })$ at some central
$\rho=\rho_0$. These have total energy approximately equal to $\omega_i$, as long as 
\begin{equation}
\begin{aligned}
\omega_i \gg \Delta \omega \geq \frac{ 1 }{ \Delta t }~,
\end{aligned}
\end{equation}
where $\Delta \omega$ is the frequency spread of each wavepacket.
Now, since we have a conformally coupled scalar,
$u_{i} = (\cos\rho)^{d-1}\tilde{u}_{i}$ is also an
orthonormal set of positive frequency mode functions in AdS.\footnote{The scaling dimension of the scalar is $\frac{d-1}{2}$, and the conformal factor is $\cos^2 \rho$.} Each
wavepacket has a total energy approximately equal to $\omega_i$, also in AdS.\footnote{There might be an overall conformal anomaly, so
that $H_{\rm AdS} = H_{\rm Mink}+H_{\rm anomaly}$, but this is irrelevant for our purposes.}
Imposing the cutoff $\omega_i \leq \epsilon^{-1}$,\footnote{This corresponds to
a proper length cutoff in AdS of $\mathcal{\epsilon}_{\rm proper} \sim
\mathcal{\epsilon}/\cos \rho_0$. We do not scale $\rho_0$ with $\ell_{\rm Pl}$, so $\epsilon, \epsilon_{\rm proper}$ scale the same way.} we want to count the
orthogonal wavepackets. This is equivalent to counting  orthogonal flat-space modes; this is of the order
\begin{equation}
\begin{aligned}
    \frac{ \Delta t \cdot \text{Vol}(S^{d-1}) }{ \epsilon^{d} }~.
\end{aligned}
\end{equation}
The numerator is $|\Gamma_{n(\mathscr{I})}|_{\mathscr{I}}$, and in the absence of spherical symmetry it is clear that $|\Gamma_{n(\mathscr{I})}|_{\mathscr{I}}$ is the correct replacement for $\Delta t \cdot
\text{Vol}(S^{d-1})$. So for a general naked timelike singularity, we find that an order
\begin{equation}
\begin{aligned}
    \frac{ |\Gamma_{n(\mathscr{I})}|_{\mathscr{I}} }{ \epsilon^{d} }
\end{aligned}
\end{equation}
number of distinct EFT modes can propagate into the singularity from $T_{-}$. Next, we want to take the UV cutoff to scale as $\epsilon=\ell_{\rm
Pl}^{1-\delta}$ for some $\delta \in (0, 1)$ so that our modes are
parametrically below the Planck scale as $\ell_{\rm Pl} \rightarrow 0$. To be
cautious, we demand that $\delta > 1/d$, which ensures that the na\"ive QFT entropy does not exceed the Bekenstein--Hawking entropy. To see this, note that if
$V$ is the spatial volume occupied by our wavepackets in AdS, then
we need
\begin{equation}
\begin{aligned}
    \frac{ V }{ \epsilon_{\rm proper}^{d} } = \cos^{d}\rho_0 \frac{
        V }{
    \epsilon^{d}} \lesssim \frac{ A }{
    \ell_{\rm Pl}^{d-1} }~.
\end{aligned}
\end{equation}
Remembering that volume and surface area of a large shell near the boundary
of AdS scale the same way as we approach the conformal boundary, we thus need
\begin{equation}
\begin{aligned}
    \cos^{d}\rho_0 < O(1)\frac{ \epsilon^{d} }{ \ell_{\rm Pl}^{d-1} } =
    O(1) \ell_{\rm Pl}^{1-d\delta}.
\end{aligned}
\end{equation}
If we take $\delta < \frac{ 1 }{ d }$, then wavepackets with mean frequency
near the cutoff would give a na\"ive QFT entropy exceeding the Bekenstein--Hawking entropy for any $\rho_0 < \pi/2$. For
$\delta = \frac{ 1 }{ d }$ we get an $O(1)$ bound on $\rho_0$ that we want to avoid, and
so we require $\delta > \frac{ 1 }{ d }$. Thus, assuming finally that
$|\Gamma_{n(\mathscr{I})}|_{\mathscr{I}} \sim O(\ell_{\rm Pl}^a)$, we
get that the number of distinct semiclassical modes that can probe the naked
singularity during time evolution through $T_{-}$ is 
\begin{equation}
\begin{aligned}
\frac{ |\Gamma_{n(\mathscr{I})}|_{\mathscr{I}} }{ \epsilon^{d} } \sim \ell_{\rm
    Pl}^{a - d- d\delta}.
\end{aligned}
\end{equation}
If $a<d-1$, there is always some $\delta > 1/d$ so that this diverges in the
$\ell_{\rm Pl}\rightarrow 0$ limit. Hence, in the regime $0<a<d-1$, even when the singularity shrinks to
zero size as $\ell_{\rm Pl} \rightarrow 0$, the singularity is in principle compatible with influencing an infinite number of EFT modes in the $G_N \rightarrow 0$ limit: it is compatible with a gravitational amount of pseudorandomness in the time evolution. 

Let us now comment on a few
well known cases of naked singularities, some of which were briefly mentioned in Sec.~\ref{sec:intro}: the singularity at the endpoint of
black hole evaporation, the Gregory--Laflamme (GL) instability \cite{GreLaf93, GreLaf94}, which occurs as the pinch-off endpoint of the evolution of higher dimensional extended black strings and other black objects, and Choptuik's
critical collapse \cite{Cho92,Chr94} resulting from a collapsing massless scalar
 field which is tuned to lie exactly at the threshold of black hole formation. In the case of the endpoint of black hole evaporation, the naked part of the singularity
 is not extended in the $G_N \rightarrow 0$ limit. See Fig.~\ref{fig:endpoint-evap-non-example}. 
 Thus the naked part of the singularity is measure zero in this limit, so this naked singularity is either Planckian or semi-Planckian. In the
 cases of the Gregory--Laflamme instability and critical collapse, it is impossible to evolve past the pinch-off using classical gravity, so in the absence of a quantum gravitational description of the system, we cannot ascertain the nature of the naked singularity.
 However, it is quite plausible that these singularities end up being Planckian or semi-Planckian.
 The critical collapse singularity is in a sense a
 zero mass black hole, which within the resolution of semiclassical gravity is
 indistinguishable from a Planck mass black hole. So if we could carry out critical collapse
 at finite $G_N$ in full nonperturbative quantum gravity, and then afterwards take the
 $G_N \rightarrow 0$ limit, it is reasonable to expect that the singularity becomes a
 Planckian or semi-Planckian singularity. The Gregory--Laflamme singularity could plausibly behave similarly: it is highly analogous to the fluid pinch-off phenomenon, where only a minimal amount of nonperturbative evolution is necessary to re-enter the EFT regime.   It
 is tempting to speculate that naked singularities in a holographic theory of quantum gravity can never be classical, and possibly not even semi-Planckian. We will formalize a conjecture to this effect in the next section.
 
Finally, let us briefly tie this discussion back to our primary theorem of Cryptographic Censorship. The fact that classical singularities are always compatible (by our definition) with gravitational pseudorandomness provides additional motivation for the intuition that the dynamics of strongly gravitating systems in the $G_{N}\rightarrow 0$ limit may admit a description in terms of pseudorandom unitaries. Under this supposition, Cryptographic Censorship guarantees that any such systems --- including those with naked singularities --- must also have event horizons. We leave the quantification of the subset of the singularity that is necessarily behind a horizon to future work. 

We wish to make clear that we do not claim that all classical singularities must act as sources of pseudorandomness; a clear counterexample is given by orbifold singularities in compact dimensions. It is the ones that weak cosmic censorship is supposed to rule out --- naked classical singularities resulting from the dynamical evolution of generic, smooth initial data and matter collapse --- that we expect to have, to a large extent, gravitationally pseudorandom time evolution.
If the singularity is static, the mode counting argument in the previous section becomes misleading. In this case, measuring later modes scattered off the singularity does not provide information we could not have already gained from
scattering earlier modes, unlike in the case where the singularity is evolving.
Thus, there might be a significantly smaller amount of data that is required to
describe scattering off the singularity than what the na\"ive mode counting
argument above would suggest, making static singularities less likely candidate sources of pseudorandom dynamics.

\subsection{Quantum cosmic censorship}\label{sec:QCC}

The combination of Cryptographic Censorship and our results above regarding compatibility of naked singularities with gravitational pseudorandomness suggest a particular reformulation of the Weak Cosmic Censorship Conjecture as a \textit{quantum gravity} conjecture rather than a classical gravity conjecture. This point bears emphasizing: there is no \emph{ab initio} reason that classical gravity must obey cosmic censorship. Even if there had been no known counterexamples, one of the best motivations for the Weak Cosmic Censorship Conjecture is empirical: the apparent absence of naked singularities in the night sky. Cosmic censorship is not required by consistency of known theoretical structures within General Relativity. By contrast, if naked singularities were typical in quantum gravity, then foundational results in AdS/CFT --- such as the behavior of the CFT thermal state OTOC being precisely attributed to bulk event horizon dynamics --- would be some remarkable conspiracy: either naked singularities engineer themselves to mimic black holes precisely in the OTOC signature, or black holes engineer themselves to mimic naked singularities. A similar coincidence would be required for various results on CFT thermalization and its bulk interpretation.\footnote{We thank D. Harlow for extensive discussions on this point.} We therefore find that a quantum gravity formulation of cosmic censorship is \textit{both} theoretically and empirically well-motivated; our results relating horizons to pseudorandomness provide a clear first step towards formulating weak cosmic censorship as an emergent phenomenon in terms of fundamental quantities in quantum gravity. (For example, this would require replacing the classical gravity notion of genericity --- an open set in the space of solutions to the Einstein equation --- with typicality.) 

For completeness, we first state the classical gravity proposal for weak cosmic censorship in AdS, which is now known to be violated in classical gravity. (It is also trivially violated in quantum gravity due to, e.g., evaporating black holes, although this requires violations of classical energy conditions.)  

\begin{conjecture}
    [Classical Weak Cosmic Censorship -- false] \cite{Santos2018, GerHor79} Consider a geodesically complete, asymptotically AdS initial data set satisfying the dominant\footnote{This is the weakest reasonable statement: we could assume only the null energy condition, but counterexamples exist even when we demand the dominant energy condition.}  energy condition. The maximal causal development of this initial data under the Einstein equation, together with boundary conditions at the AdS boundary,   is generically a strongly asymptotically predictable asymptotically AdS spacetime.
\end{conjecture}

Reformulating this proposal in terms of a fundamental quantum description of the theory requires a complete revamping of the statement and the scope of its applicability. Motivated by our results in this paper, the existing set of counterexamples to cosmic censorship, and the top down motivation from thermalization and chaos in AdS/CFT, we propose the following formulation of the Weak Cosmic Censorship Conjecture (below we work in the Heisenberg picture, so a given state is its entire history): 

\begin{conjecture}[Quantum Cosmic Censorship]
Typical states with an emergent geometry in quantum gravity do not have classical or semi-Planckian naked singularities.
\end{conjecture}

\section{Discussion}\label{sec:discussion} 
    
Cryptographic Censorship is the first precise and general result identifying concrete CFT behavior with bulk horizon formation. (While the \textit{vox populi} associates CFT thermalization to black hole equilibration,\footnote{Note that black hole equilibration was only recently established in four-dimensional vacuum asymptotically flat space~\cite{GioKla22} for a certain range of spins and remains an open problem in AdS.} this hypothetical link is at best established in certain highly symmetric and nearly static black holes.) Our first application of the theorem guarantees that in the microcanonical window, under assumption of pseudorandomness, states without  horizons are exponentially suppressed in probability (within their pseudorandom ensemble). 

As an immediate consequence of Cryptographic Censorship, we obtain a partial resolution to the black hole teleology problem for AdS spacetimes with a CFT dual. Event horizons can be diagnosed, using Cryptographic Censorship, in finite time. The implications are significant: first, black hole mechanics in AdS/CFT may admit a legitimate thermodynamic interpretation in highly dynamical states. 
Second, as our primary focus in this paper, we establish a typicality result for event horizons, which is highly suggestive of a form of Quantum Cosmic Censorship. The fact that na\"ive models of code subspaces containing classical (or semi-Planckian) timelike singularities are compatible with a sufficient amount of gravitational pseudorandomness serves as strong motivation for our conjecture that typical states do not have naked singularities that survive the classical limit. We emphasize that our proof of Cryptographic Censorship in fact provides a guarantee for event horizon formation whenever the CFT time evolution is exponentially hard to learn; our focus in this article was on complexity of learning as a result of pseudorandomness, but any such sufficient source of complexity would result in an equivalent theorem under the same proof. Below we will elaborate on potential applications of our results and discuss some aspirational expansions thereof. 

\subsubsection*{Towards a Converse}  Since information escapes from the black hole interior under inclusion of quantum corrections, there is no obvious definition of a quantum event horizon (since all information eventually reaches $\mathscr{I}$ and causal structures are only emergent in the classical limit). It is tempting to use our results to \textit{define} quantum event horizons --- and by extension quantum black holes (i.e., at finite $N$): states undergoing gravitationally pseudorandom time evolution with a semiclassical bulk dual in the large-$N$ limit.  
A definition, however, would require a proof of a converse to Cryptographic Censorship: a statement that CFT states with a semiclassical emergent bulk featuring a horizon typically exhibit gravitational pseudorandomness. A proof of such a converse would additionally facilitate a complete resolution of the teleology problem in AdS. The tentative link between chaos, thermalization, and pseudorandomness~\cite{Haa91, Pag93a, HayPre07, RobYos16,Cho21,WenCho22} could then potentially be used to understand the horizon thermodynamics\footnote{As opposed to the thermodynamics of the apparent horizon, which are by now well-understood~\cite{EngWal17b, EngWal18, BouCha19}.} of highly dynamical black holes and the statistical interpretation of the Hawking area theorem, which remains elusive~\cite{EngWal17a, EngFis18}. 

\subsubsection*{Gravitational Pseudorandomness from a Bulk Quantum Computer?}
The reader might wonder: what if we put a quantum computer in the bulk,
generating a large amount of pseudorandomness directly from dynamics of the bulk
fields? Must this incur a horizon? If the quantum computer does not incur violations of EFT in the causal wedge, then our theorem applies and indeed, a gravitationally pseudorandom computer must incur a horizon. If the quantum computer generates gravitational pseudorandomness in such a way that the causal wedge is no longer described by EFT --- e.g., if geometric invariants such as the volume scale inversely with $G_{N}$ as in the ``quantum volatile'' spacetimes of~\cite{EngLiu23} --- then the assumptions of our theorem are violated, namely that the $N\rightarrow \infty$ limit is described by QFT in curved space, so a horizon is not required to form. 

It is interesting that bulk processes that source a large amount of pseudorandomness (e.g., a bulk quantum computer with the power to do gravitationally pseudorandom computations) and do not violate EFT in the causal wedge must typically collapse into black holes. Prior descriptions of black hole formation typically involve energetics; to our knowledge, this is the first proof establishing a causal relation between bulk quantum computations and black hole formation (though see, e.g., \cite{KubMay23} for a discussion of constraints on bulk quantum computations behind horizons). Interestingly, this appears to bear some relation to the proposal of~\cite{AkeEng22} that effective field theory requires both an energy and a complexity cutoff: a bulk process sourcing large amounts of pseudorandomness must result in high reconstruction complexity. While pseudorandomness is itself not particularly complex to implement, reconstruction of pseudorandom quantities is by definition of high complexity class. In this case, Cryptographic Censorship may act to protect a computationally bounded asymptotic observer from detecting large violations of effective field theory due to computational complexity by forcing formation of an event horizon, depending on where the event horizon forms.

\subsubsection*{Stringy and Cosmological Applications}

It would be interesting to consider stringy effects and ask if (and how) our theorem might apply to, e.g., the D1-D5 system, or orbifold singularities, which are horizonless (see, e.g., \cite{BenGiu16}). Exact BPS states have trivial dynamics: we expect that to produce chaos and pseudorandom dynamics we would need to sufficiently excite these states. If it is indeed the case that chaotic, pseudorandom time evolution is produced as a result of heavily exciting such states, then our theorem suggests that a horizon will appear. More generally, it is not clear that horizonless geometries have any relation to chaos, and furthermore such states are probably not typical in the measure concentration sense that we are using in this work. 

Let us also briefly comment on non-AdS cosmological spacetimes. In particular, it is interesting to consider why our results fail for flat cosmologies, which do not have horizons. An obvious complaint is that our AdS/CFT arguments should not necessarily apply. They do not apply, however, in an interesting manner: the relevant feature of AdS/CFT that fails here is the lack of a notion of unitary time evolution on the conformal boundary (which in this case is null). Our results crucially rely on a fundamental description of the system which is a \textit{unitary} quantum field theory. This in all likelihood does not apply to flat cosmologies. It does, however,  apply to AdS big-crunch cosmologies, and in those, we do indeed have a horizon. 
\begin{figure}
    \centering
     \subcaptionbox{}[.4\textwidth]{\includegraphics[scale=0.9]{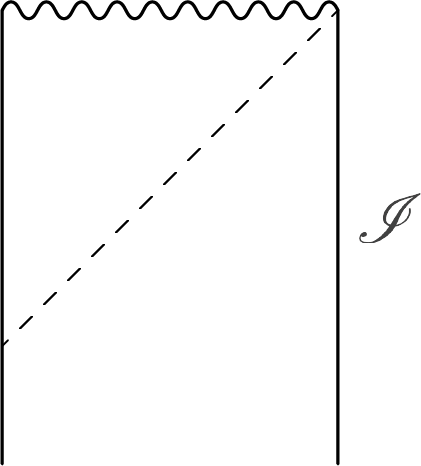}}
    \hfill
    \subcaptionbox{}[.4\textwidth]{\includegraphics[scale=0.9]{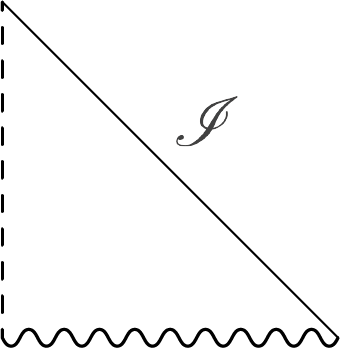}}
    \caption{The AdS-big-crunch cosmology (a) versus the flat big-bang cosmology (b). In the asymptotically flat case, there is no notion of unitary time evolution and no horizon. In the AdS case, we have both.     \label{fig:cosmologies}}
\end{figure}
See Fig.~\ref{fig:cosmologies} for a comparison. 

What about de Sitter? de Sitter cosmologies do have cosmological horizons, although those are not unambiguous in the same sense as AdS ones. There have been some suggestions (most recently~\cite{CotStr22, CotKri23}) that de Sitter time evolution --- insofar as such a notion can be defined --- is well-modeled by an isometry rather than a unitary. A generalization from pseudorandom unitaries to pseudorandom isometries was recently proposed by~\cite{AnaGul23}. It would be interesting to see if our results can be generalized to accommodate this notion and potentially have some bearing on de Sitter black holes and cosmologies; in particular, if there is any sense in which the relaxation to isometries could be mapped to the observer-dependence of de Sitter horizons. 

\subsubsection*{Some Open Questions} Cryptographic Censorship paves the way for several interesting new investigations. We have expanded on a few of those above. We briefly enumerate a few others here:
\begin{enumerate}
    \item The fact that gravitational pseudorandomness associated to a quantum computer in the bulk spacetime must result in a horizon suggests a relationship between backreaction and computational power. Does such a relation between energy and computation exist, perhaps analogous to a gravitational version of Lloyd's bound~\cite{Llo00}?

    \item Are singularities in general in fact associated with large amounts of pseudorandomness?\footnote{See \cite{KatRen23} for a unrelated attempt to connect singularities to holographic CA/CV-complexity conjectures. } This idea features prominently in our models of timelike naked singularities, but has traditionally appeared in the context of postselection models of black hole singularities and dynamics. 
   
    \item Our definition of the size of a singularity admits a middle class of singularities that disappear in the strict large-$N$ limit but which could accommodate gravitational pseudorandomness in code subspaces constructed from $O(1)$ time bands. It would be interesting to understand the significance, if any, and frequency of semi-Planckian naked singularity portions. 

    \item Explicit constructions of pseudorandom unitary ensembles to date have not been found. Still, as assumed in this work, we expect explicit examples of AdS/CFT to be well modeled by pseudorandom dynamics. Can we use complexity theoretic conjectures and aspects of AdS/CFT to demonstrate that holographic quantum systems like the Sachdev--Ye--Kitaev (SYK) model actually generate a pseudorandom unitary ensemble at long enough times? 

    \item Is there anything we can infer about code subspace-dependent cryptographic conditions/constructions necessary for pseudorandomness from the holographic manifestation of our results? That is, does the existence of subspaces without horizons (e.g., the microcanonical window around the vacuum) have any meaningful implication for cryptography? 

\end{enumerate}

\section*{Acknowledgements}
It is a pleasure to thank Chris Akers, Jordan Cotler, Jordan Docter, Roberto Emparan, Daniel Harlow, Gary Horowitz, Hsin-Yuan (Robert) Huang, Andrea Puhm, Juan Maldacena, Arvin Shahbazi-Moghaddam, Natalie Paquette, Jon Sorce, Edward Witten, and Michelle Xu for valuable discussions. This work is supported in part by the Department of Energy under Early Career Award DE-SC0021886 (NE, \AA{}F, and LY) and QuantISED DE-SC0020360 contract no. 578218 (NE), by the Heising-Simons Foundation under grant no. 2023-4430 (NE and AL), the Sloan Foundation (NE), the Templeton Foundation via the Black Hole Initiative (NE and EV), the Gordon and Betty Moore Foundation (EV), the Packard Foundation (AL), and the MIT department of physics (NE). The opinions expressed in this publication are those of the authors and do not necessarily reflect the views of the John Templeton or Moore Foundations.

\appendix 

\section{Proof of Theorem \ref{thm:learning-holography}} \label{app:proof-hol-learning} 
In this appendix, we prove Theorem \ref{thm:learning-holography}, introduced in Sec.~\ref{sec:complearnholog}. We begin with a definition of randomized $K$-Lipschitz, which we will require our algorithms to satisfy. 
\begin{defn}[Randomized $K$-Lipschitz]\label{defn:randomLipschitz}
An algorithm $\Alg$ is said to be randomized $K$-Lipschitz if it satisfies the following condition: there exists a constant $K>0$ such that 
    \begin{equation}
        \avgover{\substack{\widehat{U}_1\gets \Alg^{U_1}\\\widehat{U}_2\gets \Alg^{U_2}}} || \widehat{U}_1 - \widehat{U}_2||_\infty \leq K ||U_1 - U_2||_\infty~.
    \end{equation}
When the algorithm is deterministic, i.e., it maps a given unitary $U$ to a single output $\widehat{U}$, we will just refer to this property as $K$-Lipschitz.
\end{defn}
\noindent In addition, here we will work with operators that are exactly norm-preserving, i.e., $||\widehat{U}||_\infty = 1$. We leave it to the mathematician to confirm that this result applies equally well to approximately norm-preserving operators. 

We now repeat Theorem~\ref{thm:learning-holography} for the reader's convenience.

\learningholography*

\begin{proof}  

We prove 1. and 2. consecutively for Hilbert spaces of any finite (but sufficiently large) dimension $d$, where we always have in mind that $d$ scales with $\kappa$. Note that in the theorem above, and in the main body of the paper, we specialized to a code subspace of dimension $d = 2^\kappa$ with $\kappa = \poly (N)$; here we will leave $d$ general.

\paragraph{1. Fidelity for typical pseudorandom unitaries.}
We start by proving that the bound on the averaged fidelity in Theorem~\ref{thm:learning} also holds for \emph{typical} Haar random unitaries and states, and subsequently prove this for pseudorandom unitaries and states. To do so, we employ techniques from measure concentration. An introduction to the theory of measure concentration, including numerous references, is given in Appendix B of \cite{AkeEng22}. Here we only need the following lemma for measure concentration on the unitary group. 

\begin{lem}[Meckes (Corollary 17 in \cite{MecMec13})]\label{thm:Meckes}
Let $f: U(\unidim) \to \mathbb{R}$ be $K$-Lipschitz in the sense that $|f(U_1) - f(U_2)| \leq K ||U_1 - U_2||_2$, with $||X||_2 \equiv \sqrt{\tr X^\dagger X}$. Then in the Haar measure on $U(\unidim)$ we have the following concentration inequalities for any $\epsilon > 0$: 
    \begin{equation} \label{eq:measure-concentration}
    \begin{aligned} 
        \Pr[f \geq \braket{f}  +\epsilon] \leq \exp{\left(-\frac{\epsilon^2 \unidim}{12 K^2}\right)}~, \\
        \Pr[f \leq \braket{f}  - \epsilon] \leq \exp{\left(-\frac{\epsilon^2 \unidim}{12 K^2}\right)}~, \\
        \Pr[|f - \braket{f} | \geq \epsilon] \leq 2 \exp{\left(-\frac{\epsilon^2 \unidim}{12 K^2}\right)}~.
    \end{aligned}
    \end{equation}
\end{lem} 
For the purposes of Theorem~\ref{thm:learning-holography}, we are interested in applying this to the function $F=F(U\ket{\psi},\widehat{U}\ket{\psi})= |\braket{\psi|U^{\dagger} \widehat{U} |\psi}|^2$, where we first consider Haar random $U$ and $\ket{\psi}$, i.e., $U \gets \mu$ and $\ket{\psi} \gets \mu$, and furthermore $\widehat{U}$ is the guess that the algorithm $\Alg^U$ outputs. Thus, our task is to show that the function $F$ obeys the requirements for Lemma~\ref{thm:Meckes} to hold, such that fluctuations around the average are exponentially suppressed. That is, we need to check that $F$ is $K$-Lipschitz for some fixed $K$. 

Note that the algorithm's output $\widehat{U}$ depends on $U$. This dependence is complicated: the quantum algorithm $\Alg^U$, as a result of, e.g., intrinsic randomness in the outcome of measurements during the learning process, generally outputs a $\widehat{U}$ that is not a deterministic function of $U$. In the following, we will start by considering instead a deterministic algorithm, i.e., the algorithm is some function $\Alg(U)$ that maps every unitary $U$ to some norm-preserving operator $\widehat{U}$. 

We then prove the following lemma: 

\begin{lem}\label{lem:measure-concentration-unitaries}
Let $\Alg (U)$ denote any deterministic $K$-Lipschitz function that maps every unitary $U$ to some norm-preserving operator $\widehat{U}$. Let $\ket{\psi}$ and $\ket{\varphi}$ be any two fixed states in the Hilbert space $\Hil$. Then the function $F(U) \equiv F(U\ket{\psi}, \Alg (U)\ket{\varphi})$ is $(2 + 2K)$-Lipschitz in $U$.
\end{lem} 

\begin{proof}
Our aim is to show that 
    \begin{equation} \label{eq:lipschitz}
        |F(U_1) - F(U_2)| \leq (2+2K) ||U_1 - U_2||_2~.
    \end{equation}    
To show \eqref{eq:lipschitz}, first, notice that we can rewrite
    \begin{equation}
    \begin{aligned}
        F(U) &= F(U\ket{\psi},\Alg (U) \ket{\varphi}) = \tr\left[\Alg (U)\ket{\varphi}\bra{\varphi}\Alg (U)^\dagger U\ket{\psi}\bra{\psi}U^\dagger\right] \\
        &= \tr\left[ (\Alg (U)^\dagger U)^\dagger \, \varphi \, \Alg (U)^\dagger U \, \psi \right] \\
        &= \tr \left[\((\Alg (U)^\dagger U)^\dagger  \otimes \Alg (U)^\dagger U \)(\varphi \otimes \psi)\swap \right]~.
    \end{aligned}       
    \end{equation}
Here we used $\varphi = \ket{\varphi}\bra{\varphi}$ and $\psi = \ket{\psi}\bra{\psi}$. In the third line we introduced a second copy of the Hilbert space; $\swap$ is the swap operator. For readability, in what follows, we will abbreviate $\Alg (U_i) = \Alg_i$. Then, using H\"older's inequality, 
    \begin{equation}
    \begin{aligned}
        |F(U_1) - F(U_2)| &= \left|\tr\left[ \((\Alg_1^\dagger U_1)^\dagger  \otimes \Alg_1^\dagger U_1 - (\Alg_2^\dagger U_2)^\dagger \otimes \Alg_2^\dagger U_2\)(\varphi \otimes \psi)\swap  \right] \right| \\
        &\leq ||\varphi \otimes \psi ||_1 || \swap \((\Alg_1^\dagger U_1)^\dagger \otimes \Alg_1^\dagger U_1 - (\Alg_2^\dagger U_2)^\dagger \otimes \Alg_2^\dagger U_2\) ||_\infty~.
    \end{aligned}
    \end{equation}
Since $\varphi \otimes \psi$ is a state, its trace norm is 1. Now we can combine submultiplicativity of the operator norm with the fact that $\swap$ is unitary, such that 
    \begin{equation}
    \begin{aligned}
        || \swap \big((\Alg_1^\dagger U_1)^\dagger  \otimes \Alg_1^\dagger U_1 -& (\Alg_2^\dagger U_2)^\dagger \otimes \Alg_2^\dagger U_2\big) ||_\infty \\
        &\leq || \swap||_\infty || \big((\Alg_1^\dagger U_1)^\dagger  \otimes \Alg_1^\dagger U_1 - (\Alg_2^\dagger U_2)^\dagger \otimes \Alg_2^\dagger U_2\big) ||_\infty \\
        &= || (\Alg_1^\dagger U_1)^\dagger  \otimes \Alg_1^\dagger U_1 - (\Alg_2^\dagger U_2)^\dagger \otimes \Alg_2^\dagger U_2||_\infty~.
    \end{aligned}
    \end{equation}
Recall that we want to show that this is upper bounded by some $K' ||U_1 - U_2||_2$. Rewrite the last line as
    \begin{align} \label{eq:intermediate-point-for-PRS} \nonumber
        &|| (\Alg_1^\dagger U_1)^\dagger  \otimes \Alg_1^\dagger U_1 - (\Alg_1^\dagger U_1)^\dagger  \otimes \Alg_2^\dagger U_2 + (\Alg_1^\dagger U_1)^\dagger  \otimes \Alg_2^\dagger U_2 - (\Alg_2^\dagger U_2)^\dagger \otimes \Alg_2^\dagger U_2||_\infty \\ \nonumber
        &\quad \leq  || (\Alg_1^\dagger U_1)^\dagger  \otimes \Alg_1^\dagger U_1 - (\Alg_1^\dagger U_1)^\dagger  \otimes \Alg_2^\dagger U_2 ||_\infty + ||(\Alg_1^\dagger U_1)^\dagger  \otimes \Alg_2^\dagger U_2 - (\Alg_2^\dagger U_2)^\dagger \otimes \Alg_2^\dagger U_2||_\infty \\
        &\quad = 2 ||\Alg_1^\dagger U_1 - \Alg_2^\dagger U_2||_\infty~. 
    \end{align}
In the second line we used the triangle inequality; in the third we used {$||A\otimes B|| = ||A||\, ||B||$}, and for unitaries $||U||_\infty = 1$, and we are we are working with exactly norm-preserving operators $\widehat{U}$ such that also $||\widehat{U}||_\infty = 1$. We now again use the triangle inequality and submultiplicativity of the operator norm to rewrite the final line as
    \begin{equation}
    \begin{aligned}
        2 ||\Alg_1^\dagger U_1 - \Alg_2^\dagger U_2||_\infty &= 2 ||\Alg_1^\dagger U_1 - \Alg_2^\dagger U_1 + \Alg_2^\dagger U_1 - \Alg_2^\dagger U_2||_\infty \\ 
        &\leq 2|| \Alg_1^\dagger - \Alg_2^\dagger ||_\infty ||U_1||_\infty + 2 || \Alg_2^\dagger||_\infty || U_1 - U_2||_\infty \\ 
        &\leq (2+2K) ||U_1 - U_2||_\infty \\
        &\leq (2+2K) ||U_1 - U_2||_2~.
    \end{aligned}
    \end{equation}
In the third line we used that $\Alg (U)$ is $K$-Lipschitz for some $K$, and finally we used $||\cdot ||_\infty \leq ||\cdot ||_2$. This completes the proof that the function $F(U) \equiv F(U\ket{\psi}, \Alg (U)\ket{\varphi})$ is $(2+2K)$-Lipschitz:
    \begin{equation}
        |F(U_1) - F(U_2)| \leq (2+2K) ||U_1-U_2||_2~.
    \end{equation}
\end{proof}
Then, by Lemma \ref{thm:Meckes} and Lemma~\ref{lem:measure-concentration-unitaries}, we have the following measure concentration result: that for fraction $\geq 1 - e^{- \epsilon^2 \unidim/(48(1+K)^2)}$ of Haar random unitaries,
\begin{equation} \label{eq:measure-conc-for-Haar-unitaries}
     F(U\ket{\psi},\Alg (U)\ket{\varphi}) - \avgover{V\gets\mu} \[F(V\ket{\psi},\Alg (V
     ) \ket{\varphi}) \] < \epsilon ~.
\end{equation}
We now wish to extend this to include measure concentration over Haar random states. To this end, we prove the following lemma. 
\begin{lem}\label{lem:measure-conc-states-unitaries}
    Let $\Alg (U)$ denote any deterministic $K$-Lipschitz function that maps every unitary $U$ to some norm-preserving operator $\widehat{U}$. For fraction $\geq 1 - e^{-\epsilon^2 d/(48(1+K)^2)}$ of Haar random unitaries $U\gets\mu$ and fraction $\geq 1 - e^{-\epsilon^2 d/192}$ of Haar random states $\ket{\psi} \gets\mu$, we have
        \begin{equation}
            F(U\ket{\psi},\Alg (U) \ket{\psi}) - \avgover{\substack{V\gets\mu \\ \ket{\varphi}\gets\mu}} \[F(V\ket{\varphi},\Alg (V) \ket{\varphi}) \] < 2\epsilon ~.
        \end{equation}
\end{lem}
\begin{proof}
 Consider first fixing some $U$ and $\Alg (U)$. Let $V$ be a Haar random unitary and $\ket{\psi_0}$ some basis state, such that $V \ket{\psi_0}$ is a Haar random state. We show that the function
    \begin{equation}
        F(V) \equiv F(U V\ket{\psi_0},\Alg (U) V \ket{\psi_0}) = F(V^\dagger \Alg( U)^\dagger U V \ket{\psi_0},\ket{\psi_0})
    \end{equation}
is $4$-Lipschitz. Following the same methods as in the proof of Lemma~\ref{lem:measure-concentration-unitaries}, we find the analogue of \eqref{eq:intermediate-point-for-PRS}:
    \begin{equation}
        |F(V_1) - F(V_2)| \leq 2||V_1^\dagger \Alg (U)^\dagger UV_1 - V_2^\dagger \Alg (U)^\dagger U V_2 ||_\infty ~. 
    \end{equation}
Then, by the triangle inequality, 
    \begin{align}\nonumber
        |F(V_1) - F(V_2)| &\leq 2 ||V_1^\dagger \Alg (U)^\dagger UV_1 - V_2^\dagger \Alg (U)^\dagger UV_1 + V_2^\dagger \Alg (U)^\dagger U V_1 - V_2^\dagger \Alg (U)^\dagger U V_2 ||_\infty\\ \nonumber
        &\leq 2 ||V_1^\dagger - V_2^\dagger||_\infty ||\Alg(U)^\dagger U V_1||_\infty + 2 ||V_2^\dagger \Alg (U)^\dagger U||_\infty ||V_1 - V_2||_\infty~ \\
        &\leq 4 ||V_1 - V_2||_2~.
    \end{align}
Note that in the last line we used again that $\Alg (U)$ is exactly norm-preserving, such that $||\Alg (U)^\dagger U V_i||_\infty = 1$, as well as $||\cdot||_\infty \leq ||\cdot||_2$. 

By Lemma~\ref{thm:Meckes}, for fraction $\geq 1 - e^{-\epsilon^2 d/192}$ of $\ket{\psi} \equiv V\ket{\psi_0} \gets \mu$, 
\begin{equation} \label{eq:measure-conc-Haar-states}
    F(U\ket{\psi},\Alg (U) \ket{\psi}) - \avgover{\substack{\ket{\varphi}\gets\mu}} \[F(U\ket{\varphi},\Alg (U) \ket{\varphi}) \] < \epsilon ~.
\end{equation}
We now combine this with measure concentration for Haar random unitaries. Starting from \eqref{eq:measure-conc-for-Haar-unitaries} for two (fixed) states $\ket{\varphi}$, and then averaging over $\varphi \gets \mu$, we have for fraction $\geq 1 - e^{- \epsilon^2 \unidim/48(1+K)^2}$ of $U \gets \mu$,
    \begin{equation}
        \avgover{\ket{\varphi}\gets\mu} \[F(U\ket{\varphi},\Alg (U)\ket{\varphi}\] - \avgover{\substack{V\gets\mu \\ \ket{\varphi}\gets\mu}} \[F(V\ket{\varphi},\Alg (V) \ket{\varphi}) \] < \epsilon~.
    \end{equation}
Now add and subtract $F(U\ket{\psi},\Alg (U)\ket{\psi})$. This gives
    \begin{equation}
    \begin{aligned}
        F(U\ket{\psi}, \Alg (U) \ket{\psi}) &- \avgover{\substack{V\gets\mu \\ \ket{\varphi}\gets\mu}} \[F(V\ket{\varphi},\Alg (V) \ket{\varphi}) \] \\
        & < \epsilon +  F(U\ket{\psi}, \Alg (U) \ket{\psi}) - \avgover{\ket{\varphi}\gets\mu} \[F(U\ket{\varphi},\Alg (U)\ket{\varphi})\] ~.
    \end{aligned}
    \end{equation}
Now we can use \eqref{eq:measure-conc-Haar-states} to bound the right hand side to $2\epsilon$ for fraction $\geq 1 - e^{-\epsilon^2 d/192}$ of $\ket{\psi}\gets\mu$. This concludes the proof of Lemma~\ref{lem:measure-conc-states-unitaries}.
    
\end{proof}

To prove Theorem~\ref{thm:learning-holography}, it remains to extend the above results to (1) non-deterministic functions $\Alg(U)$, e.g., quantum algorithms $\Alg^U$, and (2) to pseudorandom unitaries and states. We first briefly address point (1). We can run the proofs of Lemma~\ref{lem:measure-concentration-unitaries} and Lemma~\ref{lem:measure-conc-states-unitaries} analogously for $\widehat{U}_i\gets \Alg^{U_i}$ until we need a smoothness property of $\widehat{U}_i$. There, instead of $K$-Lipschitz, we use the randomized $K$-Lipschitz property in Definition \ref{defn:randomLipschitz}. Furthermore, using the inequality
    \begin{equation}
        |\avg\, F(U_1) - \avg\, F(U_2)| \leq \avg \ |F (U_1) - F(U_2)|~,
    \end{equation}
we have that the function $F(V)=\avgover{\widehat{V} \gets \Alg^V}\[F(V\ket{\psi},\widehat{V}\ket{\psi})\]$ is itself $(2+2K)$-Lipschitz. This allows us to extend our measure concentration results to this function $F(V)$ where previously we proved them for $F(U)=F(U\ket{\psi},\widehat{U}\ket{\psi})$. 

Then we arrive at the following result: that for fraction $\geq 1 - e^{-\epsilon^2 d/(48(1+K)^2)}$ of $U \gets \mu$, and for fraction $\geq 1 - e^{-\epsilon^2 d/192}$ of $\ket{\psi} \gets \mu$ (alternatively, we can think of this as $\geq 1 - e^{-\epsilon^2 d/192} - e^{-\epsilon^2 d/(48(1+K)^2)}$ fraction of $U\gets \mu$ and $\ket{\psi}\gets\mu$), we have the following concentration bound: 
    \begin{equation} \label{eq:measure-concentration-Haar}
        \avgover{\substack{\widehat{U} \gets \Alg^U}}  \[ F(U\ket{\psi} , \widehat{U} \ket{\psi}) \] - \avgover{\substack{V\gets\mu\\\ket{\varphi}\gets\mu}} \[ \avgover{\widehat{V} \gets \Alg^V} \[F(V\ket{\varphi},\widehat{V}\ket{\varphi})\]\] < 2\epsilon ~.
    \end{equation}

Finally, we address point (2) above, and show that to preserve indistinguishability, the above result must also hold for \emph{pseudorandom} unitaries and states. Intuitively, if the fluctuations around the Haar random averaged fidelity are exponentially suppressed, but the fluctuations around the pseudorandom averaged fidelity are not, then one could measure these fluctuations and thereby distinguish pseudorandom unitaries and states from Haar random ones, breaking the pseudorandomness. 

We proceed to formally show this below. 
\begin{lem}\label{lem:measure-conc-PR}
Let $\kappa = \poly(N)$ be the security parameter. Let $\PRUens$ be any pseudorandom unitary (PRU) ensemble and let $\PRSens$ be any pseudorandom state (PRS) ensemble. Let $\widehat{U}$ be the output of the learning algorithm $\Alg^U$, which we assume to be (randomized) K-Lipschitz, and where $U\gets\PRUens$. Then the following concentration inequality holds, in the measure induced on $\PRUens$ and $\PRSens$, for any $\epsilon > 0$, with probability $\geq 1- \deltaMCU - \deltaMCS  - \negl' (N)$ over $U\gets\PRUens$ and $\ket{\psi} \gets \PRSens$,
\begin{equation}
\begin{aligned}
    \avgover{\substack{\widehat{U} \gets \Alg^U}} F\big(U\ket{\psi},\widehat{U}\ket{\psi}\big) - \avgover{\substack{V\gets\PRUens \\ \ket{\varphi}\gets \Psi \\
    \widehat{V} \gets \Alg^V}} \big[F\big(V\ket{\varphi},\widehat{V} \ket{\varphi}\big)\big] < 2\epsilon + \negl(N)~.
\end{aligned}
\end{equation} 
\end{lem} 

\begin{proof} To prove Lemma~\ref{lem:measure-conc-PR}, we first extend \eqref{eq:measure-concentration-Haar} to pseudorandom unitaries, and then to pseudorandom states. 

First, let $\PRSens$ be any ensemble of states. For ease of notation, we write
\begin{equation}
\beta := \avgover{\substack{U\gets\mu \\ \ket{\psi}\gets\Psi \\ \widehat{U} \gets \Alg^U}} \[F(U\ket{\psi},\widehat{U}\ket{\psi}) \]
\end{equation}
and
\begin{equation}
\beta' := \avgover{\substack{U\gets\PRUens \\ \ket{\psi}\gets\Psi \\ \widehat{U} \gets \Alg^U}} \[F(U\ket{\psi},\widehat{U}\ket{\psi})\]~.
\end{equation}
By Lemma~7 of~\cite{YanEng23}, $\left|\beta - \beta' \right| \leq \eta(\HN)$ where $\eta$ is a negligible function. 

Assume towards contradiction that there exists a pseudorandom unitary ensemble $\PRUens $, any distribution of states $\PRSens$, and a positive polynomial $p$ such that
    \begin{equation} \label{eq:PR-assump-to-contradict}
        \Pr_{\substack{U\gets\PRUens \\\ket{\psi}\gets\Psi}}\[\avgover{\substack{\widehat{U} \gets \Alg^U}}F\(U\ket{\psi},\widehat{U}\ket{\psi}\) - \beta'  \geq \epsilon + \eta (N) ~\] \geq  \deltaMCU + \frac{1}{p(N)}~.
    \end{equation} 
If the above equation holds, then there is a state $\ket{\psi_0}$ such that this equation holds with $\ket{\psi}:=\ket{\psi_0}$:
    \begin{equation}
        \Pr_{\substack{U\gets\PRUens}}\[\avgover{\substack{\widehat{U} \gets \Alg^U}}F\(U\ket{\psi_0},\widehat{U}\ket{\psi_0}\) - \beta'  \geq \epsilon + \eta (N) ~\] \geq  \deltaMCU  + \frac{1}{p(N)}~.
    \end{equation} 
We now construct an efficient algorithm $\Dist$ that uses oracle access to the pseudorandom unitary $U$ and also can run the learning algorithm $\Alg^U$.\footnote{Note that the distinguisher can depend on $\HN$, as ours does here, since e.g.\ $\ket{\psi_0}$ depends on $\HN$.} We show that if our assumption \eqref{eq:PR-assump-to-contradict} is true, then $\Dist$ can distinguish between $U\gets\mu$ and $U\gets\PRUens$, thus breaking the pseudorandomness of $\PRUens$.

The distinguisher $\Dist$ gets $\poly(N)$ copies of $\ket{\psi_0}$, and is also handed $\beta$.\footnote{The careful reader might worry that handing the distinguisher (multiple copies of) \emph{any} state in $\Hil$, or $\beta$, breaks the computational boundedness of the algorithm. Note, however, that although the algorithms we consider are limited to a polynomial number of
queries in $\log\dim\Hil$, each query can be for any state in $\Hil$ that the algorithm a) can prepare or b) has been given as ``advice". This advice can also include quantities like $\beta$. Note that the learning algorithm we consider in, e.g, Theorem~\ref{thm:learning-holography} gets only \emph{one} copy of the physical state $\ket{\psi}$; the distinguisher that we describe here is a different algorithm and can receive multiple copies as advice.} Then it proceeds as follows:
\begin{enumerate} 
\item $\Dist$ runs $\Alg$ and simulates the oracle for $\Alg$: it receives each query from $\Alg$, sends the query to its oracle $U$, and returns the response to $\Alg$.
\item $\Dist$ gives $\ket{\psi_0}$ to $\Alg$ and receives back $\widehat{U}\ket{\psi_0}$. It also gives $\ket{\psi_0}$ to its oracle $U$ and gets back $U\ket{\psi_0}$. 
\item $\Dist$ repeats this $\poly(N)$ times and then 
uses a quantum algorithm (e.g., \cite{WanZha21}) to efficiently estimate $\beta'':= F(U\ket{\psi_0},\widehat{U}\ket{\psi_0})$ up to some additive error $\hat{\epsilon}>0$, i.e., $|\beta''_{\rm est} - \beta''| \leq \hat{\epsilon}$.
\item If $\beta''_{\rm est} -\beta \geq \te$, where $\te > \hat{\epsilon}$, then $\Dist$ outputs $1$. 
\end{enumerate}  

Note that $\Dist$ runs in $\poly(\HN)$-time since it can efficiently query the oracle, and furthermore the quantum algorithm for fidelity estimation is efficient.\footnote{For example, the quantum algorithm of \cite{WanZha21} runs in $\poly(1/\hat{\epsilon})$-time.}

We now show that $\Dist$ distinguishes between $U\gets \mu$ and $U\gets \PRUens$. For Haar random $U \gets \mu$ as the oracle, 
    \begin{equation}
    \begin{aligned}
        \Pr_{U \gets \mu} \[ \Dist = 1\] &= \Pr \[\beta''_{\rm est} - \beta \geq \te \] = \Pr \[ \beta'' - \beta \geq \te - (\beta''_{\rm est} - \beta'') \] \\
        &\leq \Pr \[\beta'' - \beta \geq \te - \hat{\epsilon} \]\\
        &\leq  e^{-(\te - \hat{\epsilon})^2 d/48(1+K)^2} ~,
    \end{aligned}
    \end{equation}
where we used $- \hat{\epsilon} \leq \beta''_{\rm est} - \beta'' \leq \hat{\epsilon}$ in the second line, and \eqref{eq:measure-conc-for-Haar-unitaries} (concentration in the Haar measure) for $\epsilon = \te - \hat{\epsilon} >0$ in the last line. 

For pseudorandom $U \gets \PRUens$ as the oracle, 
    \begin{equation}
    \begin{aligned}
        \Pr_{U \gets \PRUens} \[ \Dist = 1\] &= \Pr \[\beta''_{\rm est} - \beta \geq \te \] = \Pr \[ \beta'' - \beta' \geq \te - (\beta''_{\rm est} - \beta'') -(\beta'-\beta) \]  \\
        & \geq  \Pr \[\beta'' - \beta' \geq \te + \hat{\epsilon} + \eta(N) \] \\
        &\geq  e^{-(\te + \hat{\epsilon})^2 d/48(1+K)^2} + \frac{1}{p(N)}~,
    \end{aligned}
    \end{equation}
where we used $-\eta (N) \leq \beta - \beta' \leq \eta (N)$ in the third line, and our assumption \eqref{eq:PR-assump-to-contradict} in the last line for $\epsilon = \te + \hat{\epsilon} $. 

Therefore, 
    \begin{equation}
    \begin{aligned}
        \Pr_{U \gets \PRUens}[\Dist = 1] - \Pr_{U\gets\mu}[\Dist = 1]  \geq \frac{1}{p(N)} +  e^{-\frac{(\te + \hat{\epsilon})^2 d}{48(1+K)^2}}  - e^{-\frac{(\te - \hat{\epsilon})^2 d}{48(1+K)^2}}  ~.
    \end{aligned}
    \end{equation}
This shows that $\Dist$ distinguishes between $U \gets \mu$ and $U \gets \PRUens$. Since $\Dist$ runs in $\poly (N)$ time, this contradicts the pseudorandomness of $\PRUens$ (Definition~\ref{defn:PRU}). Thus based on the pseudorandomness of $\PRUens$, we must have that for fraction $\leq \deltaMCU  + \eta'(N)$ of $U\gets\PRUens$
    \begin{equation} \label{eq:measure-conc-PRU}
      \avgover{\substack{\widehat{U} \gets \Alg^U}}F\(U\ket{\psi},\widehat{U}\ket{\psi}\) - \avgover{\substack{V\gets\PRUens \\ \ket{\varphi}\gets \Psi \\
    \widehat{V} \gets \Alg^V}} F\(V\ket{\varphi},\widehat{V}\ket{\varphi}\) \geq \epsilon  + \eta(N) ~,
    \end{equation}
where $\eta'(N)$ is negligible. In particular, this holds when $\ket{\psi}\gets\PRSens$ is Haar random. 

Thus, we have proven a measure concentration result for fluctuations around the average fidelity where we average over pseudorandom unitaries (and Haar random states). An analogous proof holds for fluctuations around the average fidelity where we average over pseudorandom states (and any distribution of unitaries $\mathcal{D}$; in particular, it holds for Haar random unitaries $U\gets \mu$). It requires Lemma 8 in \cite{YanEng23}. This would result in the following measure concentration result: that for fraction $\leq \deltaMCS  + \tilde{\eta}'(N)$ of $\ket{\psi} \gets \PRSens$,
    \begin{equation} \label{eq:measure-conc-PRS}
       \avgover{\substack{\widehat{U} \gets \Alg^U}}F\(U\ket{\psi},\widehat{U}\ket{\psi}\) - \avgover{\substack{V\gets\PRUens \\ \ket{\varphi}\gets \Psi \\
    \widehat{V} \gets \Alg^V}} \[F\(V\ket{\varphi},\widehat{V}\ket{\varphi}\)\]\geq \epsilon  + \tilde{\eta}(N) ~,
    \end{equation}
where we explicated that the negligible functions need not be the same as in \eqref{eq:measure-conc-PRU}. Combining \eqref{eq:measure-conc-PRU} and \eqref{eq:measure-conc-PRS} concludes the proof of Lemma~\ref{lem:measure-conc-PR}. 
\end{proof}

We will now use Lemma~\ref{lem:measure-conc-PR} to prove the first part of Theorem~\ref{thm:learning-holography}.
Consider any constant $\epsilon >0 $. Consider any sufficiently large value of $\HN\in\Nat$ so that Theorem~\ref{thm:learning} from~\cite{YanEng23} holds. Let $\PRU \gets\PRUens $ and $\ket{\PRS }\gets\PRSens $ denote the pseudorandom unitary and state. Let $\widehat{\PRU}$ denote the operator that $\Alg^{\PRU}$ produces. To connect to the learning result (Theorem~\ref{thm:learning}), we average over the output of the algorithm, $\widehat{U}\ket{\psi} \gets \Alg^{U}(\ket{\psi})$. We abbreviate this by $\widehat{U}\gets \Alg^{U}$. By Lemma~\ref{lem:measure-conc-PR}, we have that with probability $\geq 1-\deltaMCU - \deltaMCS - \negl'(N)$, 
    \begin{equation}\label{eqn:typical-finite-N} 
    \begin{aligned}
         \avgover{\widehat{U}\gets \Alg^{U} } \[ F\big(\PRU\ket{\PRS},\widehat{\PRU}\ket{\PRS}\big) \]  - \avgover{\substack{V\gets\PRUens\\ \ket{\varphi}\gets\PRSens\\ \widehat{V}\gets\Alg^{V}}}\[F\big(V\ket{\varphi},\widehat{V}\ket{\varphi}\big)\] < 2\epsilon + \negl(N) ~.
    \end{aligned}
    \end{equation} 
We call $(\PRU,\ket{\PRS})$ \emph{typical} if the above equation holds. The complexity of learning result (Theorem~\ref{thm:learning}) gives us an upper bound on the average fidelity. Thus we have that for any typical $(\PRU,\ket{\PRS})$,
    \begin{equation}\label{eqn:typical-fidelity-finite-N} 
        \avgover{\widehat{U}\gets \Alg^{U} } \[ F\(\PRU\ket{\PRS}, \widehat{\PRU}\ket{\PRS}\) \] \leq 1-\const+2\epsilon + \negl(N)~.
    \end{equation}
This shows the first part of the theorem (Equation~\eqref{eqn:theorem-typical-fidelity}).  \\

\noindent \emph{Haar random unitaries and states.} ~ Note that we obtain a stronger result for Haar random unitaries and states: namely, that for fraction $\geq 1-\deltaMCU - \deltaMCS$ of $U\gets\mu$ and $\ket{\psi}\gets\mu$, 
    \begin{equation}\label{eqn:typical-fidelity-finite-N-Haar} 
        \avgover{\widehat{U}\gets \Alg^{U} } \[ F\(\PRU\ket{\PRS}, \widehat{\PRU}\ket{\PRS}\) \] \leq 1-\const+2\epsilon~.
    \end{equation}
Note that for $d = 2^\kappa$ with $\kappa = \poly (N)$, this means that the above bound is violated for only a fraction that is doubly-exponentially suppressed in $N$.

\paragraph{2. Distinguishing operator.} We now show the second part of the
theorem. The relationship between fidelity and distinguishing measurements is a known result; for any finite dimensional Hilbert space, this is given in Appendix B in \cite{YanEng23}. Thus, for any typical $(\PRU,\ket{\PRS})$, there exists an operator $\operator$ (for a given  $\widehat{U}$) for which
    \begin{equation}
        \avgover{\widehat{U}\gets \Alg^{U} } \left|\tr(\operator \widehat{\PRU} \PRS \widehat{\PRU}^\dagger)-\tr(\operator\PRU \PRS \PRU^\dagger)\right|\geq 2\const-4\epsilon  - \negl(N)  ~.     
    \end{equation}
This concludes the proof of Theorem~\ref{thm:learning-holography}. 

\end{proof}

\section{Proof of Corollary~\ref{cor:large-N-limit} }\label{app:proof-large-N}
In this appendix, we prove Corollary~\ref{cor:large-N-limit}, the existence of a semiclassical distinguishing operator. We repeat Corollary~\ref{cor:large-N-limit} here for convenience:

\largeNlimitredo*

\begin{proof} 
By the assumptions on the sequences $\{\PRU_N\ket{\PRS_N}\}_\HN$ and $\{\widehat{\PRU}_N\ket{\PRS_N}\}_\HN$, for every sufficiently large $\HN$, $\PRU_N\ket{\PRS_N}$ and $\widehat{\PRU}_N\ket{\PRS_N}$ satisfy 1. and 2. in Theorem~\ref{thm:learning-holography}.

The result \eqref{eqn:typical-fidelity-finite-N} by definition guarantees that 
\begin{align}\label{eqn:learningbound}
  \avgover{\hat{\xi}}\ \left|\braket{\xi|V_N^{\dagger} V_N| \hat{\xi}}\right| \leq 1- \alpha - \negl(N).
\end{align}
Taking the limit of both sides and using the definition of an asymptotically isometric code in Assumption \ref{ass:isometry}, we see that the two bulk states $\ket{\xi}$ and $\ket{\hat{\xi}}$ obey the inequality as well. Note that to commute the average over $\hat{\xi}$ with the $N \to \infty$ limit we had to use the dominated convergence theorem.

By the Fuchs--van de Graaf inequality, the overlap of the two states $\braket{\xi|\hat{\xi}}$ upper and lower bounds their one-norm distance:
\begin{align}
    1-\sqrt{F(\xi, \hat{\xi})} \leq \frac{1}{2} || \xi - \hat{\xi}||_1 \leq \sqrt{1-F(\xi, \hat{\xi})}~.
\end{align}
When $\xi$ is a pure state, the lower bound can be strengthened:
\begin{align}
\frac{1}{2} ||\xi-\hat{\xi}||_1 \geq 1-F(\xi,\hat{\xi})~.
\end{align} 
Averaging this inequality over $\hat{\xi}$ and using \eqref{eqn:learningbound}, we see that 
\begin{align}\label{eqn:avgonenormbound}
   \avgover{\hat{\xi}} \ ||\xi - \hat{\xi}||_1 \geq 2\alpha.
\end{align}
The one-norm distance between two states on a general algebra $\mathcal{B}(\mathcal{H})$ may be written as 
\begin{equation}\label{eqn:onenorm}
\begin{aligned}
    ||\xi-\hat{\xi}||_{1} = \sup_{a \in \mathcal{B}(\mathcal{H}): ||a||_{\infty}\leq
    1}|\tr(a \xi) - \tr(a \hat{\xi})|.
\end{aligned}
\end{equation}
where the operators in the supremum are restricted to be self-adjoint.
Combining \eqref{eqn:onenorm} and \eqref{eqn:avgonenormbound}, we see that there exists some operator $\operator \in \mathcal{B}(\mathcal{H}),$ depending on $\widehat{U}$ and therefore $\hat{\xi}$, such that 
\begin{equation}
\begin{aligned}
    \avgover{\widehat{\xi}}\big|\tr(\operator\widehat{\xi}) - \tr(\operator\xi)\big| \geq 2\alpha~.
\end{aligned}
\end{equation} 
In particular, we get that there exists an operator whose expectation values for $\hat{\xi}$ and $\xi$ must differ by a constant. This concludes the proof of Corollary~\ref{cor:large-N-limit}.
\end{proof}

\bibliographystyle{jhep}
\bibliography{all}

\end{document}